\documentclass[3p,11pt]{elsarticle}

\usepackage{hyperref}
\usepackage{latexsym}
\usepackage{xspace}
\usepackage{amssymb}
\usepackage{amsmath}
\usepackage{stmaryrd}
\usepackage[all]{xy}
\usepackage{calc}
\usepackage{tikz}
\usepackage{color}
\usepackage{enumerate}
\usepackage[ruled,vlined,linesnumbered]{algorithm2e}

\newtheorem{theorem}{Theorem}[section]
\newtheorem{lemma}[theorem]{Lemma}

\newtheorem{corollary}[theorem]{Corollary}

\newtheorem{Definition}[theorem]{Definition}
\newtheorem{Example}[theorem]{Example}
\newtheorem{Remark}[theorem]{Remark}

\newenvironment{example}{\begin{Example}\begin{em}}{\end{em}\end{Example}}
\newenvironment{remark}{\begin{Remark}\begin{em}}{\end{em}\end{Remark}}
\newproof{proof}{Proof}

\SetKwInput{Input}{Input}
\SetKwInput{Output}{Output}
\SetKwInput{Purpose}{Purpose}
\SetKw{Break}{break}
\SetKw{Continue}{continue}
\SetKw{Return}{return}
\SetKw{New}{new\,}

\newcommand{\red}[1]{#1\xspace}

\newcommand{\markRed}{}

\def\eqref#1{(\ref{#1})}
\def\defeq{\stackrel{\mathrm{def}}{=}}
\def\tuple#1{\langle#1\rangle}

\newcommand{\mand}{\sqcap}
\newcommand{\mor}{\sqcup}
\newcommand{\V}{\forall}
\newcommand{\E}{\exists}

\newcommand{\mL}{\mathcal{L}}
\newcommand{\mLP}{\mathcal{L}_\Phi}
\newcommand{\mLPp}{\mathcal{L}_\Phi^0}

\newcommand{\mA}{\mathcal{A}}
\newcommand{\mI}{\mathcal{I}}
\newcommand{\mIp}{{\mathcal{I}'\!}}
\newcommand{\mIdp}{{\mathcal{I}''\!}}

\newcommand{\ALCreg}{$\mathcal{ALC}_{reg}$\xspace}

\newcommand{\CN}{\mathit{CN}}
\newcommand{\RN}{\mathit{RN}}
\newcommand{\IN}{\mathit{IN}}

\newcommand{\Self}{\mathtt{Self}}

\newcommand{\myend}{\mbox{}\hfill{\small$\Box$}}

\newcommand{\fand}{\varotimes}

\newcommand{\fto}{\Rightarrow}
\newcommand{\fequiv}{\Leftrightarrow}

\newcommand{\cnv}[1]{{#1}^{-1}}
\newcommand{\mF}{\mathcal{F}}

\newcommand{\elements}{\mathit{elements}}
\newcommand{\degree}{\mathit{degree}}
\newcommand{\subblocks}{\mathit{subblocks}}
\newcommand{\blocks}{\mathit{blocks}}
\newcommand{\parent}{\mathit{parent}}
\newcommand{\repr}{\mathit{repr}}
\newcommand{\Null}{\mathit{null}}
\newcommand{\findBlock}{\textit{find\_block}}

\newcommand{\FLG}{FLG\xspace}

\newcommand{\SV}{\Sigma_V}
\newcommand{\SE}{\Sigma_E}
\newcommand{\support}{\mathit{support}}

\newcommand{\bB}{\mathbb{B}}
\newcommand{\MinFInt}{\mbox{$\mathsf{ApproximateMinimization}$}\xspace}

\newcommand{\fALC}{\mbox{$f\!\mathcal{ALC}$}\xspace}
\newcommand{\fALCP}{\mbox{$f\!\mathcal{ALC}_\Phi$}\xspace}
\newcommand{\fALCreg}{\mbox{$f\!\mathcal{ALC}_{reg}$}\xspace}
\newcommand{\fALCregD}{\mbox{$f\!\mathcal{ALC}_{reg}^\triangle$}\xspace}
\newcommand{\fALCregU}{\mbox{$f\!\mathcal{ALC}_{reg}^U$}\xspace}

\newcommand{\mLPD}{\mathcal{L}_\Phi^\triangle}
\newcommand{\mLPU}{\mathcal{L}_\Phi^U}

\newcommand{\FfA}{FfA\xspace}
\newcommand{\FfAs}{FfAs\xspace}

\journal{arXiv}

\begin{document}
\sloppy
	
\begin{frontmatter}		
		
\title{Approximate minimization of interpretations in fuzzy description logics under the G\"odel semantics}

\author{Linh Anh Nguyen}
\ead{nguyen@mimuw.edu.pl}
\ead{nalinh@ntt.edu.vn}				
\address{Institute of Informatics, University of Warsaw, Banacha 2, 02-097 Warsaw, Poland}
\address{Faculty of Information Technology, Nguyen Tat Thanh University, Ho Chi Minh City, Vietnam}

\begin{abstract}
The problem of minimizing fuzzy interpretations in fuzzy description logics (FDLs) is important both theoretically and practically. For instance, fuzzy or weighted social networks can be modeled as fuzzy interpretations, where individuals represent actors and roles capture interactions. Minimizing such interpretations yields more compact representations, which can significantly improve the efficiency of reasoning and analysis tasks in knowledge-based systems.
We present the {\em first} algorithm that minimizes a finite fuzzy interpretation while preserving fuzzy concept assertions in FDLs without the Baaz projection operator and the universal role, under the G\"odel semantics. The considered class of FDLs ranges from the sublogic of $\fALC$ without the union operator and universal restriction to the FDL that extends $\fALCreg$ with inverse roles and nominals. 
Our algorithm is given in an extended form that supports approximate preservation: it minimizes a finite fuzzy interpretation $\mI$ while preserving fuzzy concept assertions up to a degree~$\gamma \in (0,1]$. 
Its time complexity is \mbox{$O((m\log{l} + n)\log{n})$}, where $n$ is the size of the domain of~$\mI$, $m$ is the number of nonzero instances of atomic roles in~$\mI$, and $l$ is the number of distinct fuzzy values used in such instances plus~2. 
\red{Methodologically, our approach fundamentally differs from existing ones, as it avoids quotient constructions traditionally employed for minimizing fuzzy interpretations and fuzzy automata.}
\end{abstract}


\end{frontmatter}


\section{Introduction}
\label{section:intro}

Minimization is a fundamental problem across many areas of computer science (see, e.g., \cite{KamedaW70,Hopcroft71,DBLP:journals/tcs/CardonC82,PaigeT87,Jiang1993,IlieY03,DBLP:journals/socnet/Ziberna07,Doreian2009,BSDL-INS,BIANCHINI2024114621}). In automata theory, minimizing finite automata is essential for optimizing computational resources and simplifying verification tasks. Similarly, in concurrency theory and formal methods, the minimization of labeled transition systems and Kripke models supports model checking and system equivalence analysis. 
The general goal 
is to find a simpler structure that preserves the essential semantic or behavioral properties of the original model.

Description logics (DLs) are a family of logic-based formalisms designed to represent and reason about structured knowledge~\cite{dlbook}, particularly in the context of ontologies. They form the logical foundation of the Web Ontology Language (OWL) and support key reasoning tasks such as concept subsumption, instance checking, and consistency verification. To model uncertainty and vagueness in real-world domains, various fuzzy extensions of DLs have been proposed~\cite{Straccia98,BobilloCEGPS2015,BorgwardtP17b,ai/BorgwardtDP15}. Fuzzy description logics (FDLs) allow concepts and roles to be interpreted in terms of degrees of truth.

The problem of minimizing fuzzy interpretations in FDLs is important both theoretically and practically. For instance, fuzzy or weighted social networks can be modeled as fuzzy interpretations, where individuals represent actors and roles capture interactions. Minimizing such interpretations yields more compact representations, which can significantly improve the efficiency of reasoning and analysis tasks in knowledge-based systems.
However, while the problem of minimizing or reducing finite fuzzy automata has been widely studied \cite{DBLP:journals/fss/YangL20,BASAK2002223,Belohlvek2009OnAM,DEMENDIVILGRAU2024109108,Halamish2015,LITFS2015,LI20071423,ShamsizadehZG24,SCB.18,SCI.14,StanimirovicMC.22}, the problem of minimizing finite fuzzy interpretations in FDLs has received relatively little attention.\footnote{An extended discussion on the related literature is given in Section~\ref{section: related work}.} Only the works \cite{minimization-by-fBS,DBLP:journals/fss/Nguyen24} have addressed this problem directly, under restrictive assumptions. The following subsection discuss them using a concrete example, which provides motivation for the current work. 

\subsection{Motivating example}
\label{sec: motivation}

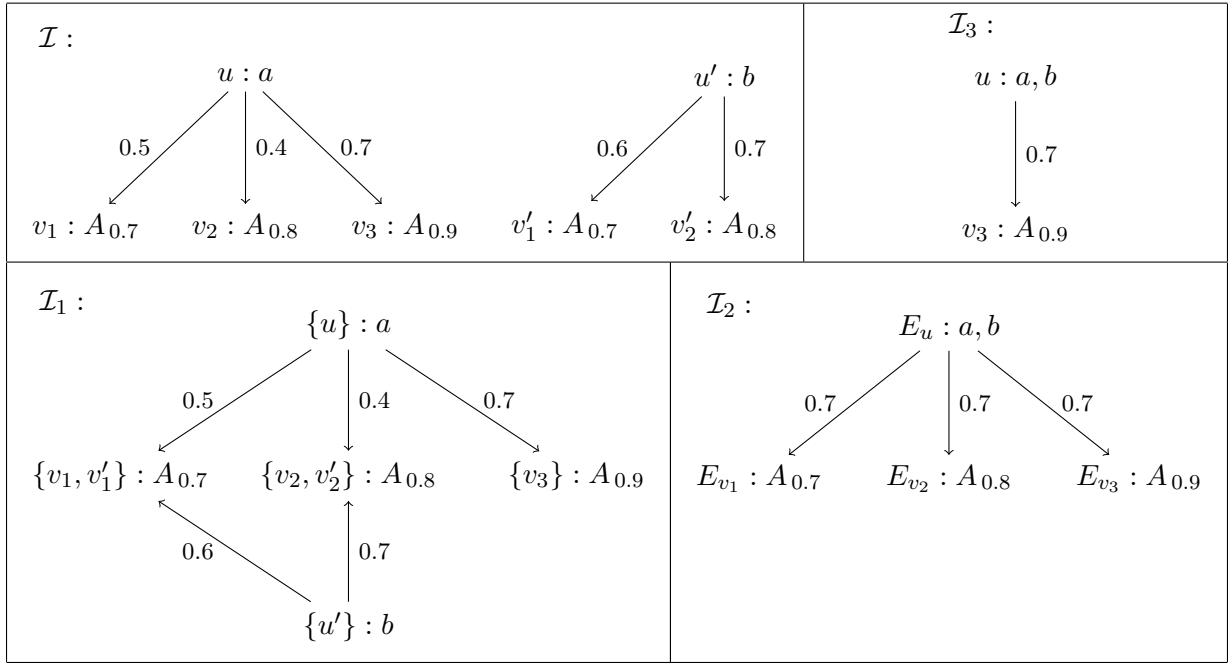
\begin{figure}[t]
\begin{center}
\begin{tabular}{|c|c|c|}
\hline
\multicolumn{2}{|c|}{
    \begin{tikzpicture}
	\node (po) {};
	\node (u) [node distance=0.7cm, below of=po] {$u: a$};
	\node (ub) [node distance=2.0cm, below of=u] {$v_2:A_{\,0.8}$};
	\node (v) [node distance=2.1cm, left of=ub] {$v_1:A_{\,0.7}$};
	\node (w) [node distance=2.1cm, right of=ub] {$v_3:A_{\,0.9}$};
    \node (I) [node distance=2.5cm, above of=v] {$\mI:\quad\quad$};
	\draw[->] (u) to node [left, xshift=-3pt]{\footnotesize{0.5}} (v);	
	\draw[->] (u) to node [right]{\footnotesize{0.4}} (ub);
	\draw[->] (u) to node [right, xshift=3pt]{\footnotesize{0.7}} (w);
	\node (vp) [node distance=2.1cm, right of=w] {$v'_1:A_{\,0.7}$};
	\node (upb) [node distance=2.1cm, right of=vp] {$v'_2:A_{\,0.8}$};
	\node (up) [node distance=2.0cm, above of=upb] {$u': b$};
	\draw[->] (up) to node [left, xshift=-3pt]{\footnotesize{0.6}} (vp);
	\draw[->] (up) to node [right]{\footnotesize{0.7}} (upb);
    \end{tikzpicture}
} 
& 
    \begin{tikzpicture}
	\node (po) {$\mI_3:\qquad\quad$};
	\node (u) [node distance=0.7cm, below of=po] {$u: a, b$};
	\node (ub) [node distance=2.0cm, below of=u] {$v_3:A_{\,0.9}$};
	\draw[->] (u) to node [right]{\footnotesize{0.7}} (ub);
    \end{tikzpicture}
\\
\hline
    \begin{tikzpicture}
	\node (po) {};
	\node (u) [node distance=0.7cm, below of=po] {$\{u\}: a$};
	\node (ub) [node distance=2.0cm, below of=u] {$\{v_2,v'_2\}:A_{\,0.8}$};
	\node (v) [node distance=3.0cm, left of=ub] {$\{v_1,v'_1\}:A_{\,0.7}$};
	\node (w) [node distance=3.0cm, right of=ub] {$\{v_3\}:A_{\,0.9}$};
    \node (I) [node distance=2.3cm, above of=v] {$\mI_1:\qquad\quad\quad$};
	\draw[->] (u) to node [left, xshift=-4pt]{\footnotesize{0.5}} (v);	
	\draw[->] (u) to node [right]{\footnotesize{0.4}} (ub);
	\draw[->] (u) to node [right, xshift=4pt]{\footnotesize{0.7}} (w);
	\node (up) [node distance=2.0cm, below of=ub] {$\{u'\}: b$};
	\draw[->] (up) to node [left, xshift=-4pt]{\footnotesize{0.6}} (v);
	\draw[->] (up) to node [right]{\footnotesize{0.7}} (ub);
    \end{tikzpicture}
&
\multicolumn{2}{|c|}{
    \begin{tikzpicture}
	\node (po) {};
	\node (u) [node distance=0.7cm, below of=po] {$E_u: a, b$};
	\node (ub) [node distance=2.0cm, below of=u] {$E_{v_2}:A_{\,0.8}$};
	\node (v) [node distance=2.5cm, left of=ub] {$E_{v_1}:A_{\,0.7}$};
	\node (w) [node distance=2.5cm, right of=ub] {$E_{v_3}:A_{\,0.9}$};
    \node (I) [node distance=2.3cm, above of=v] {$\mI_2:\quad\quad$};
	\draw
	(u) edge[->,left] node[xshift=-3pt]{\footnotesize{0.7}} (v)	
	(u) edge[->,right] node{\footnotesize{0.7}} (ub)
	(u) edge[->,right] node[xshift=3pt]{\footnotesize{0.7}} (w);
	\node (up) [node distance=2.16cm, below of=ub] {};
    \end{tikzpicture}
}    
\\ 
\hline
\end{tabular}
\caption{Fuzzy interpretations used in Section~\ref{sec: motivation}.\label{fig: JHDKJ}}
\end{center}
\end{figure}

Let the signature consist of a concept name $A$, a role name $r$ and individual names $a$ and $b$. Consider the fuzzy interpretation $\mI$ illustrated in Figure~\ref{fig: JHDKJ} and specified below (for notation and formal definitions, we refer the reader to Section~\ref{section: prel}):
\begin{itemize}
\item $\Delta^\mI = \{u,u',v_1,v_2,v_3,v'_1,v'_2\}$,\ $a^\mI = u$,\ $b^\mI = u'$, 
\item $A^\mI = \{v_1\!:\!0.7,\ v_2\!:\!0.8,\ v_3\!:\!0.9,\ v'_1\!:\!0.7,\ v'_2\!:\!0.8\}$, 
\item $r^\mI = \{\tuple{u,v_1}\!:\!0.5,\ \tuple{u,v_2}\!:\!0.4,\ \tuple{u,v_3}\!:\!0.7,\ \tuple{u',v'_1}\!:\!0.6,\ \tuple{u',v'_2}\!:\!0.7 \}$.
\end{itemize}
It extends the one from~\cite[Example~5]{minimization-by-fBS} by dealing with the additional individual name~$b$.

Let \fALCreg denote the fuzzy extension of the DL \ALCreg, defined as $\mLP$ with $\Phi = \emptyset$ in~\cite{minimization-by-fBS} and Section~\ref{section: prel}. Let \fALCregD denote the FDL that extends \fALCreg with the Baaz projection operator~$\triangle$; it is the same as \fALCP with $\Phi = \{\triangle, \circ, \mor_r, *, ?\}$ defined in~\cite{DBLP:journals/fss/Nguyen24}. 
Furthermore, let \fALCregU denote the FDL that extends \fALCreg with the universal role, defined as $\mLP$ with $\Phi = \{U\}$ in~\cite{minimization-by-fBS}. 
Denote $\Phi_0 = \emptyset$, $\Phi_1 = \{\triangle, \circ, \mor_r, *, ?\}$, $\Phi'_1 = \{\triangle\}$ and $\Phi_2 = \{U\}$.

The greatest crisp $\Phi_1$-auto-bisimulation of~$\mI$, defined as in~\cite{DBLP:journals/fss/Nguyen24}, is the smallest crisp equivalence relation on $\Delta^\mI$ that contains the pairs $\tuple{v_1, v'_1}$ and $\tuple{v_2, v'_2}$. It is also the greatest crisp $\Phi'_1$-auto-bisimulation of~$\mI$ (as the role constructors $\circ, \mor_r, *, ?$ are ``safe'' for bisimulation). 
The quotient fuzzy interpretation of $\mI$ w.r.t.\ this equivalence relation, defined as in~\cite{DBLP:journals/fss/Nguyen24}, is the fuzzy interpretation $\mI_1$ illustrated in Figure~\ref{fig: JHDKJ}. According to~\cite[Theorem~3.6]{DBLP:journals/fss/Nguyen24}, $\mI_1$ is a minimal fuzzy interpretation that validates the same set of fuzzy concept assertions (respectively, fuzzy TBox axioms) in \fALCregD as~$\mI$ over any linear and complete residuated lattice. In other words, $\mI_1$ is a minimal reduction of $\mI$ that preserves fuzzy concept assertions (respectively, fuzzy TBox axioms) in \fALCregD over any linear and complete residuated lattice.

In the case $\Phi \in \{\Phi_0, \Phi_2\}$, the greatest fuzzy $\Phi$-auto-bisimulation of~$\mI$ under the G\"odel semantics, denoted by $E$ and defined as in~\cite{minimization-by-fBS}, is the reflexive-symmetric closure of $\{\tuple{u,u'}\!:\!1$, $\tuple{v_1,v'_1}\!:\!1$, $\tuple{v_1,v'_2}\!:\!0.7$, $\tuple{v_2,v'_1}\!:\!0.7$, $\tuple{v_2,v'_2}\!:\!1$, $\tuple{v_3,v'_1}\!:\!0.7$, $\tuple{v_3,v'_2}\!:\!0.8\}$. The quotient fuzzy interpretation of $\mI$ w.r.t.\ this fuzzy equivalence relation under the G\"odel semantics, defined as in~\cite{minimization-by-fBS}, is the fuzzy interpretation $\mI_2$ illustrated in Figure~\ref{fig: JHDKJ} (cf.~\cite[Example~5]{minimization-by-fBS}). According to~\cite[Theorem~5]{minimization-by-fBS}, $\mI_2$ is a minimal fuzzy interpretation that validates the same set of fuzzy concept assertions in \fALCregU (respectively, fuzzy TBox axioms in \fALCregU, as well as fuzzy TBox axioms in \fALCreg) as~$\mI$ under the G\"odel semantics. In other words, $\mI_2$ is a minimal reduction of $\mI$ that preserves fuzzy concept assertions in \fALCregU (respectively, fuzzy TBox axioms in \fALCregU, as well as fuzzy TBox axioms in \fALCreg) under the G\"odel semantics.

Notice that $|\Delta^\mI| = 7$. 
Concentrating on the preservation of fuzzy concept assertions, note that $\mI_1$ is a minimal reduction of $\mI$ that preserves fuzzy concept assertions in \fALCregD (over any linear and complete residuated lattice), with $|\Delta^{\mI_1}| = 5$, and $\mI_2$ is a minimal reduction of $\mI$ that preserves fuzzy concept assertions in \fALCregU under the G\"odel semantics, with $|\Delta^{\mI_2}| = 4$. 
The former is the quotient fuzzy interpretation of $\mI$ w.r.t.\ the greatest crisp $\Phi'_1$-auto-bisimulation of~$\mI$, while the latter is the quotient fuzzy interpretation of $\mI$ w.r.t.\ the greatest fuzzy $\Phi_0$-auto-bisimulation of~$\mI$. Note that $\Phi'_1 = \{\triangle\}$ and a crisp $\{\triangle\}$-auto-bisimulation is a ``basic crisp auto-bisimulation'' (without additional conditions) \cite{DBLP:journals/tfs/NguyenN23,DBLP:journals/fss/Nguyen24}, while $\Phi_0 = \emptyset$ and a fuzzy $\emptyset$-auto-bisimulation is a ``basic fuzzy auto-bisimulation'' (without additional conditions) \cite{FSS2020,minimization-by-fBS}. 

Both the works \cite{minimization-by-fBS} and \cite{DBLP:journals/fss/Nguyen24} directly exploit the quotient constructions to reduce a given finite fuzzy interpretation, but they only reach a minimal reduction that preserves fuzzy concept assertions in a FDL that uses the Baaz projection operator~$\triangle$ or the universal role~$U$, respectively. 
The problem of minimizing a finite fuzzy interpretation while preserving fuzzy concept assertions in FDLs that lack these two constructors under the G\"odel semantics -- e.g., in \fALCreg -- was posed as a challenging open problem in~\cite{minimization-by-fBS}, and remained unresolved until now. The present work addresses and resolves this problem. As \fALCreg is a sublogic of \fALCregD and \fALCregU, minimizing a finite fuzzy interpretation while preserving fuzzy concept assertions only in \fALCreg may result in a smaller fuzzy interpretation than in \fALCregD and \fALCregU. In fact, as shown later in Example~\ref{example: JHRHS}, the fuzzy interpretation $\mI_3$ illustrated in Figure~\ref{fig: JHDKJ}, with $|\Delta^{\mI_3}| = 2$, is a minimal reduction of $\mI$ that preserves fuzzy concept assertions in \fALCreg under the G\"odel semantics. 

\subsection{Our contributions}

In this work, we present the {\em first} algorithm that minimizes a finite fuzzy interpretation while preserving fuzzy concept assertions in FDLs without the Baaz projection operator and the universal role, under the G\"odel semantics. The considered class of FDLs ranges from the sublogic of $\fALC$ without the union operator and universal restriction to the FDL that extends $\fALCreg$ with inverse roles and nominals. Our algorithm is provided in an extended form that supports approximate preservation: it minimizes a finite fuzzy interpretation $\mI$ while preserving fuzzy concept assertions up to a degree~$\gamma \in (0,1]$. The algorithm runs in \mbox{$O((m\log{l} + n)\log{n})$} time, where $n$ is the size of the domain of~$\mI$, $m$ is the number of nonzero instances of atomic roles in~$\mI$, and $l$ is the number of distinct fuzzy values used in such instances plus~2.

Methodologically, our algorithm exploits the compact fuzzy partition corresponding to the greatest fuzzy auto-bisimulation of~$\mI$ and constructs the reduced interpretation by selecting representatives of the fuzzy blocks of that partition. This strategy fundamentally differs from existing approaches, which rely on quotient constructions for minimizing fuzzy interpretations in FDLs and fuzzy automata.

\subsection{Structure of the work}

The remainder of this work is structured as follows. Section~\ref{section: prel} provides preliminaries on fuzzy sets, fuzzy relations, compact fuzzy partitions, FDLs, fuzzy bisimulations, and fuzzy labeled graphs. In Section~\ref{section: CCFP}, as preparation for the main result, we present an auxiliary algorithm for constructing the compact fuzzy partition corresponding to the greatest fuzzy $\Phi$-auto-bisimulation of a given finite fuzzy interpretation~$\mI$, where $\Phi$ specifies the considered FDL. Section~\ref{section: main} presents our main algorithm (as mentioned in the previous subsection), proves its correctness, and analyzes its complexity. Section~\ref{section: experiments} reports experimental results. Section~\ref{section: related work} discusses related work. Finally, Section~\ref{section: conc} provides a discussion and conclusions. In addition, Appendix~\ref{appendix A} presents further examples.


\section{Preliminaries}
\label{section: prel}

In this section, we recall basic notions and notation needed for this work, based on \cite{FSS2020,minimization-by-fBS,DBLP:journals/isci/Nguyen23,NGUYEN2025109194}. 

\subsection{Fuzzy sets and relations}

In this work, we consider FDLs under the G\"odel semantics, using the G\"odel t-norm and residuum defined as follows:
\[ a \fand b \defeq \min\{a, b\},\qquad\qquad 
   (a \fto b) \defeq \left\{\!\!\!
   		\begin{array}{ll}
   		1 & \!\textrm{if } a \leq b, \\ 
   		b & \!\textrm{otherwise},
   		\end{array}
        \right.
\]
for $a, b \in [0,1]$. 
The G\"odel t-norm biresiduum is defined as follows:
\[
    (a \fequiv b) \defeq \min\{(a \fto b), (b \fto a)\}, 
\]
and we have 
\[
	(a \fequiv b) = \left\{\!\!\!
       		\begin{array}{ll}
       		1 & \!\textrm{if } a = b, \\ 
       		\min\{a,b\} & \!\textrm{otherwise},
       		\end{array}
            \right.
\]
for $a, b \in [0,1]$. 
Given $D \subseteq [0,1]$, we define $\bigotimes\!D \defeq \inf D$. 

Let $A$, $B$, and $C$ be non-empty sets. A {\em fuzzy subset} of $A$, also called a {\em fuzzy set}, is a function \mbox{$\varphi : A \to [0,1]$}. By $\mF(A)$ we denote the family of all fuzzy subsets of~$A$. Given $\varphi, \psi \in \mF(A)$, we write $\varphi \leq \psi$ to mean that $\varphi(a) \leq \psi(a)$ for all $a \in A$. 
For $\varphi \in \mF(A)$, define the {\em support} of~$\varphi$ by $\support(\varphi) = \{a \in A \mid \varphi(a) > 0\}$. For $\{a_1, \ldots, a_k\} \subseteq A$ and $\{d_1, \ldots, d_k\} \subset (0,1]$, we write $\{a_1\!:\!d_1, \ldots, a_k\!:\!d_k\}$ to denote the fuzzy set $\varphi \in \mF(A)$ such that $\support(\varphi) = \{a_1, \ldots, a_k\}$ and $\varphi(a_i) = d_i$ for $1 \leq i \leq k$.

A \emph{fuzzy relation} between $A$ and $B$ is a fuzzy subset of $A \times B$. The \emph{inverse} of a fuzzy relation $\varphi \in \mF(A \times B)$ is the fuzzy relation $\cnv{\varphi} \in \mF(B \times A)$ defined by $\cnv{\varphi}(b,a) = \varphi(a,b)$ for $a \in A$ and $b \in B$. The \emph{composition} of fuzzy relations $\varphi \in \mF(A \times B)$ and $\psi \in \mF(B \times C)$ is the fuzzy relation $\varphi \circ \psi \in \mF(A \times C)$ given by
\[
(\varphi \circ \psi)(a,c) = \sup \{\varphi(a,b) \fand \psi(b,c) \mid b \in B\}
\]
for all $a \in A$ and $c \in C$. 
A fuzzy relation $\varphi \in \mF(A \times A)$ is called a {\em fuzzy relation on}~$A$. 

The {\em identity relation} $id_A$ on $A$ is defined by
\[
id_A(a,b) = 
\begin{cases}
1 & \text{if } a = b, \\
0 & \text{otherwise}.
\end{cases}
\]
A fuzzy relation $\varphi$ on $A$ is called
\begin{itemize}
\item \emph{reflexive} if $id_A \leq \varphi$,
\item \emph{symmetric} if $\cnv{\varphi} = \varphi$,
\item \emph{transitive} if $\varphi \circ \varphi \leq \varphi$.
\end{itemize}
A \emph{fuzzy equivalence} on $A$ is a fuzzy relation $\varphi$ on $A$ that satisfies the three above properties. 


\subsection{Compact fuzzy partitions}

Let $A$ denote a finite set and $\varphi$ a fuzzy equivalence on $A$. The {\em compact fuzzy partition corresponding to $\varphi$} (under the G\"odel semantics) \cite{DBLP:journals/isci/Nguyen23,NGUYEN2025109194} is the data structure $S$ specified as follows:
\begin{itemize}
\item if $\varphi(a,b) = 1$ for all $a,b \in A$, then $S$ has attributes $S.\elements = A$ and $S.\degree = 1$, and we call it a {\em crisp block} and denote it by $A_1$;
\item else:
    \begin{itemize}
	\item denote $d = \inf \{\varphi(a,b) \mid a,b \in A\}$;
	\item let $\sim$ denote the equivalence relation on $A$ defined by: $a \sim b$ iff $\varphi(a,b) > d$;
	\item let $\{A_1,\ldots,A_k\}$ be the partition of $A$ that corresponds to~$\sim$;
	\item for $1 \leq i \leq k$, denote by $\varphi_i$ the restriction of $\varphi$ to $A_i \times A_i$ and let $S_i$ recursively be the compact fuzzy partition corresponding to~$\varphi_i$;
	\item $S$ has attributes $S.\subblocks = \{S_1,\ldots,S_k\}$ and $S.\degree = d$, and we call it a {\em fuzzy block} and denote it by $\{S_1,\ldots,S_k\}_d$.
	\end{itemize}
\end{itemize}

\begin{table}[ht]
\footnotesize
\[
\begin{array}{|c||c|c|c|c|c|c|c|}
\hline
\varphi & a_1 & a_2 & a_3 & a_4 & a_5 & a_6 & a_7 \\
\hline\hline
a_1 & 1 & 0 & 0 & 0 & 0 & 0 & 0 \\
\hline
a_2 & 0 & 1 & 0.4 & 0.2 & 0.2 & 0.2 & 0.2 \\
\hline
a_3 & 0 & 0.4 & 1 & 0.2 & 0.2 & 0.2 & 0.2 \\
\hline
a_4 & 0 & 0.2 & 0.2 & 1 & 1 & 0.8 & 0.5 \\
\hline
a_5 & 0 & 0.2 & 0.2 & 1 & 1 & 0.8 & 0.5 \\
\hline
a_6 & 0 & 0.2 & 0.2 & 0.8 & 0.8 & 1 & 0.5 \\
\hline
a_7 & 0 & 0.2 & 0.2 & 0.5 & 0.5 & 0.5 & 1 \\
\hline
\end{array}
\]
\caption{The fuzzy equivalence $\varphi$ used in Example~\ref{example: JHFJH}.\label{table: fz r}}
\end{table}

\begin{example}\label{example: JHFJH}
Let $A = \{a_1, a_2, \ldots, a_7\}$ and let $\varphi$ be the fuzzy equivalence on $A$ specified in Table~\ref{table: fz r}. 
The compact fuzzy partition corresponding to $\varphi$ is:
\[
\{
    \{a_1\}_1, 
    \{
        \{
            \{a_2\}_1, 
            \{a_3\}_1
        \}_{0.4},
        \{
            \{
                \{a_4, a_5\}_1,
                \{a_6\}_1
            \}_{0.8},
            \{a_7\}_1 
        \}_{0.5}
    \}_{0.2} 
\}_0.
\]

\vspace{-1.5em}
\myend
\end{example}

Given a fuzzy or crisp block $S$, we define inductively:
\begin{itemize}
\item $S.\elements()$ to be the set $S.\elements$ if $S$ is a crisp block, and the set $\bigcup \{S'.\elements() \mid S' \in S.\subblocks\}$ otherwise;
\item $S.blocks()$ to be $\{S\}$ if $S$ is a crisp block, and the set $\{S\} \,\cup\, \bigcup \{S'.\blocks() \mid S' \in S.\subblocks\}$ otherwise.
\end{itemize}
In words, $S.\elements()$ is the carrier set of $S$ and $S.blocks()$ is the set consisting of all descendant blocks of $S$ including itself.


\subsection{Fuzzy description logics under the G\"odel semantics}

Let $\Phi$ be a subset of $\{I, O\}$, where $I$ and $O$ represent the features of inverse roles and nominals, respectively. In this subsection, we define the syntax and semantics of the FDL $\mLP$, which extends the DL \ALCreg by incorporating fuzzy truth values as well as the features specified by~$\Phi$.

The language of $\mLP$ employs a finite set $\RN$ of {\em role names}, a finite set $\CN$ of {\em concept names}, and a finite and non-empty set $\IN$ of {\em individual names}.\footnote{The assumption $\IN \neq \emptyset$ is made in this work for the sake of simplicity.} We assume that these three sets are pairwise disjoint. A {\em basic role} w.r.t.~$\Phi$ is either a role name from $\RN$ or, if $I \in \Phi$, the inverse $\cnv{r}$ of a role name~$r$.

{\em Concepts} and {\em roles} of $\mLP$ are defined as follows:
\begin{itemize}
	\item every role name $r \in \RN$ is a role of $\mLP$, 
	\item if $R$ and $R'$ are roles of $\mLP$ and $C$ is a concept of $\mLP$, then $R \mor R'$, $R \circ R'$, $R^*$ and $C?$ are roles of $\mLP$,
	\item if $R$ is a role of $\mLP$ and $I \in \Phi$, then $\cnv{R}$ is a role of $\mLP$,
	
	
	\item every value $d \in [0,1]$ is a concept of $\mLP$,
	\item every concept name $A \in \CN$ is a concept of $\mLP$,
	\item if $R$ is a role of $\mLP$, and $C$ and $C'$ are concepts of $\mLP$, then $C \mor C'$, $C \mand C'$, $C \to C'$, $\E R.C$ and $\V R.C$ are concepts of $\mLP$, 
	\item if $a \in \IN$ and $O \in \Phi$, then $\{a\}$ is a concept of~$\mLP$.
\end{itemize}

Concept and role names are also called {\em atomic concepts} and {\em atomic roles}, respectively. 
We use letters such as $C$ to denote concepts, $R$ to denote roles, and $a$, $b$ to denote individual names. 

We denote by $\mLPp$ the largest sublanguage of $\mLP$ that excludes the concept constructors $C \mor C'$ and $\V R.C$, as well as the role constructors $R \mor R'$, $R \circ R'$, $R^*$, and $C?$. 

For a finite set of concepts $X = \{C_1,\ldots,C_k\}$, we denote 
\[ \textstyle\bigsqcap\!X = C_1 \mand \ldots \mand C_k \mand 1. \]

\begin{figure}
\footnotesize
\[
\begin{array}{|l|l|}
\hline
\begin{array}{l}
		(R \mor R')^\mI(x,y) = \max\{R^\mI(x,y), R'^\mI(x,y)\} \\[0.8ex]
		(R \circ R')^\mI(x,y) = \sup\{R^\mI(x,z) \fand R'^\mI(z,y) \mid z \in \Delta^\mI \} \\[0.8ex]
		(R^*)^\mI(x,y) = \sup \{\bigotimes\{R^\mI(x_i,x_{i+1}) \mid 0 \leq i < k\} \mid\\ 
		\qquad k \geq 0,\ x_0,\ldots,x_k \in \Delta^\mI,\ x_0 = x,\ x_k = y\} \\[0.8ex]
		(C?)^\mI(x,y) = 
            \begin{cases}
            C^\mI(x) & \textrm{if } x = y\\
            0 & \textrm{otherwise}
            \end{cases}  \\[3.0ex]
		(\cnv{R})^\mI(x,y) = R^\mI(y,x)
\end{array}
&
\begin{array}{l}
        \ \\[-0.5em]
		d^\mI(x) = d \\[0.8ex]
		(C \mor C')^\mI(x) = \max\{C^\mI(x), C'^\mI(x)\} \\[0.8ex]
		(C \mand C')^\mI(x) = \min\{C^\mI(x), C'^\mI(x)\} \\[0.8ex]
		(C \to C)^\mI(x) = (C^\mI(x) \fto C'^\mI(x)) \\[0.8ex]
		(\E R.C)^\mI(x) = \sup \{R^\mI(x,y) \fand C^\mI(y) \mid y \in \Delta^\mI\} \\[0.5ex]
		(\V R.C)^\mI(x) = \inf \{R^\mI(x,y) \fto C^\mI(y) \mid y \in \Delta^\mI\} \\[0.5em]
		\{a\}^\mI(x) = \begin{cases}
        1 & \textrm{if } x = a^\mI \\
        0 & \textrm{otherwise}
        \end{cases}
\end{array} \\
\hline
\end{array}
\]
\caption{The semantics of complex concepts and roles.\label{fig: HGDJH}}
\end{figure}

A {\em fuzzy interpretation} is a pair \mbox{$\mI = \langle \Delta^\mI, \cdot^\mI \rangle$} consisting of a~non-empty set $\Delta^\mI$, called the {\em domain}, and the {\em interpretation function} $\cdot^\mI$, which assigns to each individual name $a$ an element \mbox{$a^\mI \in \Delta^\mI$}, to each concept name $A$ a fuzzy subset $A^\mI$ of $\Delta^\mI$, and to each role name $r$ a fuzzy binary relation $r^\mI$ on $\Delta^\mI$. 
The function $\cdot^\mI$ is extended to complex concepts and roles as defined in Figure~\ref{fig: HGDJH}, where suprema and infima are taken in the complete lattice $[0, 1]$. 
A fuzzy interpretation $\mI$ is said to be {\em finite} if its domain $\Delta^\mI$ is finite.


\newcommand{\hasExpertiseIn}{\mathit{hasExpertiseIn}}
\newcommand{\collaboratesWith}{\mathit{collaboratesWith}}
\newcommand{\isRelatedTo}{\mathit{isRelatedTo}}

\newcommand{\Topic}{\mathit{Topic}}
\newcommand{\Researcher}{\mathit{Researcher}}
\newcommand{\FuzzyLogic}{\mathit{FuzzyLogic}}
\newcommand{\DescLogic}{\mathit{DescriptionLogic}}

\newcommand{\stefan}{\mathit{stefan}}
\newcommand{\mirek}{\mathit{mirek}}
\newcommand{\linh}{\mathit{linh}}
\newcommand{\fAutomata}{\mathit{fuzzy\!\_automata}}
\newcommand{\fDescLogic}{\mathit{FDL}}
\newcommand{\minimization}{\mathit{minimization}}
\newcommand{\bisimulation}{\mathit{bisimulation}}
\newcommand{\simulation}{\mathit{simulation}}

\begin{example}\label{example: JHKSM}
Let the signature consist of $\RN = \{\hasExpertiseIn$, $\collaboratesWith$, $\isRelatedTo\}$, $\CN = \{\Researcher$, $\Topic$, $\FuzzyLogic$, $\DescLogic\}$, and $\IN = \{\linh$, $\mirek$, $\stefan\}$. Consider the fuzzy interpretation $\mI$ with:
\begin{itemize}
\item $\Delta^\mI = \{\linh$, $\mirek$, $\stefan$, $\fDescLogic$, $\fAutomata$, $\simulation$, $\bisimulation$, $\minimization\}$. 
\item $\linh^\mI = \linh$, $\mirek^\mI = \mirek$, $\stefan^\mI = \stefan$,
\item $\Researcher^\mI = \{\linh\!:\!1, \mirek\!:\!1, \stefan\!:\!1\}$, 
\item $\Topic^\mI = \{\fDescLogic\!:\!1, \fAutomata\!:\!1, \simulation\!:\!1, \bisimulation\!:\!1, \minimization\!:\!1\}$, 
\item $\FuzzyLogic^\mI = \{\fDescLogic\!:\!0.8, \fAutomata\!:\!0.4\}$, 
\item $\DescLogic^\mI = \{\fDescLogic\!:\!0.8\}$, 
\item $\hasExpertiseIn^\mI = \{\tuple{\linh,\fDescLogic}\!:\!0.8$, $\tuple{\linh,\bisimulation}\!:\!0.9$, $\tuple{\mirek,\fAutomata}\!:\!0.9$,  $\tuple{\mirek,\bisimulation}\!:\!0.8$, $\tuple{\stefan,\fAutomata}\!:\!0.9$, $\tuple{\stefan,\bisimulation}\!:\!0.8\}$, 
\item $\collaboratesWith^\mI = \{\tuple{\linh,\stefan}\!:\!0.3$, $\tuple{\stefan,\linh}\!:\!0.3$, $\tuple{\mirek,\stefan}\!:\!0.5$, $\tuple{\stefan,\mirek}\!:\!0.6\}$, 
\item $\isRelatedTo^\mI$ equal to the reflexive-symmetric-transitive closure of $\{\tuple{\bisimulation,\fDescLogic}\!:\!0.6$, $\tuple{\bisimulation,\fAutomata}\!:\!0.6$, $\tuple{\bisimulation,\simulation}\!:\!0.8$, $\tuple{\bisimulation,\minimization}\!:\!0.5\}$, under the G\"odel semantics.
\end{itemize}
We have:
\begin{eqnarray*}
(\E \hasExpertiseIn.\E \isRelatedTo.\DescLogic)^\mI & = & \{\linh\!:\!0.8, \mirek\!:\!0.6, \stefan\!:\!0.6\}, \\
(\V \hasExpertiseIn.\E \isRelatedTo.\DescLogic)^\mI & = & \{\linh\!:\!0.6, \mirek\!:\!0.6, \stefan\!:\!0.6\}, \\
(\E \collaboratesWith^*.\E \hasExpertiseIn.\DescLogic)^\mI & = & \{\linh\!:\!0.8, \mirek\!:\!0.3, \stefan\!:\!0.3\}.
\end{eqnarray*}
\myend
\end{example}


A {\em fuzzy assertion} in $\mLP$ has the form $C(a) \bowtie d$, $R(a,b) \bowtie d$, $a \doteq b$, or $a \not\doteq b$, where $C$ is a concept of $\mLP$, $R$ is a role of $\mLP$, $\bowtie\ \in \{>, \geq, <, \leq\}$ and $d \in [0,1]$. 
A {\em fuzzy concept assertion} is a fuzzy assertion of the form $C(a) \bowtie d$. 
A~{\em fuzzy ABox} in $\mLP$ is a finite set of fuzzy assertions in $\mLP$. 

We say that a fuzzy interpretation $\mI$ {\em satisfies} a fuzzy assertion $\psi$, and write \mbox{$\mI \models \psi$}, if:
\begin{itemize}
	\item case $\psi = (C(a) \bowtie d)\,$: $C^\mI(a^\mI) \bowtie d$,
	\item case $\psi = (R(a,b) \bowtie d)\,$: $R^\mI(a^\mI,b^\mI) \bowtie d$,
	\item case $\psi = (a \doteq b)\,$: $a^\mI = b^\mI$, 
	\item case $\psi = (a \not\doteq b)\,$: $a^\mI \neq b^\mI$. 
\end{itemize}
In addition, we say that $\mI$ is a {\em model} of a fuzzy ABox $\mA$, and write $\mI \models \mA$, if it satisfies all assertions in~$\mA$.


\subsection{Fuzzy bisimulations}
\label{sec: fus-bis}

In this subsection, we recall the notion of fuzzy bisimulation defined in~\cite{FSS2020} for FDLs under the G{\"o}del semantics, but restrict it to the case $\Phi \subseteq \{I,O\}$. We then recall some results of~\cite{FSS2020}.

Let $\mI$ and $\mIp$ be fuzzy interpretations.
A fuzzy relation $Z \in \mF(\Delta^\mI \times \Delta^\mIp)$ is called a {\em fuzzy $\Phi$-bisimulation} between $\mI$ and $\mIp$ if it satisfies the following conditions, for every $x \in \Delta^\mI$, $x' \in \Delta^\mIp$, $A \in \CN$, $a \in \IN$, and every basic role $R$ w.r.t.~$\Phi$:
\begin{eqnarray}
	&& Z(x,x') \leq \left(A^\mI(x) \fequiv A^\mIp(x')\right); \label{eq: FB 2} \\
	&& \V y \in \Delta^\mI\, \E y' \in \Delta^\mIp\, \left(Z(x,x') \fand R^\mI(x,y) \leq Z(y,y') \fand R^\mIp(x',y')\right); \label{eq: FB 3} \\
    && \V y' \in \Delta^\mIp\, \E y \in \Delta^\mI\, \left(Z(x,x') \fand R^\mIp(x',y') \leq Z(y,y') \fand R^\mI(x,y)\right); \label{eq: FB 4}
\end{eqnarray}
if $O \in \Phi$, then
\begin{eqnarray}
	&& Z(x,x') \leq (x = a^\mI \Leftrightarrow x' = a^\mIp). \label{eq: FB 5}
\end{eqnarray}

A fuzzy $\Phi$-bisimulation between $\mI$ and itself is called a {\em fuzzy $\Phi$-auto-bisimulation} of~$\mI$.

The following results are restricted versions of the ones of~\cite{FSS2020} (cf.~\cite{minimization-by-fBS}).

\begin{theorem}\label{theorem: JHDHJ}
If $Z$ is a fuzzy $\Phi$-bisimulation of between finite fuzzy interpretations $\mI$ and $\mIp$, then for every $x \in \Delta^\mI$, $x' \in \Delta^\mIp$ and every concept $C$ of $\mLP$, 
\[ Z(x,x') \leq \left(C^\mI(x) \fequiv C^\mIp(x')\right). \]
\end{theorem}

\begin{theorem} \label{theorem: fG H-M}
The greatest fuzzy $\Phi$-bisimulation between finite fuzzy interpretations $\mI$ and $\mIp$ exists and is equal to the fuzzy relation $Z \in \mF(\Delta^\mI \times \Delta^\mIp)$ defined by
\[ Z(x,x') = \inf\{C^\mI(x) \fequiv C^\mIp(x') \mid \textrm{$C$ is a concept of $\mLPp$}\}. \]
\end{theorem}


\begin{corollary}
The greatest fuzzy $\Phi$-auto-bisimulation of any finite fuzzy interpretation $\mI$ exists and is a fuzzy equivalence. 
\end{corollary}

We say that fuzzy interpretations $\mI$ and $\mI'$ are {\em $\Phi$-bisimilar} if there exists a fuzzy $\Phi$-bisimulation $Z$ between them such that $Z(a^\mI,a^\mIp)=1$ for all $a \in \IN$. 

A fuzzy ABox $\mA$ is said to be {\em invariant under $\Phi$-bisimilarity} between finite fuzzy interpretations if, for every finite fuzzy interpretations $\mI$ and $\mI'$ that are $\Phi$-bisimilar, $\mI \models \mA$ iff $\mI' \models \mA$. 

\begin{theorem}\label{theorem: IFDMS}
Let $\mA$ be a fuzzy ABox in $\mLP$. If $O \in \Phi$ or $\mA$ consists of only fuzzy concept assertions, then $\mA$ is invariant under $\Phi$-bisimilarity between finite fuzzy interpretations.
\end{theorem}


\subsection{Fuzzy labeled graphs}

A {\em fuzzy labeled graph} (\FLG) \cite{DBLP:journals/isci/Nguyen23} is a structure $G = \tuple{V, E, L, \SV, \SE}$, where 
$V$ is a set of vertices, 
$\SV$ is a set of vertex labels, 
$L: V \to \mF(\SV)$ is a function that labels vertices, 
$\SE$ is a set of edge labels, 
and \mbox{$E \in \mF(V \times \SE \times V)$} is a fuzzy set of labeled edges. 
It is {\em finite} if $\SV$, $\SE$ and $V$ are finite. 

A {\em fuzzy auto-bisimulation} of an \FLG $G = \tuple{V, E, L, \SV, \SE}$ is a fuzzy relation $Z \in \mF(V \times V)$ that satisfies the following conditions for every $p \in \SV$, $r \in \SE$, and $x,x',y,y' \in V$:
\begin{eqnarray}
&& Z(x,x') \leq (L(x)(p) \fequiv L(x')(p)) \label{eq: FB1} \\
&& \E y'\! \in\!V (Z(x,x') \!\fand\! E(x,r,y) \leq Z(y,y') \!\fand\! E(x',r,y')) \label{eq: FB2} \\
&& \E y \in\!V (Z(x,x') \!\fand\! E(x',r,y') \leq Z(y,y') \!\fand\! E(x,r,y)). \label{eq: FB3}
\end{eqnarray}

It is known that the greatest fuzzy auto-bisimulation of any finite \FLG exists and is a fuzzy equivalence \cite[Corollary~5.3]{FBSML}. In \cite{DBLP:journals/isci/Nguyen23}, we provided an algorithm for computing the compact fuzzy partition corresponding to the greatest fuzzy auto-bisimulation of a finite \FLG~$G$. It runs in time $O((m\log{l} + n)\log{n})$, where $n = |V|$, $m = |\support(E)|$, and $l = |\{E(e) : e \in \support(E)\} \cup \{0,1\}|$. 

\section{Computing compact fuzzy partitions}
\label{section: CCFP}

As preparation for the next section, we present below Algorithm~\ref{alg1}, which, given a finite fuzzy interpretation~$\mI$ together with $\Phi \subseteq \{I, O\}$, constructs the compact fuzzy partition corresponding to the greatest fuzzy $\Phi$-auto-bisimulation of~$\mI$.

Following~\cite[Algorithm~2]{DBLP:journals/fss/Nguyen24}, we define the {\em fuzzy graph corresponding to $\mI$ w.r.t.\ $\Phi \subseteq \{I,O\}$} to be $G = \tuple{V, E, L, \SV, \SE}$ specified as follows:
\begin{itemize}
\item if $O \in \Phi$, then $\SV = \CN \cup \IN$, else $\SV = \CN$;
\item if $I \in \Phi$, then $\SE = \RN \cup \{\cnv{r} \mid r \in \RN\}$, else $\SE = \RN$;
\item $V = \Delta^\mI$;
\item $E(x,R,y) = R^\mI(x,y)$, for $R \in \SE$ and $x,y \in V$;
\item $L(x)(A) = A^\mI(x)$, for $A \in \CN$ and $x \in V$;
\item if $O \in \Phi$, then $L(x)(a)$ = (if $x \neq a^\mI$ then 0 else 1), for $a \in \IN$ and $x \in V$.
\end{itemize}

\begin{algorithm}
\caption{ComputeCompactFuzzyPartition\label{alg1}}
\Input{a finite fuzzy interpretation $\mI$ and $\Phi \subseteq \{I,O\}$.}
\Output{the compact fuzzy partition corresponding to the greatest fuzzy $\Phi$-auto-bisimulation of~$\mI$.}

\BlankLine

construct the fuzzy graph $G$ corresponding to $\mI$ w.r.t.~$\Phi$\label{step: alg1 1}\;

execute the {\em ComputeFuzzyPartitionEfficiently} algorithm from \cite{DBLP:journals/isci/Nguyen23} to compute the compact fuzzy partition $\bB$ corresponding to the greatest fuzzy auto-bisimulation of~$G$\label{step: alg1 2}\;

\Return{$\bB$};
\end{algorithm}

\begin{lemma}\label{lemma: JHJLS}
Algorithm~\ref{alg1} is correct and has the time complexity \mbox{$O((m\log{l} + n)\log{n})$}, where $n = |\Delta^\mI|$, $m = |\{\tuple{x,r,y} \mid r \in \RN$, $x,y \in \Delta^\mI$, $r^\mI(x,y) > 0\}|$ and $l = |\{r^\mI(x,y) \mid r \in \RN$, $x,y \in \Delta^\mI\}| + 2$, under the assumption that the sizes of $\RN$, $\CN$ and $\IN$ are constants.
\end{lemma}

\begin{proof}
It is straightforward to check that $Z$ is a fuzzy $\Phi$-auto-bisimulation of~$\mI$ iff $Z$ is a fuzzy auto-bisimulation of~$G$. Hence, the greatest fuzzy auto-bisimulation of~$G$ is the greatest fuzzy $\Phi$-auto-bisimulation of~$\mI$ and $\bB$ is in fact the compact fuzzy partition corresponding to the greatest fuzzy $\Phi$-auto-bisimulation of~$\mI$. 

Statement~\ref{step: alg1 1} runs in $O(m+n)$ time. By~\cite[Theorem~A.2]{DBLP:journals/isci/Nguyen23}, statement~\ref{step: alg1 2} runs in \mbox{$O((m\log{l} + n)\log{n})$} time. Consequently, the time complexity of Algorithm~\ref{alg1} is of this latter order. 
\myend
\end{proof}

\section{Minimizing fuzzy interpretations}
\label{section: main}

In this section, we present Algorithm~\ref{alg2}, which, given a finite fuzzy interpretation $\mI$ together with $\Phi \subseteq \{I,O\}$ and $\gamma \in (0,1]$, constructs a {\em minimal reduction of $\mI$} (w.r.t.\ the size of the domain) {\em that preserves fuzzy concept assertions of $\mL$ up to degree~$\gamma$}, where $\mL$ is any description language between $\mLPp$ and $\mLP$ (i.e., $\mLPp \subseteq \mL \subseteq \mLP$). Such a minimal reduction is defined to be a finite fuzzy interpretation $\mIp$ such that
\begin{itemize}
\item for every concept $C$ of $\mL$ and every $a \in \IN$, $\gamma \leq (C^\mI(a^\mI) \fequiv C^\mIp(a^\mIp))$;
\item if $\mIdp$ is a finite fuzzy interpretation such that, for every concept $C$ of $\mL$ and every $a \in \IN$, $\gamma \leq (C^\mI(a^\mI) \fequiv C^\mIdp(a^\mIdp))$, then $|\Delta^\mIdp| \geq |\Delta^\mIp|$. 
\end{itemize}
In the case $\gamma = 1$, $\mIp$ is also called a {\em minimal reduction of $\mI$ that preserves fuzzy concept assertions of~$\mL$}. 

Given a compact fuzzy partition $\bB$ of a set $X$ together with $B \in \bB.\blocks()$, we define $B.\parent() = \Null$ if $B = \bB$, and define $B.\parent()$ to be the block $B' \in \bB.\blocks()$ such that $B \in B'.\subblocks$ otherwise. Given $x \in \bB.\elements()$ and $d \in (0,1]$, we define $\bB.\findBlock(x,d)$ to be the block $B \in \bB.\blocks()$ with the smallest $B.\degree$ such that $B.\degree \geq d$ and $x \in B.\elements()$. 

Algorithm~\ref{alg2} first executes Algorithm~\ref{alg1} to construct the compact fuzzy partition $\bB$ corresponding to the greatest fuzzy $\Phi$-auto-bisimulation of~$\mI$. It then constructs a reduction $\mIp$ of $\mI$ by using $\bB$, starting with setting a new attribute $B.\repr$ to $\Null$, for each block $B \in \bB.\blocks()$, where $B.\repr$ is intended to specify the ``representative'' of $B$ (which belongs to $B.\elements()$ and $\Delta^\mI$). Roughly speaking, if $B.\degree = d$ and $B.\repr \neq \Null$, then $B.\repr$ can be used to represent any element of $B$ ``at the level~$d$''. The domain of $\mIp$ consists of all $B.\repr$ such that $B \in \bB.\blocks()$ and $B.\repr \neq \Null$. Whenever $B.\repr$ is set to an element $v$ different from $\Null$, if $B' = B.\parent() \neq \Null$ and $B'.\repr = \Null$, then $B'.\repr$ is also set to~$v$. For each $x \in \Delta^\mIp$, the ``contents'' of $x$ in $\mIp$ are the same as in $\mI$, which means that $A^\mIp(x) = A^\mI(x)$ for all $A \in \CN$. 

Algorithm~\ref{alg2} uses a priority queue $Q$, which is initialized and extended as follows: when $x$ is added (or has just been added) to $\Delta^\mIp$, if there exists $y \in \Delta^\mI$ and a basic role $R$ w.r.t.~$\Phi$ such that $R^\mI(x,y) > 0$, then the triple $\tuple{x,R,y}$ is inserted to $Q$, where $R^\mI(x,y)$ specifies its priority. 

As an initialization, for each $a \in \IN$ and for $B = \bB.\findBlock(a^\mI, \gamma)$, if $B.\repr = \Null$, then $a^\mI$ is added to $\Delta^\mIp$, $a^\mIp$ is set to $a^\mI$ and $B.\repr$ is set to $a^\mI$, else $a^\mIp$ is set to $B.\repr$. After that, as the main outer loop, values $d$ from the following set are considered in decreasing order:
\[ D = \{\gamma\} \cup \{r^\mI(x,y) \mid x,y \in \Delta^\mI, r \in \RN \textrm{ and } r^\mI(x,y) \leq \gamma\} \setminus \{0\}. \]
As the inner loop, each triple $\tuple{x,R,y}$ from the priority $Q$ with $R^\mI(x,y) \geq d$ is extracted and processed as follows:
\begin{itemize}
\item with $B = \bB.\findBlock(y,d)$, if $B.\repr = \Null$, then $y$ is added to $\Delta^\mIp$ and $B.\repr$ is set to $y$;
\item with $y' = B.\repr$, if $R^\mIp(x,y') = 0$, then $R^\mIp(x,y')$ is set to $d$.
\end{itemize} 

We have implemented Algorithm~\ref{alg2} in Python and made the source code publicly available~\cite{min2025-prog}. The following three examples illustrate the execution of this algorithm. One can use that implementation to verify their results and to inspect additional details. 

\begin{algorithm}
\caption{\MinFInt\label{alg2}}
\Input{a finite fuzzy interpretation $\mI$, $\Phi \subseteq \{I,O\}$ and $\gamma \in (0,1]$.}
\Output{a minimal reduction of $\mI$ that preserves fuzzy concept assertions of $\mL$ up to degree $\gamma$, where $\mL$ is any description language between $\mLPp$ and $\mLP$.}

\BlankLine

execute Algorithm~\ref{alg1} to construct the compact fuzzy partition $\bB$ corresponding to the greatest fuzzy $\Phi$-auto-bisimulation of~$\mI$\label{step: alg2 1}\;

\lForEach{$B \in \bB.\blocks()$}{$B.\repr := \Null$\label{step: alg2 2}}

\BlankLine

initialize $\mIp$ by setting $\Delta^\mIp = \emptyset$\label{step: alg2 3}\;

\ForEach{$a \in \IN$}{
    $B := \bB.\findBlock(a^\mI, \gamma)$\label{step: alg2 6}\;
    \uIf{$B.\repr = \Null$}{
        add $a^\mI$ to $\Delta^\mIp$ and set $a^\mIp := a^\mI$\label{step: alg2 8}\;
        \lForEach{concept name $A \in \CN$}{$A^\mIp(a^\mIp) := A^\mI(a^\mI)$}
        \lRepeat{$B = \Null$ or $B.\repr \neq \Null$}{$B.\repr := a^\mI$, $B := B.\parent()$\label{step: alg2 10}}
    }
    \lElse{$a^\mIp := B.\repr$\label{step: alg2 11}}
}

\BlankLine

initialize a priority queue $Q$ with all triples $\tuple{x,R,y}$ such that $x \in \Delta^\mIp$, $y \in \Delta^\mI$ and $R$ is a basic role w.r.t.~$\Phi$ with $R^\mI(x,y) > 0$, where $R^\mI(x,y)$ specifies the priority\label{step: alg2 12m}\;

let $D = \{\gamma\} \cup \{r^\mI(x,y) \mid x,y \in \Delta^\mI, r \in \RN$ and $r^\mI(x,y) < \gamma\} \setminus \{0\}$\label{step: alg2 12}\;
\ForEach{$d \in D$ in decreasing order\label{step: alg2 13}}{
    \While{$Q$ is not empty and its top element $\tuple{x,R,y}$ satisfies $R^\mI(x,y) \geq d$\label{step: alg2 14}}{
        extract $\tuple{x,R,y}$ from $Q$\label{step: alg2 15}\;
        $B := \bB.\findBlock(y,d)$, $B' := B$\label{step: alg2 16}\;
        \If{$B.\repr = \Null$\label{step: alg2 17}}{
            add $y$ to $\Delta^\mIp$\label{step: alg2 18}\;
            \lForEach{concept name $A \in \CN$}{$A^\mIp(y) := A^\mI(y)$}
            \lRepeat{$B' = \Null$ or $B'.\repr \neq \Null$}{$B'.\repr := y$, $B' := B'.\parent()$\label{step: alg2 20}}
            insert into $Q$ all triples $\tuple{y,S,z}$ such that $z \in \Delta^\mI$ and $S$ is a basic role w.r.t.~$\Phi$ with $S^\mI(y,z) > 0$\label{step: alg2 insertions into Q}\;
        }
        $y' := B.\repr$\label{step: alg2 21}\;
        \If{$R^\mIp(x,y') = 0$ (i.e., was not set)}{
            \lIf{$R$ is a role name}{$R^\mIp(x,y') := d$\label{step: alg2 23}}
            \lElse{$r^\mIp(y', x) := d$, where $r$ is the role name such that $R = \cnv{r}$\label{step: alg2 24}}
        }
    }
}

\BlankLine

\Return $\mIp$\;
\end{algorithm}

\begin{example}\label{example: JHRHS}
Let $\mI$ be the fuzzy interpretation from the motivating example given in Section~\ref{sec: motivation} and let $\Phi = \emptyset$. The specification of $\mI$ is stored in the file {\em in1.txt} of~\cite{min2025-prog}. The greatest fuzzy $\Phi$-auto-bisimulation of~$\mI$ is the fuzzy equivalence relation $E$ mentioned in that subsection. The compact fuzzy partition corresponding to $E$ is $\bB = \{\{u,u'\}_1$, $\{\{v_1,v'_1\}_1$, $\{\{v_2,v'_2\}_1$, $\{v_3\}_1\}_{0.8}\}_{0.7}\}_0$. Consider the execution of Algorithm~\ref{alg2} for $\mI$ using $\Phi$ and $\gamma = 1$. The resulting fuzzy interpretation, denoted as $\mI_3$ in Figure~\ref{fig: JHDKJ}, is illustrated there. Details are provided below.

After executing statements \ref{step: alg2 1}-\ref{step: alg2 12}, we have:
\begin{itemize}
\item $\Delta^\mIp = \{u\}$, $a^\mIp = b^\mIp = u$;
\item for $B \in \bB.\blocks()$, $B.\repr = u$ if $B = \{u,u'\}_1$ or $B = \bB$, and $B.\repr = \Null$ otherwise;
\item $Q$ contains $\tuple{u,r,v_1}$, $\tuple{u,r,v_2}$ and $\tuple{u,r,v_3}$;
\item $D = \{1, 0.7, 0.6, 0.5, 0.4\}$. 
\end{itemize}

The first iteration of the ``foreach'' loop at statement~\ref{step: alg2 13} is executed with $d = 1$, in which the inner ``while'' loop terminates immediately.

The second iteration of the ``foreach'' loop at statement~\ref{step: alg2 13} is executed with $d = 0.7$. During this iteration, the inner ``while'' loop executes only one iteration, with $\tuple{x,R,y} = \tuple{u,r,v_3}$, in which:
\begin{itemize}
\item for $B = \bB.\findBlock(v_3, 0.7) = \{\{v_1,v'_1\}_1, \{\{v_2,v'_2\}_1, \{v_3\}_1\}_{0.8}\}_{0.7}$, since $B.\repr = \Null$, $v_3$ is added to $\Delta^\mIp$ and $B.\repr$ is set to $v_3$;
\item $r^\mIp(u,v_3)$ is set to $0.7$. 
\end{itemize} 

The third iteration of the ``foreach'' loop at statement~\ref{step: alg2 13} is executed with $d = 0.6$, in which the inner ``while'' loop terminates immediately.

The fourth iteration of the ``foreach'' loop at statement~\ref{step: alg2 13} is executed with $d = 0.5$. During this iteration, the inner ``while'' loop executes only one iteration, with $\tuple{x,R,y} = \tuple{u,r,v_1}$, in which $\mIp$ remains unchanged because, for $B = \bB.\findBlock(v_1, 0.5) = \{\{v_1,v'_1\}_1$, $\{\{v_2,v'_2\}_1$, $\{v_3\}_1\}_{0.8}\}_{0.7}$, we have $B.\repr = v_3$ and $r^\mIp(u,v_3) = 0.7$. 

The fifth iteration of the ``foreach'' loop at statement~\ref{step: alg2 13} is executed with $d = 0.4$. During this iteration, the inner ``while'' loop executes only one iteration, with $\tuple{x,R,y} = \tuple{u,r,v_2}$, in which $\mIp$ remains unchanged. The reasons are similar to those given in the above paragraph.

The algorithm terminates and returns $\mIp$ with $|\Delta^\mIp| = 2$.
\myend
\end{example}

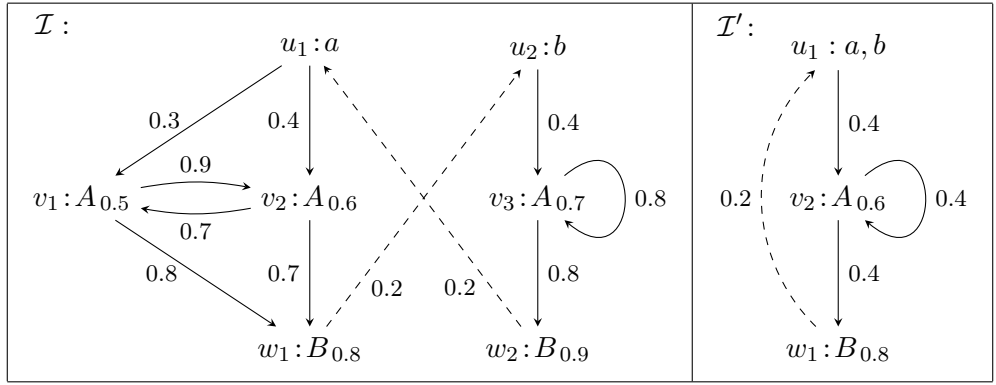
\begin{figure}[t]
\begin{center}
\begin{tabular}{|c|c|}
\hline
\begin{tikzpicture}[->,>=stealth,auto,black]
	\node (I) {$\mI:\qquad$};
	\node (bI) [node distance=0.3cm, below of=I] {};
	\node (u1) [node distance=3.0cm, right of=bI] {$u_1\!:\!a$};
	\node (u2) [node distance=3.0cm, right of=u1] {$u_2\!:\!b$};
	\node (v1) [node distance=2.0cm, below of=bI] {$v_1\!:\!A_{\,0.5}$};
	\node (v2) [node distance=2.0cm, below of=u1] {$v_2\!:\!A_{\,0.6}$};
	\node (v3) [node distance=2.0cm, below of=u2] {$v_3\!:\!A_{\,0.7}$};
	\node (w1) [node distance=2.0cm, below of=v2] {$w_1\!:\!B_{\,0.8}$};
	\node (w2) [node distance=2.0cm, below of=v3] {$w_2\!:\!B_{\,0.9}$};
	\draw (u1) to node [left,xshift=-3pt]{\footnotesize{0.3}} (v1);	
	\draw (u1) to node [left]{\footnotesize{0.4}} (v2);	
	\draw (u2) to node [right]{\footnotesize{0.4}} (v3);	
	\draw (v1) to node [left,xshift=-3pt]{\footnotesize{0.8}} (w1);	
	\draw (v2) to node [left]{\footnotesize{0.7}} (w1);	
	\draw (v3) to node [right]{\footnotesize{0.8}} (w2);
    \draw (v1) edge[bend left=10] node[above]{\footnotesize{0.9}} (v2);
    \draw (v2) edge[bend left=10] node[below]{\footnotesize{0.7}} (v1);
    \draw (v3) edge[out=40,in=-40,looseness=6] node[right]{\footnotesize{0.8}} (v3);
	\draw[dashed] (w1) to node[right,pos=0.15,xshift=2pt]{\footnotesize{0.2}} (u2);	
	\draw[dashed] (w2) to node[left,pos=0.15,xshift=-2pt]{\footnotesize{0.2}} (u1);
\end{tikzpicture}
&
\begin{tikzpicture}[->,>=stealth,auto,black]
	\node (Ip) {$\mIp:\qquad\qquad\qquad\quad$};
	\node (u1) [node distance=0.3cm, below of=Ip] {$u_1:a,b$};
	\node (v2) [node distance=2.0cm, below of=u1] {$v_2\!:\!A_{\,0.6}$};
	\node (w1) [node distance=2.0cm, below of=v2] {$w_1\!:\!B_{\,0.8}$};
	\draw (u1) to node [right]{\footnotesize{0.4}} (v2);	
	\draw (v2) to node [right]{\footnotesize{0.4}} (w1);	
    \draw (v2) edge[out=40,in=-40,looseness=6] node[right]{\footnotesize{0.4}} (v2);
    \draw[dashed] (w1) edge[bend left=45] node[left]{\footnotesize{0.2}} (u1);
\end{tikzpicture}
\\
\hline
\end{tabular}
\caption{Fuzzy interpretation used in Example~\ref{example: HGRJS}.\label{fig: HGRJS}}
\end{center}
\end{figure}

\begin{example}\label{example: HGRJS}
Let $\RN=\{r,s\}$, $\CN = \{A,B\}$, $\IN = \{a,b\}$ and $\Phi = \emptyset$. Consider the fuzzy interpretation $\mI$ illustrated in Figure~\ref{fig: HGRJS} and specified below:
\begin{itemize}
\item $\Delta^\mI = \{u_1,u_2,v_1,v_2,v_3,w_1,w_2\}$,\ $a^\mI = u_1$,\ $b^\mI = u_2$, 
\item $A^\mI = \{v_1\!:\!0.5,\ v_2\!:\!0.6,\ v_3\!:\!0.7\}$,\  $B^\mI = \{w_1\!:\!0.8,\ w_2\!:\!0.9\}$,
\item \mbox{$r^\mI = \{\tuple{u_1,v_1}\!:\!0.3$}, \mbox{$\tuple{u_1,v_2}\!:\!0.4$}, \mbox{$\tuple{u_2,v_3}\!:\!0.4$}, \mbox{$\tuple{v_1,v_2}\!:\!0.9$}, \mbox{$\tuple{v_2,v_1}\!:\!0.7$}, \mbox{$\tuple{v_3,v_3}\!:\!0.8$}, \mbox{$\tuple{v_1,w_1}\!:\!0.8$}, \mbox{$\tuple{v_2,w_1}\!:\!0.7$}, \mbox{$\tuple{v_3,w_2}\!:\!0.8\}$}, 
\item $s^\mI = \{\tuple{w_1,u_2}\!:\!0.2,\ \tuple{w_2,u_1}\!:\!0.2\}$.
\end{itemize}
The specification of this fuzzy interpretation is stored in the file {\em in2.txt} of~\cite{min2025-prog}. 
The compact fuzzy partition $\bB$ corresponding to the greatest fuzzy $\Phi$-auto-bisimulation of~$\mI$ is 
\[ \bB = 
\{
  \{u_1, u_2\}_1, 
  \{
    \{v_1\}_1, 
    \{v_2\}_1, 
    \{v_3\}_1
  \}_{0.5},
  \{
    \{w_1\}_1, 
    \{w_2\}_1
  \}_{0.8}
\}_0.
\]
We denote
$B_u = \{u_1, u_2\}_1$, 
$B_v = \{
    \{v_1\}_1, 
    \{v_2\}_1, 
    \{v_3\}_1
  \}_{0.5}$ and
$B_w = \{
    \{w_1\}_1, 
    \{w_2\}_1
  \}_{0.8}$.
Thus, $\bB = \{B_u, B_v, B_w\}_0$. 
Consider the execution of Algorithm~\ref{alg2} for $\mI$ using $\Phi$ and $\gamma = 1$. 
The resulting fuzzy interpretation $\mI'$ is illustrated in Figure~\ref{fig: HGRJS}. 
Details are provided below.

After executing statements \ref{step: alg2 1}-\ref{step: alg2 12}, we have:
\begin{itemize}
\item $\Delta^\mIp = \{u_1\}$, $a^\mIp = b^\mIp = u_1$;
\item for $B \in \bB.\blocks()$, $B.\repr = u_1$ if $B = B_u$ or $B = \bB$, and $B.\repr = \Null$ otherwise;
\item $Q$ contains $\tuple{u_1,r,v_1}$ and $\tuple{u_1,r,v_2}$;
\item $D = \{1, 0.9, 0.8, 0.7, 0.4, 0.3, 0.2\}$. 
\end{itemize}

In the first four iterations of the ``foreach'' loop at statement~\ref{step: alg2 13}, which are executed with $d \in \{1, 0.9, 0.8, 0.7\}$, the inner ``while'' loop terminates immediately.

Consider the fifth iteration of the ``foreach'' loop at statement~\ref{step: alg2 13}, with $d = 0.4$:
\begin{itemize}
\item The first iteration of the inner ``while'' loop is executed with $\tuple{x,R,y} = \tuple{u_1, r, v_2}$, in which:
    \begin{itemize}
    \item since $\bB.\findBlock(v_2, 0.4) = B_v$ and $B_v.\repr = \Null$, $v_2$ is added to $\Delta^\mIp$ and $B_v.\repr$ is set to $v_2$;
    \item the triples $\tuple{v_2,r,v_1}$ and $\tuple{v_2,r,w_1}$ are inserted into $Q$ and $r^\mIp(u_1,v_2)$ is set to $0.4$.
    \end{itemize}
\item The second iteration of the inner ``while'' loop is executed with $\tuple{x,R,y} = \tuple{v_2, r, v_1}$. In this iteration, since $\bB.\findBlock(v_1, 0.4) = B_v$ and $B_v.\repr = v_2$, $r^\mIp(v_2,v_2)$ is set to $0.4$.
\item The third iteration of the inner ``while'' loop is executed with $\tuple{x,R,y} = \tuple{v_2, r, w_1}$, in which:
    \begin{itemize}
    \item since $\bB.\findBlock(w_1, 0.4) = B_w$ and $B_w.\repr = \Null$, $w_1$ is added to $\Delta^\mIp$ and $B_w.\repr$ is set to $w_1$;
    \item the triple $\tuple{w_1,s,u_2}$ is inserted into $Q$ and $r^\mIp(v_2,w_1)$ is set to $0.4$.
    \end{itemize}
\end{itemize}

Consider the sixth iteration of the ``foreach'' loop at statement~\ref{step: alg2 13}, with $d = 0.3$. During this iteration, the inner ``while'' loop executes only one iteration, with $\tuple{x,R,y} = \tuple{u_1,r,v_1}$, in which we have $\bB.\findBlock(v_1, 0.3) = B_v$, with $B_v.\repr = v_2$, and $\mIp$ remains unchanged. 

Consider the seventh iteration of the ``foreach'' loop at statement~\ref{step: alg2 13}, with $d = 0.2$. During this iteration, the inner ``while'' loop executes only one iteration, with $\tuple{x,R,y} = \tuple{w_1,s,u_2}$, in which, since $\bB.\findBlock(u_2, 0.2) = B_u$ and $B_u.\repr = u_1$, $s^\mIp(w_1,u_1)$ is set to $0.2$.

After the seventh iteration, the ``foreach'' loop at statement~\ref{step: alg2 13} terminates and the algorithm returns $\mIp$ with $|\Delta^\mIp| = 3$.
\myend
\end{example}    

\begin{figure}[t]
\begin{center}
\begin{tabular}{|c|c|}
\hline
\begin{tikzpicture}[->,>=stealth,auto,black]
	\node (I) {$\mI:$};
	\node (bI) [node distance=0.3cm, below of=I] {};
	\node (u1) [node distance=1cm, right of=bI] {$u_1\!:\!a$};
	\node (v1) [node distance=1.7cm, below of=u1] {$v_1$};
	\node (w1) [node distance=1.7cm, below of=v1] {$w_1\!:\!A_{\,1}\ \,$};
	\node (u2) [node distance=2cm, right of=u1] {$u_2\!:\!b$};
	\node (v2) [node distance=2cm, right of=v1] {$v_2$};
	\node (w2) [node distance=2cm, right of=w1] {$w_2\!:\!A_{\,0.8}$};
	\draw (u1) to node [left]{\footnotesize{1}} (v1);	
	\draw (v1) to node [left]{\footnotesize{1}} (w1);	
	\draw (u2) to node [right]{\footnotesize{0.9}} (v2);	
	\draw (v2) to node [right]{\footnotesize{0.9}} (w2);	
\end{tikzpicture}
&
\begin{tikzpicture}[->,>=stealth,auto,black]
	\node (I) {$\mIp:$};
	\node (bI) [node distance=0.3cm, below of=I] {};
	\node (u1) [node distance=1cm, right of=bI] {$u_1:a,b$};
	\node (v1) [node distance=1.7cm, below of=u1] {$v_1$};
	\node (w1) [node distance=1.7cm, below of=v1] {$w_1\!:\!A_{\,1}\ \,$};
	\draw (u1) to node [left]{\footnotesize{0.8}} (v1);	
	\draw (v1) to node [left]{\footnotesize{0.8}} (w1);	
\end{tikzpicture}
\\
\hline
\end{tabular}
\caption{Fuzzy interpretation used in Example~\ref{example: YNSJA}.\label{fig: YNSJA}}
\end{center}
\end{figure}
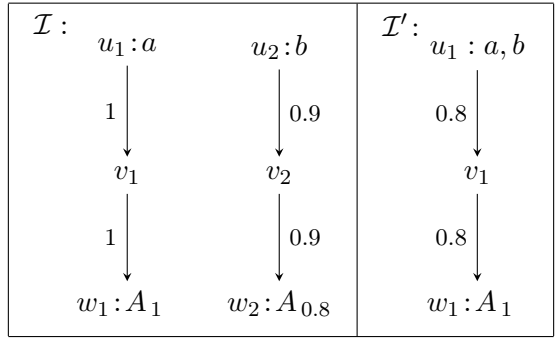

The following example shows that there are cases in which executing Algorithm~\ref{alg2} results in a domain reduction only when $\gamma < 1$.

\begin{example}\label{example: YNSJA}
Let $\RN=\{r\}$, $\CN = \{A\}$, $\IN = \{a,b\}$ and $\Phi = \emptyset$. Consider the fuzzy interpretation $\mI$ illustrated in Figure~\ref{fig: YNSJA} and specified below:
\begin{itemize}
\item $\Delta^\mI = \{u_1,u_2,v_1,v_2,w_1,w_2\}$,\ $a^\mI = u_1$,\ $b^\mI = u_2$,\ $A^\mI = \{w_1\!:\!1,w_2\!:\!0.8\}$,
\item $r^\mI = \{\tuple{u_1,v_1}\!:\!1,\ \tuple{v_1,w_1}\!:\!1,\ \tuple{u_2,v_2}\!:\!0.9,\ \tuple{v_2,w_2}\!:\!0.9\}$.
\end{itemize}
The specification of this fuzzy interpretation is stored in the file {\em in3.txt} of~\cite{min2025-prog}. 
The compact fuzzy partition $\bB$ corresponding to the greatest fuzzy $\Phi$-auto-bisimulation of~$\mI$ is $\{\{\{u_1\}_1,\{u_2\}_1\}_{0.8}$, $\{\{v_1\}_1,\{v_2\}_1\}_{0.8}$, $\{\{w_1\}_1,\{w_2\}_1\}_{0.8}\}_0$. Denote $B_u = \{\{u_1\}_1,\{u_2\}_1\}_{0.8}$, $B_v = \{\{v_1\}_1,\{v_2\}_1\}_{0.8}$ and $B_w = \{\{w_1\}_1,\{w_2\}_1\}_{0.8}$. Thus, $\bB = \{B_u, B_v, B_w\}_0$. 

Executing Algorithm~\ref{alg2} on $\mI$ with $\Phi$ and any $\gamma \in (0.8, 1]$ yields no domain reduction. In particular, when \mbox{$\gamma = 1$} the algorithm returns the same fuzzy interpretation~$\mI$. Consider the execution of Algorithm~\ref{alg2} for $\mI$ using $\Phi$ and $\gamma = 0.8$. The resulting fuzzy interpretation $\mI'$ is illustrated in Figure~\ref{fig: YNSJA}. Details are provided below.

After executing statements \ref{step: alg2 1}-\ref{step: alg2 12}, we have:
\begin{itemize}
\item $\Delta^\mIp = \{u_1\}$, $a^\mIp = b^\mIp = u_1$;
\item for $B \in \bB.\blocks()$, $B.\repr = u$ if $B = B_u$ or $B = \bB$, and $B.\repr = \Null$ otherwise;
\item $D = \{0.8\}$ and $Q$ contains only $\tuple{u_1,r,v_1}$.
\end{itemize}

Consider the only iteration of the ``foreach'' loop at statement~\ref{step: alg2 13}, with $d = 0.8$. 
The first iteration of the inner ``while'' loop is executed with $\tuple{x,R,y} = \tuple{u_1, r, v_1}$, in which:
    \begin{itemize}
    \item since $\bB.\findBlock(v_1, 0.8) = B_v$ and $B_v.\repr = \Null$, $v_1$ is added to $\Delta^\mIp$ and $B_v.\repr$ is set to~$v_1$;
    \item the triple $\tuple{v_1,r,w_1}$ is inserted into $Q$ and $r^\mIp(u_1,v_1)$ is set to $0.8$.
    \end{itemize}
The second iteration of the inner ``while'' loop is executed with $\tuple{x,R,y} = \tuple{v_1, r, w_1}$, in which:
    \begin{itemize}
    \item since $\bB.\findBlock(w_1, 0.8) = B_w$ and $B_w.\repr = \Null$, $w_1$ is added to $\Delta^\mIp$ and $B_w.\repr$ is set to $w_1$;
    \item $r^\mIp(v_1,w_1)$ is set to $0.8$.
    \end{itemize}
The fuzzy interpretation $\mIp$ returned by Algorithm~\ref{alg2} has $|\Delta^\mIp| = 3$.
\myend
\end{example}

{\markRed
Further examples, which involve inverse roles and/or nominals (i.e., the cases where $\emptyset \subset \Phi \subseteq \{I, O\}$), are presented in~\ref{appendix A}. All of the above examples (together with the latter) are intended to aid in understanding Algorithm~\ref{alg2}. Below, we provide an additional high-level intuition for this algorithm.

Let $\mI$ be an input finite fuzzy interpretation, $Z_0$ the greatest fuzzy $\Phi$-auto-bisimulation of~$\mI$, and $\bB$ the compact fuzzy partition corresponding to~$Z_0$.  
Recall that, for $x,x' \in \Delta^\mI$ and $d \in [0,1]$, $Z_0(x,x') \geq d$ iff there exists $B \in \bB.\blocks()$ such that $B.\degree \geq d$ and $x,x' \in B.\elements()$. Consequently, if $B = \bB.\findBlock(x,d)$, then $B.\elements()$ is the subset of $\Delta^\mI$ consisting of all $x'$ such that $Z_0(x,x') \geq d$. 

Algorithm~\ref{alg2} constructs a reduced fuzzy interpretation $\mIp$ by selecting certain individuals from $\Delta^\mI$ to include in $\Delta^\mIp$, and by setting $A^\mIp(x) = A^\mI(x)$ for all $A \in \CN$ and all $x \in \Delta^\mIp$. To ensure that $\gamma \le (C^\mI(a^\mI) \fequiv C^\mIp(a^\mIp))$ holds for every concept $C$ of $\mLP$ and every $a \in \IN$, it suffices -- by Theorem~\ref{theorem: JHDHJ} -- to guarantee the existence of a fuzzy $\Phi$-bisimulation $Z$ between $\mI$ and $\mIp$ such that $\gamma \le Z(a^\mI, a^\mIp)$. Such a fuzzy relation $Z$ will be explicitly defined in the proof of correctness of Algorithm~\ref{alg2} (Lemma~\ref{lemma: ISJNA}), based on $Z_0$. Below we provide the intuition behind the existence of such a fuzzy $\Phi$-bisimulation.

In the first phase (statements~\ref{step: alg2 1}--\ref{step: alg2 11}), Algorithm~\ref{alg2} computes $\bB$ and sets $\Delta^\mIp$ to a maximal subset of $\{a^\mI \mid a \in \IN\}$ such that, whenever $a^\mI, b^\mI \in \Delta^\mIp$ and $a^\mI \neq b^\mI$, we have $Z_0(a^\mI,b^\mI) < \gamma$. Thus, if $Z_0(a^\mI,b^\mI) \ge \gamma$, then at most one of $a^\mI$ and $b^\mI$ belongs to $\Delta^\mIp$. 
One of the goals of the next phase (statements \ref{step: alg2 12m}--\ref{step: alg2 24}) is to ensure the following: for every fuzzy value $d \in (0,\gamma]$ used in $\mI$, for every $x \in \Delta^\mIp$ and $y \in \Delta^\mI$, and for every basic role $R$ w.r.t.\ $\Phi$, if $R^\mI(x,y) \geq d$, then there exists $y' \in \Delta^\mIp$ such that $Z_0(y,y') \geq d$ and $R^\mIp(x,y') \geq d$. This corresponds to the condition~\eqref{eq: FB 3} (with $x' = x$). Given $d$, $x$, $y$ and $R$ as mentioned, with $B = \bB.\findBlock(y,d)$, $B.\repr$ (the representative of $B$) is set to $y$ if it is still $\Null$, and then with $y' = B.\repr$, $R^\mIp(x,y')$ is set to $d$ if it is still 0 (i.e., has not been set earlier). With such a chosen $y'$, we have $Z_0(y,y') \geq d$. 

The main loop of the algorithm processes the possible fuzzy values $d$ in decreasing order, which brings several advantages. First, if the mentioned $R^\mIp(x,y')$ was set earlier, let's say, to $d'$, then $d' \geq d$ and we have $R^\mIp(x,y') \geq d$ as expected. Second, if $R^\mIp(x,y')$ is set to $d$ in the step described above, then this not only ensures the condition~\eqref{eq: FB 3} (with $x' = x$), but also increases the likelihood of satisfying the condition~\eqref{eq: FB 4} (with $x'$ and $x$ replaced by $x$ and some other individual, respectively), due to the decreasing order of $d$. 
Third, and most importantly, setting $B.\repr = y$ for $B = \bB.\findBlock(y,d)$ allows $y$ to serve as a representative not only of $B$ but, by propagation (statements~\ref{step: alg2 10} and~\ref{step: alg2 20}), also of all individuals that are ``equivalent'' to $y$ at any level below $d$, provided they did not already have representatives. Processing the fuzzy values $d$ in decreasing order increases the number of elements of $\Delta^\mI$ that those added to $\Delta^\mIp$ can represent. This is how the algorithm minimizes $\Delta^\mIp$ while ensuring that $\mIp$ satisfies the required properties.
}

The following lemma provides some detailed properties of Algorithm~\ref{alg2}.

\begin{lemma}\label{lemma: JHJWS}
Consider Algorithm~\ref{alg2}. It holds that 
\begin{enumerate}[(a)]
\item\label{JHJWS1} for every $B \in \bB.\blocks()$, if $B.\repr \neq \Null$, then $B.\repr \in B.\elements()$;
\item\label{JHJWS2} $\Delta^\mIp = \{B.\repr \mid B \in \bB.\blocks(), B.\repr \neq \Null\} \subseteq \Delta^\mI$;
\item\label{JHJWS3} for every $B' \in \bB.\blocks()$ and every $B \in B'.\blocks()$, if $B'.\repr = \Null$, then $B.\repr = \Null$.
\end{enumerate}
\end{lemma}

\begin{proof}
Consider the assertion~\eqref{JHJWS1} and suppose that $B \in \bB.\blocks()$ and $B.\repr \neq \Null$. If $B.\repr$ is set at the first iteration of the ``repeat'' loop at statement~\ref{step: alg2 10}, then $B = \bB.\findBlock(a^\mI, \gamma)$ (from statement~\ref{step: alg2 6}) and hence $B.\repr = a^\mI \in B.\elements()$. Similarly, if $B.\repr$ is set at the first iteration of the ``repeat'' loop at statement~\ref{step: alg2 20}, then $B = \bB.\findBlock(y, d)$ (from statement~\ref{step: alg2 16}) and hence $B.\repr = y \in B.\elements()$. If $B.\repr$ is set in a later (i.e., not the first) iteration of either of those two loops, then the assertion $B.\repr \in B.\elements()$ follows from the corresponding assertion for the previous iteration. 

The assertion~\eqref{JHJWS2} clearly holds. 
The assertion~\eqref{JHJWS3} holds due to statements~\ref{step: alg2 10} and~\ref{step: alg2 20} of Algorithm~\ref{alg2}.
\myend
\end{proof}

The following lemma states that Algorithm~\ref{alg2} returns a reduction of the input fuzzy interpretation $\mI$ that preserves fuzzy concept assertions of $\mLP$ up to degree~$\gamma$.

\begin{lemma}\label{lemma: ISJNA}
Let $\mI$ be a finite fuzzy interpretation and let $\Phi \subseteq \{I,O\}$ and $\gamma \in (0,1]$. Let $\mIp$ be the fuzzy interpretation returned by Algorithm~\ref{alg2} when executed for $\mI$ using $\Phi$ and $\gamma$. Then, for every concept $C$ of $\mLP$ and every $a \in \IN$, $\gamma \leq (C^\mI(a^\mI) \fequiv C^\mIp(a^\mIp))$.
\end{lemma}

\begin{proof}
%
Let $Z_0$ be the greatest fuzzy $\Phi$-auto-bisimulation of $\mI$ and let $Z \in \mF(\Delta^\mI \times \Delta^\mIp)$ be the fuzzy relation constructed during the execution of Algorithm~\ref{alg2} as follows. Initially, $Z$ is initialized to the empty fuzzy relation, and non-zero pairs are subsequently added to $Z$ in the following manner. After executing statement~\ref{step: alg2 8}, define $Z(v,a^\mIp) = \min(\gamma, Z_0(v,a^\mIp))$ for all $v \in \Delta^\mI$. Similarly, after executing statement~\ref{step: alg2 18}, set $Z(v,y) = \min(d, Z_0(v,y))$ for all $v \in \Delta^\mI$. 
We first prove that $Z$ is a fuzzy $\Phi$-bisimulation between~$\mI$ and~$\mIp$. 

Observe that the condition~\eqref{eq: FB 2} holds due to $Z_0$ and the setting of $A^\mIp$, for $A \in \CN$. 

Consider the condition~\eqref{eq: FB 3} in the following form
\[ \V y_0 \in \Delta^\mI\ \E y' \in \Delta^\mIp\, \left(Z(x_0,x) \fand R^\mI(x_0,y_0) \leq Z(y_0,y') \fand R^\mIp(x,y')\right), \]
with $x_0 \in \Delta^\mI$, $x \in \Delta^\mIp$ and $R$ being a basic role w.r.t.~$\Phi$. Let $y_0 \in \Delta^\mI$, \mbox{$d_0 = Z(x_0,x) \fand R^\mI(x_0,y_0)$} and suppose $d_0 > 0$. We need to show that there exists $y' \in \Delta^\mIp$ such that 
\begin{equation}\label{eq: JRNZW}
d_0 \leq Z(y_0,y') \fand R^\mIp(x,y').
\end{equation}
Since $Z_0(x_0,x) \geq Z(x_0,x) \geq d_0$, we have \mbox{$d_0 \leq Z_0(x_0,x) \fand R^\mI(x_0,y_0)$}. By~\eqref{eq: FB 3} with $\mIp$, $Z$, $x$, $x'$ and $y$ replaced by $\mI$, $Z_0$, $x_0$, $x$ and $y_0$, respectively, it follows that there exists $y \in \Delta^\mI$ such that 
\begin{equation}\label{eq: HFKEX}
d_0 \leq Z_0(y_0,y) \fand R^\mI(x,y).
\end{equation}
Consider the iteration of the ``while'' loop at statement~\ref{step: alg2 14} for $\tuple{x,R,y}$. Let $d$, $B$ and $y'$ be the involved variables in that iteration. Since $Z(x_0,x) \geq d_0$ and $R^\mI(x,y) \geq d_0$, we must have $d \geq d_0$. 
In addition, $B = \bB.\findBlock(y,d)$ and $y' = B.\repr$. After that iteration, we have $R^\mIp(x, y') \geq d \geq d_0$. In addition, by Lemma~\ref{lemma: JHJWS}\eqref{JHJWS1}, $B.\repr \in B.\elements()$ and hence $Z_0(y,y') \geq d \geq d_0$. By~\eqref{eq: HFKEX}, $Z_0(y_0,y) \geq d_0$. Hence, $Z_0(y_0,y') \geq d_0$. As $y'$ has been added to $\Delta^\mIp$ no later than during the mentioned iteration of the ``while'' loop, which is involved with $d \geq d_0$, it follows that $Z(y_0,y') \geq d_0$. This completes the proof of~\eqref{eq: JRNZW}.

We prove~\eqref{eq: FB 4} in the following forms
\begin{eqnarray}
&& \E y_0 \in \Delta^\mI\, \left(Z(x_0,x) \fand R^\mIp(x,y') \leq Z(y_0,y') \fand R^\mI(x_0,y_0)\right), \label{eq: IUENA 1}\\
&& \E x_0 \in \Delta^\mI\, \left(Z(y_0,y') \fand R^\mIp(x,y') \leq Z(x_0,x) \fand R^\mI(x_0,y_0)\right), \label{eq: IUENA 2}
\end{eqnarray}
where $x_0$ in~\eqref{eq: IUENA 1} and $y_0$ in \eqref{eq: IUENA 2} are elements of $\Delta^\mI$, $x,y' \in \Delta^\mIp$, $R$ is a basic role w.r.t.~$\Phi$, $I \in \Phi$ in the case of \eqref{eq: IUENA 2}, and $R^\mIp(x,y')$ is assigned a value by statement~\ref{step: alg2 23} or~\ref{step: alg2 24}.

Consider~\eqref{eq: IUENA 1}. 
Let $d_0 = Z(x_0,x) \fand R^\mIp(x,y')$ and suppose $d_0 > 0$. We need to show that there exists $y_0 \in \Delta^\mI$ such that
\begin{equation}\label{eq: JDKKA}
d_0 \leq Z(y_0,y') \fand R^\mI(x_0,y_0).
\end{equation}
Since $Z(x_0,x) \geq d_0$, we have $Z_0(x_0,x) \geq d_0$. 
Let $d = R^\mIp(x,y')$. Thus, $d \geq d_0$. Consider the iteration of the ``while'' loop in which $R^\mIp(x,y')$ is set to $d$ by statement~\ref{step: alg2 23} or~\ref{step: alg2 24}. Let $y \in \Delta^\mI$ be the involved element in that iteration. We have $y' = B.\repr$, with $B = \bB.\findBlock(y,d)$, and $R^\mI(x,y) \geq d$. 
Hence, $Z_0(x_0,x) \fand R^\mI(x,y) \geq d_0$. By~\eqref{eq: FB 4} with $\mIp$, $Z$, $x$ and $x'$ replaced by $\mI$, $Z_0$, $x_0$ and $x$, respectively, it follows that there exists $y_0 \in \Delta^\mI$ such that 
\begin{equation}\label{eq: JHFJS}
d_0 \leq Z_0(y_0,y) \fand R^\mI(x_0, y_0).
\end{equation}
Thus, $d_0 \leq Z_0(y_0,y)$. In addition, since $y' = \bB.\findBlock(y,d).\repr$, we have $Z_0(y,y') \geq d \geq d_0$. Therefore, $Z_0(y_0,y') \geq d_0$. 
Observe that $y'$ must have been added to $\Delta^\mIp$ no later than during the mentioned iteration of the ``while'' loop, which is involved with $d \geq d_0$. As a consequence, we have $Z(y_0,y') \geq d_0$. This together with~\eqref{eq: JHFJS} implies~\eqref{eq: JDKKA}.

Consider~\eqref{eq: IUENA 2} under the assumption $I \in \Phi$. 
Let $d_0 = Z(y_0,y') \fand R^\mIp(x,y')$ and suppose $d_0 > 0$. We need to show that there exists $x_0 \in \Delta^\mI$ such that
\begin{equation}\label{eq: JDKKA 2}
d_0 \leq Z(x_0,x) \fand R^\mI(x_0,y_0).
\end{equation}
Since $Z(y_0,y') \geq d_0$, we have $Z_0(y_0,y') \geq d_0$. 
Let $d = R^\mIp(x,y')$. Thus, $d \geq d_0$. Consider the iteration of the ``while'' loop in which $R^\mIp(x,y')$ is set to $d$ by statement~\ref{step: alg2 23} or~\ref{step: alg2 24}. Let $y \in \Delta^\mI$ be the involved element in that iteration. We have $y' = B.\repr$, with $B = \bB.\findBlock(y,d)$, and $R^\mI(x,y) \geq d$. 
By Lemma~\ref{lemma: JHJWS}\eqref{JHJWS1}, $B.\repr \in B.\elements()$ and hence $Z_0(y,y') \geq d \geq d_0$. Since $Z_0(y_0,y') \geq d_0$, it follows that $Z_0(y_0,y) \geq d_0$, and consequently, $Z_0(y_0,y) \fand (\cnv{R})^\mI(y,x) \geq d_0$. By~\eqref{eq: FB 4} with $\mIp$, $Z$, $x$, $x'$, $y$, $y'$ and $R$ replaced by $\mI$, $Z_0$, $y_0$, $y$, $x_0$, $x$ and $\cnv{R}$, respectively, it follows that there exists $x_0 \in \Delta^\mI$ such that 
\begin{equation}\label{eq: JHFJS 2}
d_0 \leq Z_0(x_0,x) \fand (\cnv{R})^\mI(y_0, x_0).
\end{equation}
Thus, $Z_0(x_0,x) \geq d_0$. Observe that $x$ must have been added to $\Delta^\mIp$ no later than during the mentioned iteration of the ``while'' loop, which is involved with $d \geq d_0$. As a consequence, we have $Z(x_0,x) \geq d_0$. This together with~\eqref{eq: JHFJS 2} implies~\eqref{eq: JDKKA 2}. 

{\markRed
Consider the condition~\eqref{eq: FB 5} when $O \in \Phi$, with $a \in \IN$. If either ($x = a^\mI$ and $x' = a^\mIp$) or ($x \neq a^\mI$ and $x' \neq a^\mIp$) holds, then \eqref{eq: FB 5} clearly holds. There remain two other cases. 

Consider the case where $x = a^\mI$ and $x' \neq a^\mIp$. Suppose $x' = a^\mI$. Thus, $a^\mIp \neq a^\mI$, and we must have $a^\mIp = \bB.\findBlock(a^\mI,\gamma).\repr$, which implies $Z_0(a^\mI,a^\mIp) \geq \gamma > 0$. By~\eqref{eq: FB 5} with $\mIp$ and $Z$ replaced by $\mI$ and $Z_0$, respectively, it follows that $a^\mIp = a^\mI$, a contradiction. Hence, we must have $x' \neq a^\mI$. By~\eqref{eq: FB 5} with $\mIp$ and $Z$ replaced by $\mI$ and $Z_0$, respectively, it follows that $Z_0(x,x') = 0$. Therefore, $Z(x,x') = 0$ and \eqref{eq: FB 5} holds. 

Consider the case where $x \neq a^\mI$ and $x' = a^\mIp$. If $x' \neq a^\mI$, then $a^\mIp \neq a^\mI$, and as for the above case, this leads to a contradiction. Hence, $x' = a^\mI$. By~\eqref{eq: FB 5} with $\mIp$ and $Z$ replaced by $\mI$ and $Z_0$, respectively, it follows that $Z_0(x,x') = 0$. Therefore, $Z(x,x') = 0$ and \eqref{eq: FB 5} holds.
}

We have proved that $Z$ is a fuzzy $\Phi$-bisimulation between~$\mI$ and~$\mIp$. 

{\markRed
We now show that $Z(a^\mI, a^\mIp) = \gamma$ for all $a \in \IN$. If $a^\mIp = a^\mI$, then $Z_0(a^\mI, a^\mIp) = 1$ and hence $Z(a^\mI, a^\mIp) = \gamma$. 
Consider the case $a^\mIp \neq a^\mI$ and let $B = \bB.\findBlock(a^\mI, \gamma)$. We have $a^\mIp = B.\repr$. By Lemma~\ref{lemma: JHJWS}\eqref{JHJWS1}, $a^\mIp \in B.\elements()$. Hence, $Z_0(a^\mI, a^\mIp) \geq \gamma$, and consequently, $Z(a^\mI, a^\mIp) = \gamma$.

Therefore, by Theorem~\ref{theorem: JHDHJ}, for every concept $C$ of $\mLP$ and every $a \in \IN$, 
\[
    \gamma = Z(a^\mI, a^\mIp) \leq \left(C^\mI(a^\mI) \fequiv C^\mIp(a^\mIp)\right),
\]
which completes the proof. 
}
\myend
\end{proof}

To proceed, we need the following lemma, which is independent of Algorithm~\ref{alg2}.

\begin{lemma}\label{lemma: JHJHS}
Let $\Phi \subseteq \{I,O\}$, $\mI$ be a finite fuzzy interpretation and $Z_0$ the greatest fuzzy $\Phi$-auto-bisimulation of $\mI$. For every $x \in \Delta^\mI$ and $d \in (0,1]$, there exists a concept $C$ of $\mLPp$ such that $C^\mI(x) = 1$ and, for every $x' \in \Delta^\mI$ with $Z_0(x,x') < d$, $C^\mI(x') < d$.
\end{lemma}

\begin{proof}
Consider an arbitrary $x' \in \Delta^\mI$ with $Z_0(x,x') < d$. 
As the G\"odel semantics is used, by Theorem~\ref{theorem: fG H-M}, there exists a concept $D_{x'}$ of $\mLPp$ such that $Z_0(x,x') = (D_{x'}^\mI(x) \fequiv D_{x'}^\mI(x'))$. Thus, $D_{x'}^\mI(x) \neq D_{x'}^\mI(x')$ and $\min(D_{x'}^\mI(x), D_{x'}^\mI(x')) = Z_0(x,x') < d$. If $D_{x'}^\mI(x) < D_{x'}^\mI(x')$, then let $C_{x'} = (D_{x'} \to D_{x'}^\mI(x))$, else let $C_{x'} = (D_{x'}^\mI(x) \to D_{x'})$. Thus, $C_{x'}^\mI(x) = 1$ and $C_{x'}^\mI(x') < d$. 
Let $C = \bigsqcap \{C_{x'} \mid x' \in \Delta^\mI, Z_0(x,x') < d\}$. It is a concept of~$\mLPp$. We have $C^\mI(x) = 1$ and, for every $x' \in \Delta^\mI$ with $Z_0(x,x') < d$, $C^\mI(x') < d$.
\myend
\end{proof}

The following lemma states that the fuzzy interpretation $\mIp$ returned by Algorithm~\ref{alg2} is in fact minimal.

\begin{lemma}\label{lemma: JRJSL}
Let $\mI$ be a finite fuzzy interpretation and let $\Phi \subseteq \{I,O\}$ and $\gamma \in (0,1]$. Let $\mIp$ be the fuzzy interpretation returned by Algorithm~\ref{alg2} when executed for $\mI$ using $\Phi$ and $\gamma$. Furthermore, let $\mIdp$ be an arbitrary finite fuzzy interpretation such that, for every concept $C$ of $\mLPp$ and every $a \in \IN$, $\gamma \leq (C^\mI(a^\mI) \fequiv C^\mIdp(a^\mIdp))$. Then, $|\Delta^\mIdp| \geq |\Delta^\mIp|$. 
\end{lemma}

\begin{proof}
During the execution of Algorithm~\ref{alg2}, when an element $v$ is added to $\Delta^\mIp$, we define $d_v$, $\pi_v$ and $R_v$ as follows. If $v$ is $a^\mI$, which is added to $\Delta^\mIp$ at statement~\ref{step: alg2 8}, then let $d_v = \gamma$, $\pi_v = \Null$ and $R_v = \Null$. If $v$ is $y$, which is added to $\Delta^\mIp$ at statement~\ref{step: alg2 18}, then let $d_v = d$, $\pi_v = x$ and $R_v = R$. Roughly speaking, $\pi_v$ denotes the ``predecessor'' of $v$ in $\mIp$, $R_v$ the label of the ``edge'' connecting that ``predecessor'' to $v$, and $d_v$ a lower bound for the degree of that ``edge''. 

In the following, let $v$ denote an arbitrary element of $\Delta^\mIp$ and let $v_0,v_1,\ldots,v_n$ be the involved elements of $\Delta^\mIp$ such that $v_n = v$, $v_i = \pi_{v_{i+1}}$ for each $i$ from $n-1$ down to $0$, and $\pi_{v_0} = \Null$. Note that
\begin{equation}\label{eq: JHJHE}
R_{v_i}^\mIp(v_{i-1},v_i) \geq d_{v_i}\ \textrm{for}\ 1 \leq i \leq n,\ \ \textrm{and}\ \ \gamma \geq d_{v_1} \geq \ldots \geq d_{v_n} = d_v.
\end{equation} 

Let $Z_0$ be the greatest fuzzy $\Phi$-auto-bisimulation of $\mI$. By Lemma~\ref{lemma: JHJHS}, there exists a concept $C_v$ of $\mLPp$ such that $C_v^\mI(v) = 1$ and, for every $v' \in \Delta^\mI$ with $Z_0(v,v') < d_v$, $C_v^\mI(v') < d_v$. Let $d_v^-$ denote the arithmetic average of $d_v$ and $\sup\{C_v^\mI(v') \mid v' \in \Delta^\mI$, $Z_0(v,v') < d_v\}$. Thus, $C_v^\mI(v') < d_v^- < d_v$, for all $v' \in \Delta^\mI$ with $Z_0(v,v') < d_v$.  Let $D_v$ be the concept specified as follows:
\[ D_v = C_v \mand \textstyle\bigsqcap\{(C_{v'} \to d_{v'}^-) \mid v' \in \Delta^\mIp \setminus \{v\}\}. \]
Furthermore, let $E_v = \E R_{v_1}.\cdots \E R_{v_n}.D_v$, which is a concept of $\mLPp$, and let $a_v$ be an individual name from $\IN$ such that $a_v^\mIp = v_0$.

By Lemma~\ref{lemma: JHJWS}\eqref{JHJWS3} and due to the decreasing order in which the values from $D$ are considered by the ``foreach'' loop at statement~\ref{step: alg2 13}, $\bB.\findBlock(v,d_v) \cap \Delta^\mIp = \{v\}$. Consequently, for any $v' \in \Delta^\mIp$ different from $v$, we have $Z_0(v,v') < \min(d_v, d_{v'})$, which implies that $D_v^\mI(v) = 1$. 
Let $Z$ be the fuzzy $\Phi$-bisimulation between~$\mI$ and~$\mIp$ constructed as in the proof of Lemma~\ref{lemma: ISJNA}. Recall the construction of $Z$ to see that $Z(v,v) = d_v$. Since $D_v^\mI(v) = 1$, by Theorem~\ref{theorem: JHDHJ}, it follows that $D_v^\mIp(v) \geq d_v$. Consequently, by~\eqref{eq: JHJHE}, we have 
\[
    E_v^\mIp(a_v^\mIp) = (\E R_{v_1}.\cdots \E R_{v_n}.D_v)^\mIp(v_0) \geq d_v.
\]
By Lemma~\ref{lemma: ISJNA}, $\gamma \leq (E_v^\mI(a_v^\mI) \fequiv E_v^\mIp(a_v^\mIp))$. Consequently, it follows that $E_v^\mI(a_v^\mI) \geq d_v$. Similarly, using the assumption $\gamma \leq (E_v^\mI(a_v^\mI) \fequiv E_v^\mIdp(a_v^\mIdp))$, we can further derive that $E_v^\mIdp(a_v^\mIdp) \geq d_v$. Therefore, there exists $x_v \in \Delta^\mIdp$ such that $D_v^\mIdp(x_v) \geq d_v$. 

For a contradiction, assume that there exist distinct $u,v \in \Delta^\mIp$ such that $x_u = x_v$. For simplicity, denote $x = x_u = x_v$. Thus, $D_v^\mIdp(x) \geq d_v$ and $D_u^\mIdp(x) \geq d_u$. Consequently, 
\[ C_v^\mIdp(x) \mand (C_u^\mIdp(x) \to d_u^-) \geq d_v\ \ \textrm{and}\ \ C_u^\mIdp(x) \mand (C_v^\mIdp(x) \to d_v^-) \geq d_u. \]
This implies that 
\[ C_v^\mIdp(x) \geq d_v > d_v^- \geq d_u \ \ \textrm{and}\ \ C_u^\mIdp(x) \geq d_u > d_u^- \geq d_v, \]
which gives a contradiction. Therefore, the mapping that associates each $v \in \Delta^\mIp$ with $x_v \in \Delta^\mIdp$ is an injection. We conclude that \mbox{$|\Delta^\mIdp| \geq |\Delta^\mIp|$}. 
\myend
\end{proof}

We reach the following theorem, according to which Algorithm~\ref{alg2} produces a finite fuzzy interpretation $\mIp$ such that, for any language $\mL$ between $\mLPp$ and $\mLP$, $\mIp$ is a minimal reduction of $\mI$ that preserves fuzzy concept assertions of~$\mL$ up to degree~$\gamma$. 

\begin{theorem}\label{theorem: JHDJW}
Algorithm~\ref{alg2} is correct. That is, given as inputs a finite fuzzy interpretation $\mI$, $\Phi \subseteq \{I,O\}$ and $\gamma \in (0,1]$, Algorithm~\ref{alg2} returns a finite fuzzy interpretation $\mIp$ that satisfies the following properties:
\begin{itemize}
\item[(a)] for every concept $C$ of $\mLP$ and every $a \in \IN$, $\gamma \leq (C^\mI(a^\mI) \fequiv C^\mIp(a^\mIp))$;
\item[(b)] for every fuzzy interpretation $\mIdp$ that satisfies $\gamma \leq (C^\mI(a^\mI) \fequiv C^\mIdp(a^\mIdp))$ for all concepts $C$ of $\mLPp$ and all $a \in \IN$, it holds that $|\Delta^\mIdp| \geq |\Delta^\mIp|$. 
\end{itemize}
In addition, 
\begin{itemize}
\item[(c)] there exists a fuzzy $\Phi$-bisimulation $Z$ between $\mI$ and $\mIp$ such that $Z(a^\mI,a^\mIp) = \gamma$ for all $a \in \IN$.
\end{itemize}
\end{theorem}

The first two assertions follow immediately from Lemmas~\ref{lemma: ISJNA} and~\ref{lemma: JRJSL}, respectively, while the additional one follows from the proof of Lemma~\ref{lemma: ISJNA}. 

The following corollary concerns a particular case of the above theorem. 

\begin{corollary}
Algorithm~\ref{alg2}, when executed for a finite fuzzy interpretation $\mI$ with $\Phi \subseteq \{I,O\}$ and $\gamma = 1$, returns a minimal reduction of $\mI$ that preserves fuzzy concept assertions of $\mL$, for any $\mLPp \subseteq \mL \subseteq \mLP$.
\end{corollary}

We can consider the preservation of other fuzzy assertions when $O \in \Phi$. 
A {\em minimal reduction of $\mI$ that preserves fuzzy ABoxes in~$\mLP$} is a finite fuzzy interpretation $\mIp$ such that:
\begin{itemize}
\item for every ABox $\mA$ in $\mLP$, $\mI \models \mA$ iff $\mIp \models \mA$;
\item if $\mIdp$ is a finite fuzzy interpretation such that, for every ABox $\mA$ in $\mLP$, $\mI \models \mA$ iff $\mIdp \models \mA$, then $|\Delta^\mIdp| \geq |\Delta^\mIp|$. 
\end{itemize}

\begin{corollary}
Algorithm~\ref{alg2}, when executed for a finite fuzzy interpretation $\mI$ with $\{O\} \subseteq \Phi \subseteq \{I,O\}$ and $\gamma = 1$, returns a minimal reduction of $\mI$ that preserves fuzzy ABoxes in~$\mLP$.
\end{corollary}

This corollary follows from Theorems~\ref{theorem: JHDJW} and~\ref{theorem: IFDMS}.

We now analyze the time complexity of Algorithm~\ref{alg2}. 

\begin{remark}\label{remark: JDKAA}
Before proceeding, we first discuss how to efficiently implement statements~\ref{step: alg2 6} and~\ref{step: alg2 16} of Algorithm~\ref{alg2} in a way that does not interfere with the continuation of its execution. The idea is that, once $B$ has been computed as specified in those statements, its internal structure is no longer important due to the decreasing order in which the values from $D$ are processed by the ``foreach'' loop in statement~\ref{step: alg2 13}. Consequently, such blocks $B$ can be flattened, and a union-find data structure can be employed to reduce the overall cost of all merges, as done in~\cite{min2025-prog}. The expected effect of flattening is to retain $B.\degree$ and $B.\repr$, while replacing $B.\subblocks$ (when $B$ is a fuzzy block) with the crisp set $B.\elements()$. This replacement can be performed by recursively flattening each block in $B.\subblocks$ and then merging the results. Using the union-find technique, the total cost of all merges is $O(n \log n)$, where $n = |\Delta^\mI|$. Since the merges are performed in the background, each call to $\bB.\findBlock(a^\mI, \gamma)$ at statement~\ref{step: alg2 6} or to $\bB.\findBlock(y,d)$ at statement~\ref{step: alg2 16} incurs only constant additional cost.
\myend
\end{remark}

\begin{theorem}\label{theorem: JHRKN}
Algorithm~\ref{alg2} can be implemented to run in \mbox{$O((m\log{l} + n)\log{n})$} time, where $n = |\Delta^\mI|$, $m = |\{\tuple{x,r,y} \mid r \in \RN$, $x,y \in \Delta^\mI$, $r^\mI(x,y) > 0\}|$ and $l = |\{r^\mI(x,y) \mid r \in \RN$, $x,y \in \Delta^\mI\}| + 2$, under the assumption that the sizes of $\RN$, $\CN$, and $\IN$ are constants. 
\end{theorem}

\begin{proof}
By Lemma~\ref{lemma: JHJLS}, statement~\ref{step: alg2 1} runs in \mbox{$O((m\log{l} + n)\log{n})$} time.
Since $|\bB.\blocks()| < 2n$, statements~\ref{step: alg2 2}-\ref{step: alg2 11} execute in $O(n)$ time.
Statements~\ref{step: alg2 12m} and~\ref{step: alg2 12} run in $O(n + m\log{n})$ time.

The ``while'' loop at statement~\ref{step: alg2 14} performs at most $2m$ iterations. By Remark~\ref{remark: JDKAA}, the total time spent on all executions of statement~\ref{step: alg2 16} is $O(m + n\log{n})$. The ``repeat'' loop at statement~\ref{step: alg2 20} executes fewer than $2n$ times in total, with a cumulative cost of $O(n)$. The total number of triples inserted into the priority queue $Q$ is bounded by $m$, so all executions of statement~\ref{step: alg2 insertions into Q} take $O(n + m\log{n})$ time. Apart from statements~\ref{step: alg2 16}, \ref{step: alg2 20}, and~\ref{step: alg2 insertions into Q}, all other operations inside the ``while'' loop execute in $O(\log{n})$ time per iteration.

Summing up, the total time required to execute the ``foreach'' loop at statement~\ref{step: alg2 13} is \mbox{$O((m + n)\log{n})$}, and the overall running time of Algorithm~\ref{alg2} is \mbox{$O((m\log{l} + n)\log{n})$}.
\myend
\end{proof}

\section{Experimental results}
\label{section: experiments}

As mentioned earlier, our implementation of Algorithm~\ref{alg2} in Python is publicly available~\cite{min2025-prog}. The main module of the implementation~\cite{min2025-prog} is named {\em fuzzyMinG.py}, and it can be executed as follows:
\begin{verbatim}
    python3 fuzzyMinG.py [options] < input_file > output_file
\end{verbatim}
where {\em input\_file} specifies the input fuzzy interpretation, using the format described in~\cite{min2025-prog}. The available options are:
\begin{itemize}
\item {\em gamma=value} : specifies the value of $\gamma \in (0,1]$ (default: 1),
\item {\em withI} : includes $I$ in~$\Phi$,
\item {\em withO} : includes $O$ in~$\Phi$,
\item {\em verbose} : enables verbose mode.
\end{itemize}
Our implementation can be used to check all the examples given in Section~\ref{section: main} and~\ref{appendix A}. Here are examples:
\begin{verbatim}
    python3 fuzzyMinG.py verbose withI < in2.txt > out2-I.txt
    python3 fuzzyMinG.py verbose gamma=0.8 withO < in3.txt > out3-0.8-O.txt
\end{verbatim}

\begin{table}[t!]
\begin{center}
\footnotesize
\begin{tabular}{|r|l||r|r||r|r||r|r|r|}
\hline
\# & Parameters for {\em runExperiment} & \multicolumn{2}{|c||}{$\mLPU$} & \multicolumn{2}{|c||}{$\mLPD$} & \multicolumn{3}{|c|}{$\mLP$} \\
\cline{3-4} 
\cline{5-6}
\cline{7-9}
   &                              & $n_3$ & $m_3$ & $n_2$ & $m_2$ & $n_1$ & $m_1$ & Red. \\
\hline\hline
$\!$1$\!$&$\!$100 500 $10^3$ 10 20 3 3 3 1 0 0$\!$&$\!$25812$\!$&$\!$80918$\!$&$\!$7867$\!$&$\!$22028$\!$&$\!$6808$\!$&$\!$20099$\!$&$\!$86\%$\!$\\
$\!$2$\!$&$\!$100 500 $10^3$ 10 20 3 3 3 1 0 1$\!$&$\!$25907$\!$&$\!$81190$\!$&$\!$7952$\!$&$\!$22049$\!$&$\!$6888$\!$&$\!$20115$\!$&$\!$86\%$\!$\\
$\!$3$\!$&$\!$100 500 2000 20 40 3 3 3 1 0 1$\!$&$\!$38245$\!$&$\!$188193$\!$&$\!$30472$\!$&$\!$137579$\!$&$\!$29685$\!$&$\!$136104$\!$&$\!$41\%$\!$\\
$\!$4$\!$&$\!$100 500 2000 20 40 3 3 3 1 0 0$\!$&$\!$38373$\!$&$\!$188454$\!$&$\!$30768$\!$&$\!$138813$\!$&$\!$29989$\!$&$\!$137339$\!$&$\!$40\%$\!$\\
$\!$5$\!$&$\!$100 500 $10^3$ 10 20 3 3 3 0 0 0$\!$&$\!$40643$\!$&$\!$96760$\!$&$\!$32629$\!$&$\!$77410$\!$&$\!$32223$\!$&$\!$76923$\!$&$\!$36\%$\!$\\
$\!$6$\!$&$\!$100 500 $10^3$ 10 20 3 3 3 0 0 1$\!$&$\!$40923$\!$&$\!$96912$\!$&$\!$33104$\!$&$\!$78010$\!$&$\!$32723$\!$&$\!$77559$\!$&$\!$35\%$\!$\\
$\!$7$\!$&$\!$100 500 4000 40 80 3 3 3 1 0 1$\!$&$\!$44762$\!$&$\!$392675$\!$&$\!$43866$\!$&$\!$379156$\!$&$\!$43733$\!$&$\!$378907$\!$&$\!$13\%$\!$\\
$\!$8$\!$&$\!$100 500 4000 40 80 3 3 3 1 0 0$\!$&$\!$44756$\!$&$\!$392769$\!$&$\!$43884$\!$&$\!$379408$\!$&$\!$43752$\!$&$\!$379154$\!$&$\!$12\%$\!$\\
$\!$9$\!$&$\!$100 500 2000 20 40 3 3 3 0 0 0$\!$&$\!$49044$\!$&$\!$199907$\!$&$\!$48099$\!$&$\!$196061$\!$&$\!$48092$\!$&$\!$196054$\!$&$\!$4\%$\!$\\
$\!$10$\!$&$\!$100 500 2000 20 40 3 3 3 0 0 1$\!$&$\!$49075$\!$&$\!$199910$\!$&$\!$48157$\!$&$\!$196105$\!$&$\!$48150$\!$&$\!$196098$\!$&$\!$4\%$\!$\\
$\!$11$\!$&$\!$100 500 $10^3$ 10 20 3 3 3 1 1 0$\!$&$\!$49023$\!$&$\!$99965$\!$&$\!$48978$\!$&$\!$99943$\!$&$\!$48958$\!$&$\!$99922$\!$&$\!$2\%$\!$\\
$\!$12$\!$&$\!$100 500 $10^3$ 10 20 3 3 3 0 1 0$\!$&$\!$49036$\!$&$\!$99967$\!$&$\!$48994$\!$&$\!$99945$\!$&$\!$48980$\!$&$\!$99932$\!$&$\!$2\%$\!$\\
$\!$13$\!$&$\!$100 500 $10^3$ 10 20 3 3 3 1 1 1$\!$&$\!$49062$\!$&$\!$99967$\!$&$\!$49021$\!$&$\!$99947$\!$&$\!$49004$\!$&$\!$99930$\!$&$\!$2\%$\!$\\
$\!$14$\!$&$\!$100 500 $10^3$ 10 20 3 3 3 0 1 1$\!$&$\!$49098$\!$&$\!$99976$\!$&$\!$49060$\!$&$\!$99959$\!$&$\!$49050$\!$&$\!$99948$\!$&$\!$2\%$\!$\\
\hline
$\!$15$\!$&$\!$100 $10^3$ 2000 10 20 10 10 10 1 0 0$\!$&$\!$61016$\!$&$\!$182386$\!$&$\!$12192$\!$&$\!$30767$\!$&$\!$10483$\!$&$\!$28188$\!$&$\!$90\%$\!$\\
$\!$16$\!$&$\!$100 $10^3$ 2000 10 20 10 10 10 1 0 1$\!$&$\!$61043$\!$&$\!$182423$\!$&$\!$12447$\!$&$\!$31368$\!$&$\!$10693$\!$&$\!$28727$\!$&$\!$89\%$\!$\\
$\!$17$\!$&$\!$100 $10^3$ 4000 10 40 10 10 10 1 0 0$\!$&$\!$81028$\!$&$\!$389853$\!$&$\!$49049$\!$&$\!$191367$\!$&$\!$47061$\!$&$\!$188333$\!$&$\!$53\%$\!$\\
$\!$18$\!$&$\!$100 $10^3$ 4000 10 40 10 10 10 1 0 1$\!$&$\!$80992$\!$&$\!$389813$\!$&$\!$49308$\!$&$\!$192791$\!$&$\!$47314$\!$&$\!$189763$\!$&$\!$53\%$\!$\\
$\!$19$\!$&$\!$100 $10^3$ 2000 10 20 10 10 10 0 0 0$\!$&$\!$82985$\!$&$\!$196071$\!$&$\!$66269$\!$&$\!$156223$\!$&$\!$65565$\!$&$\!$155403$\!$&$\!$34\%$\!$\\
$\!$20$\!$&$\!$100 $10^3$ 2000 10 20 10 10 10 0 0 1$\!$&$\!$83105$\!$&$\!$196118$\!$&$\!$66650$\!$&$\!$157018$\!$&$\!$65949$\!$&$\!$156202$\!$&$\!$34\%$\!$\\
$\!$21$\!$&$\!$100 $10^3$ 8000 10 80 10 10 10 1 0 0$\!$&$\!$90946$\!$&$\!$793341$\!$&$\!$78745$\!$&$\!$618652$\!$&$\!$78450$\!$&$\!$618231$\!$&$\!$22\%$\!$\\
$\!$22$\!$&$\!$100 $10^3$ 8000 10 80 10 10 10 1 0 1$\!$&$\!$91132$\!$&$\!$793544$\!$&$\!$79091$\!$&$\!$621119$\!$&$\!$78806$\!$&$\!$620707$\!$&$\!$21\%$\!$\\
$\!$23$\!$&$\!$100 $10^3$ 4000 10 40 10 10 10 0 0 0$\!$&$\!$98172$\!$&$\!$399930$\!$&$\!$96254$\!$&$\!$392068$\!$&$\!$96216$\!$&$\!$392044$\!$&$\!$4\%$\!$\\
$\!$24$\!$&$\!$100 $10^3$ 4000 10 40 10 10 10 0 0 1$\!$&$\!$98204$\!$&$\!$399926$\!$&$\!$96297$\!$&$\!$392099$\!$&$\!$96261$\!$&$\!$392074$\!$&$\!$4\%$\!$\\
$\!$25$\!$&$\!$100 $10^3$ 2000 10 20 10 10 10 0 1 1$\!$&$\!$98142$\!$&$\!$199981$\!$&$\!$97996$\!$&$\!$199916$\!$&$\!$97985$\!$&$\!$199905$\!$&$\!$2\%$\!$\\
$\!$26$\!$&$\!$100 $10^3$ 2000 10 20 10 10 10 0 1 0$\!$&$\!$98172$\!$&$\!$199974$\!$&$\!$98015$\!$&$\!$199904$\!$&$\!$98002$\!$&$\!$199891$\!$&$\!$2\%$\!$\\
$\!$27$\!$&$\!$100 $10^3$ 2000 10 20 10 10 10 1 1 0$\!$&$\!$98178$\!$&$\!$199979$\!$&$\!$98033$\!$&$\!$199917$\!$&$\!$98020$\!$&$\!$199904$\!$&$\!$2\%$\!$\\
$\!$28$\!$&$\!$100 $10^3$ 2000 10 20 10 10 10 1 1 1$\!$&$\!$98205$\!$&$\!$199982$\!$&$\!$98053$\!$&$\!$199918$\!$&$\!$98035$\!$&$\!$199900$\!$&$\!$2\%$\!$\\
\hline
$\!$29$\!$&$\!$$10^4$ 10 20 2 2 100 1 1 1 0 0$\!$&$\!$66561$\!$&$\!$179656$\!$&$\!$59876$\!$&$\!$154571$\!$&$\!$48389$\!$&$\!$135097$\!$&$\!$52\%$\!$\\
$\!$30$\!$&$\!$$10^4$ 10 20 2 2 100 1 1 1 0 1$\!$&$\!$67037$\!$&$\!$179907$\!$&$\!$60402$\!$&$\!$155212$\!$&$\!$49169$\!$&$\!$136221$\!$&$\!$51\%$\!$\\
$\!$31$\!$&$\!$$10^4$ 10 20 2 2 100 1 1 0 0 0$\!$&$\!$86107$\!$&$\!$195614$\!$&$\!$72014$\!$&$\!$158565$\!$&$\!$68495$\!$&$\!$153152$\!$&$\!$32\%$\!$\\
$\!$32$\!$&$\!$$10^4$ 10 20 2 2 100 1 1 0 0 1$\!$&$\!$89537$\!$&$\!$196972$\!$&$\!$74760$\!$&$\!$159751$\!$&$\!$72758$\!$&$\!$156755$\!$&$\!$27\%$\!$\\
$\!$33$\!$&$\!$$10^4$ 10 30 3 3 100 1 1 1 0 1$\!$&$\!$80362$\!$&$\!$288491$\!$&$\!$81252$\!$&$\!$285041$\!$&$\!$75079$\!$&$\!$272676$\!$&$\!$25\%$\!$\\
$\!$34$\!$&$\!$$10^4$ 10 30 3 3 100 1 1 1 0 0$\!$&$\!$80533$\!$&$\!$288538$\!$&$\!$81431$\!$&$\!$285198$\!$&$\!$75191$\!$&$\!$272794$\!$&$\!$25\%$\!$\\
$\!$35$\!$&$\!$$10^4$ 10 30 3 3 100 1 1 0 0 0$\!$&$\!$96942$\!$&$\!$299151$\!$&$\!$94335$\!$&$\!$289965$\!$&$\!$93929$\!$&$\!$289046$\!$&$\!$6\%$\!$\\
$\!$36$\!$&$\!$$10^4$ 10 30 3 3 100 1 1 0 0 1$\!$&$\!$98026$\!$&$\!$299487$\!$&$\!$95367$\!$&$\!$290390$\!$&$\!$95119$\!$&$\!$289908$\!$&$\!$5\%$\!$\\
$\!$37$\!$&$\!$$10^4$ 10 20 2 2 100 1 1 0 1 0$\!$&$\!$99076$\!$&$\!$199861$\!$&$\!$98865$\!$&$\!$199684$\!$&$\!$98832$\!$&$\!$199695$\!$&$\!$1\%$\!$\\
$\!$38$\!$&$\!$$10^4$ 10 20 2 2 100 1 1 0 1 1$\!$&$\!$99271$\!$&$\!$199914$\!$&$\!$99015$\!$&$\!$199710$\!$&$\!$98994$\!$&$\!$199688$\!$&$\!$1\%$\!$\\
$\!$39$\!$&$\!$$10^4$ 10 20 2 2 100 1 1 1 1 0$\!$&$\!$99556$\!$&$\!$199938$\!$&$\!$99424$\!$&$\!$199843$\!$&$\!$99400$\!$&$\!$199843$\!$&$\!$1\%$\!$\\
$\!$40$\!$&$\!$$10^4$ 10 20 2 2 100 1 1 1 1 1$\!$&$\!$99635$\!$&$\!$199944$\!$&$\!$99463$\!$&$\!$199820$\!$&$\!$99446$\!$&$\!$199800$\!$&$\!$1\%$\!$\\
\hline
\end{tabular}
\caption{Results of performance tests -- Part~I.\label{table: wyniki1}}
\end{center}
\end{table}

We have also conducted performance tests. For this purpose, we implemented the module {\em experiments.py}~\cite{min2025-prog}, which contains, among other things, the function
\[ \mathit{runExperiment(k, n', m', o, p, l, sCN, sRN, acyclic, withI, withO)}. \]
This function first generates a random fuzzy interpretation~$\mI$ using the given parameters. It then computes: 
\begin{itemize}
\item a minimization of~$\mI$ preserving fuzzy concept assertions of $\mLP$, denoted by $\mI_1$, using Algorithm~\ref{alg2}, 
\item a minimization of~$\mI$ preserving fuzzy concept assertions of $\mLPD$, denoted by $\mI_2$, using an algorithm from~\cite{DBLP:journals/fss/Nguyen24} (see Algorithm~2 and Corollary~3.8 of~\cite{DBLP:journals/fss/Nguyen24}),
\item a minimization of~$\mI$ preserving fuzzy concept assertions of $\mLPU$, denoted by $\mI_3$, using an algorithm from~\cite{minimization-by-fBS} (see Definition~3 and Theorem~5 of~\cite{minimization-by-fBS}),
\end{itemize}
where $\Phi$ is determined by the parameters {\em withI} and {\em withO}, and where $\mLPD$ and $\mLPU$ are extensions of $\mLP$ with the Baaz projection operator or the universal role, respectively.  
The function returns $(n_3, m_3, n_2, m_2, n_1, m_1, t_1)$, where $n_1$, $n_2$, and $n_3$ denote the domain sizes of $\mI_1$, $\mI_2$, and $\mI_3$, respectively, $m_1$, $m_2$, and $m_3$ denote the numbers of nonzero instances of atomic roles in the respective fuzzy interpretations, and $t_1$ is the execution time of Algorithm~\ref{alg2}, measured in seconds. In the case where $\Phi = \emptyset$, the FDLs $\mLP$, $\mLPD$, and $\mLPU$ are \fALCreg, \fALCregD and \fALCregU, respectively, as discussed in Section~\ref{sec: motivation}. 

Such a fuzzy interpretation~$\mI$ is randomly generated based on the following criteria:
\begin{itemize}
\item The domain $\Delta^\mI$ consists of $k$ pairwise disconnected components, where disconnectedness means that for any $r \in \RN$, there is no pair $\tuple{x, y}$ of individuals from different components such that $r^\mI(x, y) > 0$.

\item Each component of $\Delta^\mI$ contains exactly $n'$ individuals, with $m'$ nonzero instances of atomic roles between individuals within that component. Each component includes $o$ named individuals ($o \leq n'$). Individuals are consecutively numbered, and named individuals appear first within each component -- this ordering is required when {\em acyclic = True}.

\item Each component contains exactly $p$ instances of atomic concepts. That is, for each component~$X$, the cardinality of the set $\{\tuple{A, x} \mid A \in \CN, x \in X, A^\mI(x) > 0\}$ is equal to~$p$.

\item The number of distinct nonzero fuzzy values used for concept and role instances in~$\mI$ is~$l$.

\item The parameters $\mathit{sCN}$ and $\mathit{sRN}$ mean $|\CN|$ and $|\RN|$, respectively.

\item If {\em acyclic = True}, then all atomic role instances in $\mI$ form an acyclic graph. This is ensured by the condition that for any $r \in \RN$, $r^\mI(x, y) > 0$ only when the identification number of $x$ is less than that of~$y$.
\end{itemize}

\begin{table}[t!]
\begin{center}
\footnotesize
\begin{tabular}{|r|l||r|r||r|r||r|r|r|}
\hline
\# & Parameters for {\em runExperiment} & \multicolumn{2}{|c||}{$\mLPU$} & \multicolumn{2}{|c||}{$\mLPD$} & \multicolumn{3}{|c|}{$\mLP$} \\
\cline{3-4} 
\cline{5-6}
\cline{7-9}
   &                              & $n_3$ & $m_3$ & $n_2$ & $m_2$ & $n_1$ & $m_1$ & Red. \\
\hline\hline
$\!$41$\!$&$\!$2000 100 200 10 10 10 1 2 1 0 0$\!$&$\!$107559$\!$&$\!$330328$\!$&$\!$75053$\!$&$\!$219598$\!$&$\!$64743$\!$&$\!$200504$\!$&$\!$68\%$\!$\\
$\!$42$\!$&$\!$2000 100 200 10 10 10 1 2 1 0 1$\!$&$\!$108984$\!$&$\!$331482$\!$&$\!$76313$\!$&$\!$220796$\!$&$\!$66281$\!$&$\!$202184$\!$&$\!$67\%$\!$\\
$\!$43$\!$&$\!$2000 100 200 10 10 10 1 2 0 0 0$\!$&$\!$163636$\!$&$\!$387704$\!$&$\!$135503$\!$&$\!$318889$\!$&$\!$133179$\!$&$\!$316044$\!$&$\!$33\%$\!$\\
$\!$44$\!$&$\!$2000 100 400 10 10 10 1 2 1 0 1$\!$&$\!$153551$\!$&$\!$749158$\!$&$\!$141038$\!$&$\!$661715$\!$&$\!$136405$\!$&$\!$652677$\!$&$\!$32\%$\!$\\
$\!$45$\!$&$\!$2000 100 400 10 10 10 1 2 1 0 0$\!$&$\!$153578$\!$&$\!$749409$\!$&$\!$141199$\!$&$\!$662494$\!$&$\!$136490$\!$&$\!$653488$\!$&$\!$32\%$\!$\\
$\!$46$\!$&$\!$2000 100 200 10 10 10 1 2 0 0 1$\!$&$\!$168908$\!$&$\!$390646$\!$&$\!$140414$\!$&$\!$321535$\!$&$\!$138464$\!$&$\!$319216$\!$&$\!$31\%$\!$\\
$\!$47$\!$&$\!$2000 100 800 10 10 10 1 2 1 0 0$\!$&$\!$177714$\!$&$\!$1560757$\!$&$\!$175969$\!$&$\!$1525522$\!$&$\!$174703$\!$&$\!$1522458$\!$&$\!$13\%$\!$\\
$\!$48$\!$&$\!$2000 100 800 10 10 10 1 2 1 0 1$\!$&$\!$177744$\!$&$\!$1561200$\!$&$\!$176012$\!$&$\!$1525152$\!$&$\!$174714$\!$&$\!$1522235$\!$&$\!$13\%$\!$\\
$\!$49$\!$&$\!$2000 100 400 10 10 10 1 2 0 0 0$\!$&$\!$196248$\!$&$\!$799530$\!$&$\!$192838$\!$&$\!$785526$\!$&$\!$192783$\!$&$\!$785461$\!$&$\!$4\%$\!$\\
$\!$50$\!$&$\!$2000 100 400 10 10 10 1 2 0 0 1$\!$&$\!$196670$\!$&$\!$799598$\!$&$\!$193246$\!$&$\!$785381$\!$&$\!$193194$\!$&$\!$785321$\!$&$\!$3\%$\!$\\
$\!$51$\!$&$\!$2000 100 200 10 10 10 1 2 0 1 0$\!$&$\!$196261$\!$&$\!$399852$\!$&$\!$196152$\!$&$\!$399777$\!$&$\!$196048$\!$&$\!$399686$\!$&$\!$2\%$\!$\\
$\!$52$\!$&$\!$2000 100 200 10 10 10 1 2 1 1 0$\!$&$\!$196338$\!$&$\!$399864$\!$&$\!$196238$\!$&$\!$399790$\!$&$\!$196112$\!$&$\!$399679$\!$&$\!$2\%$\!$\\
$\!$53$\!$&$\!$2000 100 200 10 10 10 1 2 0 1 1$\!$&$\!$196635$\!$&$\!$399878$\!$&$\!$196490$\!$&$\!$399785$\!$&$\!$196416$\!$&$\!$399711$\!$&$\!$2\%$\!$\\
$\!$54$\!$&$\!$2000 100 200 10 10 10 1 2 1 1 1$\!$&$\!$196749$\!$&$\!$399894$\!$&$\!$196606$\!$&$\!$399803$\!$&$\!$196507$\!$&$\!$399703$\!$&$\!$2\%$\!$\\
\hline
$\!$55$\!$&$\!$$5 \!\cdot\! 10^4$ 5 6 1 2 10 1 1 0 0 0$\!$&$\!$117289$\!$&$\!$223800$\!$&$\!$67394$\!$&$\!$107751$\!$&$\!$44931$\!$&$\!$75999$\!$&$\!$82\%$\!$\\
$\!$56$\!$&$\!$$5 \!\cdot\! 10^4$ 5 6 1 2 10 1 1 1 0 0$\!$&$\!$76584$\!$&$\!$181259$\!$&$\!$66691$\!$&$\!$139477$\!$&$\!$48830$\!$&$\!$109067$\!$&$\!$80\%$\!$\\
$\!$57$\!$&$\!$$5 \!\cdot\! 10^4$ 5 6 1 2 10 1 1 1 0 1$\!$&$\!$90733$\!$&$\!$195909$\!$&$\!$75566$\!$&$\!$149324$\!$&$\!$61340$\!$&$\!$125244$\!$&$\!$75\%$\!$\\
$\!$58$\!$&$\!$$5 \!\cdot\! 10^4$ 5 6 1 2 10 1 1 0 0 1$\!$&$\!$162194$\!$&$\!$255861$\!$&$\!$90876$\!$&$\!$121438$\!$&$\!$81607$\!$&$\!$108408$\!$&$\!$67\%$\!$\\
$\!$59$\!$&$\!$$5 \!\cdot\! 10^4$ 5 10 1 2 10 1 1 1 0 0$\!$&$\!$130503$\!$&$\!$394310$\!$&$\!$130376$\!$&$\!$363135$\!$&$\!$90129$\!$&$\!$273048$\!$&$\!$64\%$\!$\\
$\!$60$\!$&$\!$$5 \!\cdot\! 10^4$ 5 10 1 2 10 1 1 1 0 1$\!$&$\!$131998$\!$&$\!$395982$\!$&$\!$130616$\!$&$\!$363632$\!$&$\!$92251$\!$&$\!$278014$\!$&$\!$63\%$\!$\\
$\!$61$\!$&$\!$$5 \!\cdot\! 10^4$ 5 15 1 2 10 1 1 1 0 0$\!$&$\!$143582$\!$&$\!$590409$\!$&$\!$169828$\!$&$\!$634230$\!$&$\!$106095$\!$&$\!$428117$\!$&$\!$58\%$\!$\\
$\!$62$\!$&$\!$$5 \!\cdot\! 10^4$ 5 15 1 2 10 1 1 1 0 1$\!$&$\!$144251$\!$&$\!$591556$\!$&$\!$169718$\!$&$\!$633891$\!$&$\!$107420$\!$&$\!$434702$\!$&$\!$57\%$\!$\\
$\!$63$\!$&$\!$$5 \!\cdot\! 10^4$ 5 10 1 2 10 1 1 0 0 0$\!$&$\!$210619$\!$&$\!$475110$\!$&$\!$171067$\!$&$\!$363035$\!$&$\!$149414$\!$&$\!$321684$\!$&$\!$40\%$\!$\\
$\!$64$\!$&$\!$$5 \!\cdot\! 10^4$ 5 10 1 2 10 1 1 0 0 1$\!$&$\!$223061$\!$&$\!$484146$\!$&$\!$177929$\!$&$\!$367383$\!$&$\!$169306$\!$&$\!$351189$\!$&$\!$32\%$\!$\\
$\!$65$\!$&$\!$$5 \!\cdot\! 10^4$ 5 6 1 2 10 1 1 0 1 0$\!$&$\!$228911$\!$&$\!$291944$\!$&$\!$211182$\!$&$\!$265136$\!$&$\!$185088$\!$&$\!$262951$\!$&$\!$26\%$\!$\\
$\!$66$\!$&$\!$$5 \!\cdot\! 10^4$ 5 6 1 2 10 1 1 1 1 0$\!$&$\!$234320$\!$&$\!$293642$\!$&$\!$223602$\!$&$\!$275970$\!$&$\!$189542$\!$&$\!$272972$\!$&$\!$24\%$\!$\\
$\!$67$\!$&$\!$$5 \!\cdot\! 10^4$ 5 6 1 2 10 1 1 0 1 1$\!$&$\!$234807$\!$&$\!$294246$\!$&$\!$215268$\!$&$\!$266929$\!$&$\!$214727$\!$&$\!$266297$\!$&$\!$14\%$\!$\\
$\!$68$\!$&$\!$$5 \!\cdot\! 10^4$ 5 15 1 2 10 1 1 0 0 0$\!$&$\!$243099$\!$&$\!$744521$\!$&$\!$238320$\!$&$\!$719598$\!$&$\!$224099$\!$&$\!$672244$\!$&$\!$10\%$\!$\\
$\!$69$\!$&$\!$$5 \!\cdot\! 10^4$ 5 6 1 2 10 1 1 1 1 1$\!$&$\!$239496$\!$&$\!$295404$\!$&$\!$226233$\!$&$\!$277221$\!$&$\!$225191$\!$&$\!$275788$\!$&$\!$10\%$\!$\\
$\!$70$\!$&$\!$$5 \!\cdot\! 10^4$ 5 15 1 2 10 1 1 0 0 1$\!$&$\!$246532$\!$&$\!$747320$\!$&$\!$239137$\!$&$\!$720133$\!$&$\!$236137$\!$&$\!$705277$\!$&$\!$6\%$\!$\\
\hline
$\!$71$\!$&$\!$$6 \!\cdot\! 10^4$ 5 6 1 2 100 10 10 1 0 0$\!$&$\!$172759$\!$&$\!$329557$\!$&$\!$109014$\!$&$\!$208070$\!$&$\!$98565$\!$&$\!$196032$\!$&$\!$67\%$\!$\\
$\!$72$\!$&$\!$$7 \!\cdot\! 10^4$ 5 6 1 2 100 10 10 1 0 0$\!$&$\!$200892$\!$&$\!$384075$\!$&$\!$126752$\!$&$\!$242600$\!$&$\!$114299$\!$&$\!$228217$\!$&$\!$67\%$\!$\\
$\!$73$\!$&$\!$$8 \!\cdot\! 10^4$ 5 6 1 2 100 10 10 1 0 0$\!$&$\!$229118$\!$&$\!$438958$\!$&$\!$144992$\!$&$\!$278003$\!$&$\!$130111$\!$&$\!$260747$\!$&$\!$67\%$\!$\\
$\!$74$\!$&$\!$$9 \!\cdot\! 10^4$ 5 6 1 2 100 10 10 1 0 0$\!$&$\!$256798$\!$&$\!$493347$\!$&$\!$161921$\!$&$\!$311295$\!$&$\!$145035$\!$&$\!$291601$\!$&$\!$68\%$\!$\\
$\!$75$\!$&$\!$$10^5$ 5 6 1 2 100 10 10 1 0 0$\!$&$\!$284873$\!$&$\!$548022$\!$&$\!$179661$\!$&$\!$345894$\!$&$\!$160386$\!$&$\!$323435$\!$&$\!$68\%$\!$\\
$\!$76$\!$&$\!$$6 \!\cdot\! 10^4$ 5 10 1 2 100 10 10 1 0 0$\!$&$\!$215332$\!$&$\!$575199$\!$&$\!$169853$\!$&$\!$448029$\!$&$\!$158062$\!$&$\!$432159$\!$&$\!$47\%$\!$\\
$\!$77$\!$&$\!$$7 \!\cdot\! 10^4$ 5 10 1 2 100 10 10 1 0 0$\!$&$\!$250916$\!$&$\!$670726$\!$&$\!$198055$\!$&$\!$522775$\!$&$\!$183865$\!$&$\!$503612$\!$&$\!$47\%$\!$\\
$\!$78$\!$&$\!$$8 \!\cdot\! 10^4$ 5 10 1 2 100 10 10 1 0 0$\!$&$\!$286090$\!$&$\!$766130$\!$&$\!$225608$\!$&$\!$596347$\!$&$\!$208982$\!$&$\!$573736$\!$&$\!$48\%$\!$\\
$\!$79$\!$&$\!$$9 \!\cdot\! 10^4$ 5 10 1 2 100 10 10 1 0 0$\!$&$\!$321463$\!$&$\!$861594$\!$&$\!$253583$\!$&$\!$671245$\!$&$\!$234513$\!$&$\!$645265$\!$&$\!$48\%$\!$\\
$\!$80$\!$&$\!$$10^5$ 5 10 1 2 100 10 10 1 0 0$\!$&$\!$356668$\!$&$\!$957381$\!$&$\!$281428$\!$&$\!$745838$\!$&$\!$259517$\!$&$\!$715904$\!$&$\!$48\%$\!$\\
\hline
\end{tabular}
\caption{Results of performance tests -- Part~II.\label{table: wyniki2}}
\end{center}
\end{table}

In Tables~\ref{table: wyniki1} and~\ref{table: wyniki2}, we present the results of our performance tests. The values for the parameters {\em acyclic}, {\em withI}, and {\em withO} are indicated as 1 ({\em True}) or 0 ({\em False}). The first column lists test identifiers, while the last column reports the domain reduction percentage achieved by Algorithm~\ref{alg2}, calculated as \mbox{$1 - n_1/(k \cdot n')$}. The results shown in the last seven columns are rounded averages computed over three repetitions of each test. 

Denote $n = k \cdot n'$ (the domain size of $\mI$) and $m = k \cdot m'$ (the number of nonzero instances of atomic roles in $\mI$). 
In the tests, these values are as follows:
\begin{itemize}
\item Tests~1--14: $n = 50\,000$ and $100\,000 \leq m \leq 400\,000$, 
\item Tests~15--40: $n = 100\,000$ and $200\,000 \leq m \leq 800\,000$,
\item Tests~41--54: $n = 200\,000$ and $400\,000 \leq m \leq 1\,600\,000$, 
\item Tests~55--70: $n = 250\,000$ and $300\,000 \leq m \leq 750\,000$, 
\item Tests~71--80: $300\,000 \leq n \leq 500\,000$ and $360\,000 \leq m \leq 10^6$.
\end{itemize}

Observe that:
\begin{itemize}
\item According to the nature of the minimizations, $\mI_1$ is consistently smaller than both $\mI_2$ and $\mI_3$ across all tests, with respect to both domain size and the number of nonzero atomic role instances. On average, the domain of $\mI_1$ is smaller than that of $\mI_2$ and $\mI_3$ by $7\%$ and $20\%$, respectively.
\item The reductions achieved by all three kinds of minimization generally increase when the parameter {\em withI} is switched from 1 to 0 and when {\em acyclic} is changed from 0 to 1.
\end{itemize}

Note that across the tests, the product $mn$ ranges from $5 \cdot 10^9$ to $5 \cdot 10^{11}$. 
We conducted our experiments on a laptop equipped with an Intel(R) Core(TM) i5-6200U CPU @~2.30GHz and 8.00~GB of RAM. 
Algorithm~\ref{alg2} completed any of the tests in less than 90 seconds, with an average runtime of less than 33 seconds. 
This empirical performance is consistent with the theoretical time complexity of Algorithm~\ref{alg2}, namely \mbox{$O((m\log{l} + n)\log{n})$}.


\begin{table}[t!]
\begin{center}
\footnotesize
\begin{tabular}{|r|l|r|r|r|r|r|}
\hline
\# & Parameters for {\em runExperiment2} & \multicolumn{5}{|c|}{Domain reduction} \\
   &                                     & \multicolumn{5}{|c|}{for a given threshold $\gamma$} \\
\cline{3-7} 
   &                                     & 1 & 0.8 & 0.4 & 0.2 & 0.1 \\
\hline\hline
$\!$1$\!$&$\!$100 500 $10^3$ 10 20 10 3 3 1 0 1$\!$&$\!$85\%$\!$&$\!$85\%$\!$&$\!$86\%$\!$&$\!$87\%$\!$&$\!$87\%$\!$\\
$\!$2$\!$&$\!$100 500 $10^3$ 10 20 10 3 3 1 0 0$\!$&$\!$85\%$\!$&$\!$85\%$\!$&$\!$86\%$\!$&$\!$87\%$\!$&$\!$87\%$\!$\\
$\!$3$\!$&$\!$100 500 2000 20 40 10 3 3 1 0 0$\!$&$\!$38\%$\!$&$\!$38\%$\!$&$\!$40\%$\!$&$\!$42\%$\!$&$\!$43\%$\!$\\
$\!$4$\!$&$\!$100 500 2000 20 40 10 3 3 1 0 1$\!$&$\!$38\%$\!$&$\!$38\%$\!$&$\!$40\%$\!$&$\!$42\%$\!$&$\!$43\%$\!$\\
$\!$5$\!$&$\!$100 500 $10^3$ 10 20 10 3 3 0 0 0$\!$&$\!$35\%$\!$&$\!$35\%$\!$&$\!$35\%$\!$&$\!$35\%$\!$&$\!$35\%$\!$\\
$\!$6$\!$&$\!$100 500 $10^3$ 10 20 10 3 3 0 0 1$\!$&$\!$34\%$\!$&$\!$34\%$\!$&$\!$35\%$\!$&$\!$35\%$\!$&$\!$35\%$\!$\\
$\!$7$\!$&$\!$100 500 4000 40 80 10 3 3 1 0 0$\!$&$\!$11\%$\!$&$\!$12\%$\!$&$\!$13\%$\!$&$\!$14\%$\!$&$\!$14\%$\!$\\
$\!$8$\!$&$\!$100 500 4000 40 80 10 3 3 1 0 1$\!$&$\!$11\%$\!$&$\!$12\%$\!$&$\!$13\%$\!$&$\!$14\%$\!$&$\!$14\%$\!$\\
$\!$9$\!$&$\!$100 500 2000 20 40 10 3 3 0 0 0$\!$&$\!$4\%$\!$&$\!$4\%$\!$&$\!$4\%$\!$&$\!$4\%$\!$&$\!$4\%$\!$\\
$\!$10$\!$&$\!$100 500 2000 20 40 10 3 3 0 0 1$\!$&$\!$4\%$\!$&$\!$4\%$\!$&$\!$4\%$\!$&$\!$4\%$\!$&$\!$4\%$\!$\\
$\!$11$\!$&$\!$100 500 $10^3$ 10 20 10 3 3 0 1 0$\!$&$\!$2\%$\!$&$\!$2\%$\!$&$\!$2\%$\!$&$\!$2\%$\!$&$\!$2\%$\!$\\
$\!$12$\!$&$\!$100 500 $10^3$ 10 20 10 3 3 1 1 0$\!$&$\!$2\%$\!$&$\!$2\%$\!$&$\!$2\%$\!$&$\!$2\%$\!$&$\!$2\%$\!$\\
$\!$13$\!$&$\!$100 500 $10^3$ 10 20 10 3 3 1 1 1$\!$&$\!$2\%$\!$&$\!$2\%$\!$&$\!$2\%$\!$&$\!$2\%$\!$&$\!$2\%$\!$\\
$\!$14$\!$&$\!$100 500 $10^3$ 10 20 10 3 3 0 1 1$\!$&$\!$2\%$\!$&$\!$2\%$\!$&$\!$2\%$\!$&$\!$2\%$\!$&$\!$2\%$\!$\\
\hline
$\!$15$\!$&$\!$100 $10^3$ 2000 10 20 10 10 10 1 0 0$\!$&$\!$90\%$\!$&$\!$90\%$\!$&$\!$90\%$\!$&$\!$90\%$\!$&$\!$90\%$\!$\\
$\!$16$\!$&$\!$100 $10^3$ 2000 10 20 10 10 10 1 0 1$\!$&$\!$89\%$\!$&$\!$89\%$\!$&$\!$90\%$\!$&$\!$90\%$\!$&$\!$90\%$\!$\\
$\!$17$\!$&$\!$100 $10^3$ 4000 10 40 10 10 10 1 0 1$\!$&$\!$53\%$\!$&$\!$53\%$\!$&$\!$54\%$\!$&$\!$55\%$\!$&$\!$56\%$\!$\\
$\!$18$\!$&$\!$100 $10^3$ 4000 10 40 10 10 10 1 0 0$\!$&$\!$53\%$\!$&$\!$53\%$\!$&$\!$54\%$\!$&$\!$55\%$\!$&$\!$55\%$\!$\\
$\!$19$\!$&$\!$100 $10^3$ 2000 10 20 10 10 10 0 0 0$\!$&$\!$34\%$\!$&$\!$34\%$\!$&$\!$35\%$\!$&$\!$35\%$\!$&$\!$35\%$\!$\\
$\!$20$\!$&$\!$100 $10^3$ 2000 10 20 10 10 10 0 0 1$\!$&$\!$34\%$\!$&$\!$34\%$\!$&$\!$34\%$\!$&$\!$35\%$\!$&$\!$35\%$\!$\\
$\!$21$\!$&$\!$100 $10^3$ 8000 10 80 10 10 10 1 0 1$\!$&$\!$22\%$\!$&$\!$22\%$\!$&$\!$22\%$\!$&$\!$23\%$\!$&$\!$23\%$\!$\\
$\!$22$\!$&$\!$100 $10^3$ 8000 10 80 10 10 10 1 0 0$\!$&$\!$21\%$\!$&$\!$21\%$\!$&$\!$22\%$\!$&$\!$23\%$\!$&$\!$23\%$\!$\\
$\!$23$\!$&$\!$100 $10^3$ 4000 10 40 10 10 10 0 0 1$\!$&$\!$4\%$\!$&$\!$4\%$\!$&$\!$4\%$\!$&$\!$4\%$\!$&$\!$4\%$\!$\\
$\!$24$\!$&$\!$100 $10^3$ 4000 10 40 10 10 10 0 0 0$\!$&$\!$4\%$\!$&$\!$4\%$\!$&$\!$4\%$\!$&$\!$4\%$\!$&$\!$4\%$\!$\\
$\!$25$\!$&$\!$100 $10^3$ 2000 10 20 10 10 10 1 1 0$\!$&$\!$2\%$\!$&$\!$2\%$\!$&$\!$2\%$\!$&$\!$2\%$\!$&$\!$2\%$\!$\\
$\!$26$\!$&$\!$100 $10^3$ 2000 10 20 10 10 10 1 1 1$\!$&$\!$2\%$\!$&$\!$2\%$\!$&$\!$2\%$\!$&$\!$2\%$\!$&$\!$2\%$\!$\\
$\!$27$\!$&$\!$100 $10^3$ 2000 10 20 10 10 10 0 1 1$\!$&$\!$2\%$\!$&$\!$2\%$\!$&$\!$2\%$\!$&$\!$2\%$\!$&$\!$2\%$\!$\\
$\!$28$\!$&$\!$100 $10^3$ 2000 10 20 10 10 10 0 1 0$\!$&$\!$2\%$\!$&$\!$2\%$\!$&$\!$2\%$\!$&$\!$2\%$\!$&$\!$2\%$\!$\\
\hline
$\!$29$\!$&$\!$$10^4$ 10 20 2 2 10 1 1 1 0 0$\!$&$\!$60\%$\!$&$\!$60\%$\!$&$\!$65\%$\!$&$\!$72\%$\!$&$\!$78\%$\!$\\
$\!$30$\!$&$\!$$10^4$ 10 20 2 2 10 1 1 1 0 1$\!$&$\!$59\%$\!$&$\!$59\%$\!$&$\!$63\%$\!$&$\!$68\%$\!$&$\!$72\%$\!$\\
$\!$31$\!$&$\!$$10^4$ 10 30 3 3 10 1 1 1 0 0$\!$&$\!$34\%$\!$&$\!$34\%$\!$&$\!$41\%$\!$&$\!$50\%$\!$&$\!$59\%$\!$\\
$\!$32$\!$&$\!$$10^4$ 10 20 2 2 10 1 1 0 0 0$\!$&$\!$33\%$\!$&$\!$34\%$\!$&$\!$35\%$\!$&$\!$36\%$\!$&$\!$38\%$\!$\\
$\!$33$\!$&$\!$$10^4$ 10 30 3 3 10 1 1 1 0 1$\!$&$\!$33\%$\!$&$\!$34\%$\!$&$\!$40\%$\!$&$\!$48\%$\!$&$\!$55\%$\!$\\
$\!$34$\!$&$\!$$10^4$ 10 20 2 2 10 1 1 0 0 1$\!$&$\!$28\%$\!$&$\!$28\%$\!$&$\!$29\%$\!$&$\!$29\%$\!$&$\!$30\%$\!$\\
$\!$35$\!$&$\!$$10^4$ 10 30 3 3 10 1 1 0 0 0$\!$&$\!$6\%$\!$&$\!$6\%$\!$&$\!$7\%$\!$&$\!$7\%$\!$&$\!$8\%$\!$\\
$\!$36$\!$&$\!$$10^4$ 10 30 3 3 10 1 1 0 0 1$\!$&$\!$5\%$\!$&$\!$5\%$\!$&$\!$6\%$\!$&$\!$6\%$\!$&$\!$6\%$\!$\\
$\!$37$\!$&$\!$$10^4$ 10 20 2 2 10 1 1 0 1 0$\!$&$\!$1\%$\!$&$\!$1\%$\!$&$\!$1\%$\!$&$\!$1\%$\!$&$\!$1\%$\!$\\
$\!$38$\!$&$\!$$10^4$ 10 20 2 2 10 1 1 0 1 1$\!$&$\!$1\%$\!$&$\!$1\%$\!$&$\!$1\%$\!$&$\!$1\%$\!$&$\!$1\%$\!$\\
$\!$39$\!$&$\!$$10^4$ 10 20 2 2 10 1 1 1 1 0$\!$&$\!$1\%$\!$&$\!$1\%$\!$&$\!$1\%$\!$&$\!$1\%$\!$&$\!$1\%$\!$\\
$\!$40$\!$&$\!$$10^4$ 10 20 2 2 10 1 1 1 1 1$\!$&$\!$1\%$\!$&$\!$1\%$\!$&$\!$1\%$\!$&$\!$1\%$\!$&$\!$1\%$\!$\\
\hline
\end{tabular}
\caption{Results of experiments for different values of~$\gamma$ -- Part~I.\label{table: results2a}}
\end{center}
\end{table}

\begin{table}[t!]
\begin{center}
\footnotesize
\begin{tabular}{|r|l|r|r|r|r|r|}
\hline
\# & Parameters for {\em runExperiment2} & \multicolumn{5}{|c|}{Domain reduction} \\
   &                                     & \multicolumn{5}{|c|}{for a given threshold $\gamma$} \\
\cline{3-7} 
   &                                     & 1 & 0.8 & 0.4 & 0.2 & 0.1 \\
\hline\hline
$\!$41$\!$&$\!$2000 100 200 10 10 10 1 2 1 0 0$\!$&$\!$68\%$\!$&$\!$68\%$\!$&$\!$70\%$\!$&$\!$72\%$\!$&$\!$74\%$\!$\\
$\!$42$\!$&$\!$2000 100 200 10 10 10 1 2 1 0 1$\!$&$\!$67\%$\!$&$\!$67\%$\!$&$\!$69\%$\!$&$\!$71\%$\!$&$\!$73\%$\!$\\
$\!$43$\!$&$\!$2000 100 200 10 10 10 1 2 0 0 0$\!$&$\!$33\%$\!$&$\!$33\%$\!$&$\!$34\%$\!$&$\!$34\%$\!$&$\!$34\%$\!$\\
$\!$44$\!$&$\!$2000 100 400 10 10 10 1 2 1 0 0$\!$&$\!$32\%$\!$&$\!$32\%$\!$&$\!$34\%$\!$&$\!$37\%$\!$&$\!$38\%$\!$\\
$\!$45$\!$&$\!$2000 100 400 10 10 10 1 2 1 0 1$\!$&$\!$32\%$\!$&$\!$32\%$\!$&$\!$34\%$\!$&$\!$37\%$\!$&$\!$39\%$\!$\\
$\!$46$\!$&$\!$2000 100 200 10 10 10 1 2 0 0 1$\!$&$\!$31\%$\!$&$\!$31\%$\!$&$\!$31\%$\!$&$\!$31\%$\!$&$\!$32\%$\!$\\
$\!$47$\!$&$\!$2000 100 800 10 10 10 1 2 1 0 0$\!$&$\!$13\%$\!$&$\!$13\%$\!$&$\!$14\%$\!$&$\!$16\%$\!$&$\!$17\%$\!$\\
$\!$48$\!$&$\!$2000 100 800 10 10 10 1 2 1 0 1$\!$&$\!$13\%$\!$&$\!$13\%$\!$&$\!$14\%$\!$&$\!$16\%$\!$&$\!$17\%$\!$\\
$\!$49$\!$&$\!$2000 100 400 10 10 10 1 2 0 0 0$\!$&$\!$4\%$\!$&$\!$4\%$\!$&$\!$4\%$\!$&$\!$4\%$\!$&$\!$4\%$\!$\\
$\!$50$\!$&$\!$2000 100 400 10 10 10 1 2 0 0 1$\!$&$\!$3\%$\!$&$\!$3\%$\!$&$\!$3\%$\!$&$\!$3\%$\!$&$\!$4\%$\!$\\
$\!$51$\!$&$\!$2000 100 200 10 10 10 1 2 0 1 0$\!$&$\!$2\%$\!$&$\!$2\%$\!$&$\!$2\%$\!$&$\!$2\%$\!$&$\!$2\%$\!$\\
$\!$52$\!$&$\!$2000 100 200 10 10 10 1 2 1 1 0$\!$&$\!$2\%$\!$&$\!$2\%$\!$&$\!$2\%$\!$&$\!$2\%$\!$&$\!$2\%$\!$\\
$\!$53$\!$&$\!$2000 100 200 10 10 10 1 2 0 1 1$\!$&$\!$2\%$\!$&$\!$2\%$\!$&$\!$2\%$\!$&$\!$2\%$\!$&$\!$2\%$\!$\\
$\!$54$\!$&$\!$2000 100 200 10 10 10 1 2 1 1 1$\!$&$\!$2\%$\!$&$\!$2\%$\!$&$\!$2\%$\!$&$\!$2\%$\!$&$\!$2\%$\!$\\
\hline
$\!$55$\!$&$\!$$5 \!\cdot\! 10^4$ 5 6 1 2 10 1 1 0 0 0$\!$&$\!$82\%$\!$&$\!$83\%$\!$&$\!$86\%$\!$&$\!$92\%$\!$&$\!$97\%$\!$\\
$\!$56$\!$&$\!$$5 \!\cdot\! 10^4$ 5 6 1 2 10 1 1 1 0 0$\!$&$\!$80\%$\!$&$\!$81\%$\!$&$\!$88\%$\!$&$\!$95\%$\!$&$\!$99\%$\!$\\
$\!$57$\!$&$\!$$5 \!\cdot\! 10^4$ 5 6 1 2 10 1 1 1 0 1$\!$&$\!$75\%$\!$&$\!$76\%$\!$&$\!$78\%$\!$&$\!$79\%$\!$&$\!$80\%$\!$\\
$\!$58$\!$&$\!$$5 \!\cdot\! 10^4$ 5 6 1 2 10 1 1 0 0 1$\!$&$\!$67\%$\!$&$\!$67\%$\!$&$\!$68\%$\!$&$\!$69\%$\!$&$\!$70\%$\!$\\
$\!$59$\!$&$\!$$5 \!\cdot\! 10^4$ 5 10 1 2 10 1 1 1 0 0$\!$&$\!$64\%$\!$&$\!$64\%$\!$&$\!$72\%$\!$&$\!$86\%$\!$&$\!$98\%$\!$\\
$\!$60$\!$&$\!$$5 \!\cdot\! 10^4$ 5 10 1 2 10 1 1 1 0 1$\!$&$\!$63\%$\!$&$\!$63\%$\!$&$\!$69\%$\!$&$\!$76\%$\!$&$\!$79\%$\!$\\
$\!$61$\!$&$\!$$5 \!\cdot\! 10^4$ 5 15 1 2 10 1 1 1 0 0$\!$&$\!$58\%$\!$&$\!$59\%$\!$&$\!$76\%$\!$&$\!$97\%$\!$&$\!$99\%$\!$\\
$\!$62$\!$&$\!$$5 \!\cdot\! 10^4$ 5 15 1 2 10 1 1 1 0 1$\!$&$\!$57\%$\!$&$\!$58\%$\!$&$\!$70\%$\!$&$\!$79\%$\!$&$\!$80\%$\!$\\
$\!$63$\!$&$\!$$5 \!\cdot\! 10^4$ 5 10 1 2 10 1 1 0 0 0$\!$&$\!$40\%$\!$&$\!$41\%$\!$&$\!$43\%$\!$&$\!$50\%$\!$&$\!$62\%$\!$\\
$\!$64$\!$&$\!$$5 \!\cdot\! 10^4$ 5 10 1 2 10 1 1 0 0 1$\!$&$\!$32\%$\!$&$\!$32\%$\!$&$\!$33\%$\!$&$\!$35\%$\!$&$\!$37\%$\!$\\
$\!$65$\!$&$\!$$5 \!\cdot\! 10^4$ 5 6 1 2 10 1 1 0 1 0$\!$&$\!$26\%$\!$&$\!$26\%$\!$&$\!$27\%$\!$&$\!$43\%$\!$&$\!$76\%$\!$\\
$\!$66$\!$&$\!$$5 \!\cdot\! 10^4$ 5 6 1 2 10 1 1 1 1 0$\!$&$\!$24\%$\!$&$\!$24\%$\!$&$\!$27\%$\!$&$\!$54\%$\!$&$\!$89\%$\!$\\
$\!$67$\!$&$\!$$5 \!\cdot\! 10^4$ 5 6 1 2 10 1 1 0 1 1$\!$&$\!$14\%$\!$&$\!$14\%$\!$&$\!$14\%$\!$&$\!$14\%$\!$&$\!$14\%$\!$\\
$\!$68$\!$&$\!$$5 \!\cdot\! 10^4$ 5 15 1 2 10 1 1 0 0 0$\!$&$\!$10\%$\!$&$\!$10\%$\!$&$\!$13\%$\!$&$\!$28\%$\!$&$\!$54\%$\!$\\
$\!$69$\!$&$\!$$5 \!\cdot\! 10^4$ 5 6 1 2 10 1 1 1 1 1$\!$&$\!$10\%$\!$&$\!$10\%$\!$&$\!$10\%$\!$&$\!$10\%$\!$&$\!$10\%$\!$\\
$\!$70$\!$&$\!$$5 \!\cdot\! 10^4$ 5 15 1 2 10 1 1 0 0 1$\!$&$\!$6\%$\!$&$\!$6\%$\!$&$\!$6\%$\!$&$\!$8\%$\!$&$\!$10\%$\!$\\
\hline
$\!$71$\!$&$\!$$10^5$ 5 6 1 2 10 10 10 1 0 0$\!$&$\!$71\%$\!$&$\!$71\%$\!$&$\!$74\%$\!$&$\!$76\%$\!$&$\!$77\%$\!$\\
$\!$72$\!$&$\!$$9 \!\cdot\! 10^4$ 5 6 1 2 10 10 10 1 0 0$\!$&$\!$71\%$\!$&$\!$71\%$\!$&$\!$73\%$\!$&$\!$76\%$\!$&$\!$77\%$\!$\\
$\!$73$\!$&$\!$$8 \!\cdot\! 10^4$ 5 6 1 2 10 10 10 1 0 0$\!$&$\!$71\%$\!$&$\!$71\%$\!$&$\!$73\%$\!$&$\!$75\%$\!$&$\!$77\%$\!$\\
$\!$74$\!$&$\!$$7 \!\cdot\! 10^4$ 5 6 1 2 10 10 10 1 0 0$\!$&$\!$71\%$\!$&$\!$71\%$\!$&$\!$73\%$\!$&$\!$75\%$\!$&$\!$77\%$\!$\\
$\!$75$\!$&$\!$$6 \!\cdot\! 10^4$ 5 6 1 2 10 10 10 1 0 0$\!$&$\!$70\%$\!$&$\!$71\%$\!$&$\!$73\%$\!$&$\!$75\%$\!$&$\!$76\%$\!$\\
$\!$76$\!$&$\!$$10^5$ 5 10 1 2 10 10 10 1 0 0$\!$&$\!$52\%$\!$&$\!$52\%$\!$&$\!$55\%$\!$&$\!$57\%$\!$&$\!$59\%$\!$\\
$\!$77$\!$&$\!$$9 \!\cdot\! 10^4$ 5 10 1 2 10 10 10 1 0 0$\!$&$\!$52\%$\!$&$\!$52\%$\!$&$\!$54\%$\!$&$\!$57\%$\!$&$\!$59\%$\!$\\
$\!$78$\!$&$\!$$8 \!\cdot\! 10^4$ 5 10 1 2 10 10 10 1 0 0$\!$&$\!$52\%$\!$&$\!$52\%$\!$&$\!$54\%$\!$&$\!$57\%$\!$&$\!$59\%$\!$\\
$\!$79$\!$&$\!$$7 \!\cdot\! 10^4$ 5 10 1 2 10 10 10 1 0 0$\!$&$\!$51\%$\!$&$\!$51\%$\!$&$\!$54\%$\!$&$\!$57\%$\!$&$\!$59\%$\!$\\
$\!$80$\!$&$\!$$6 \!\cdot\! 10^4$ 5 10 1 2 10 10 10 1 0 0$\!$&$\!$51\%$\!$&$\!$51\%$\!$&$\!$54\%$\!$&$\!$56\%$\!$&$\!$58\%$\!$\\
\hline
\end{tabular}
\caption{Results of experiments for different values of~$\gamma$ -- Part~II.\label{table: results2b}}
\end{center}
\end{table}

{\markRed
We have also conducted experiments to study the effect of domain reduction using different values of the parameter~$\gamma$. Before reporting the results, note that, in general, the smaller the threshold~$\gamma$ is, the stronger the domain reduction. However, this does not mean that a dramatic decrease in~$\gamma$ necessarily leads to a stronger reduction. For example, reducing~$\gamma$ from~1 to~0.1 does not result in a stronger reduction for the fuzzy interpretations considered in Examples~\ref{example: JHRHS} and~\ref{example: HGRJS}. Recall also that, for the fuzzy interpretation considered in Example~\ref{example: YNSJA}, reducing~$\gamma$ from~1 to~0.8 enables a reduction (from six to three individuals). However, further reducing~$\gamma$ to~0.1 does not yield a stronger reduction.

The module {\em experiments.py}~\cite{min2025-prog} also contains the function
\[ \mathit{runExperiment2(k, n', m', o, p, l, sCN, sRN, acyclic, withI, withO)}, \]
which has the same parameters as the function {\em runExperiment}.
This function first generates a random fuzzy interpretation $\mI$ using the given parameters, in the same way as {\em runExperiment}. It then computes, for each $\gamma \in \{1, 0.8, 0.4, 0.2, 0.1\}$, a minimal reduction $\mIp$ of $\mI$ that preserves fuzzy concept assertions of $\mLP$ up to degree~$\gamma$.
The function returns the tuple $(n_1$, $m_1$, $n_{0.8}$, $m_{0.8}$, $n_{0.4}$, $m_{0.4}$, $n_{0.2}$, $m_{0.2}$, $n_{0.1}$, $m_{0.1})$, where, for each $\gamma \in \{1, 0.8, 0.4, 0.2, 0.1\}$, $n_\gamma$ denotes the domain size of $\mIp$ and $m_\gamma$ denotes the number of nonzero instances of atomic roles in $\mIp$.

In Tables~\ref{table: results2a} and~\ref{table: results2b}, we present the results of experiments studying the effect of domain reduction for different values of~$\gamma$, using the {\em runExperiment2} function. The experiments use the same combinations of parameter values as those used for the performance tests reported in Tables~\ref{table: wyniki1} and~\ref{table: wyniki2}, except that the parameter~$l$ (the number of distinct nonzero fuzzy values used for concept and role instances in~$\mI$) is now set to~10. The first column lists experiment identifiers, while the last five columns report the domain reduction percentage achieved by Algorithm~\ref{alg2}, calculated as \mbox{$1 - n_\gamma/n$}, where $\gamma \in \{1, 0.8, 0.4, 0.2, 0.1\}$ and $n = k \cdot n'$ (the domain size of~$\mI$). The results shown in the last five columns are rounded averages computed over three repetitions of each experiment.
}

{\markRed
\section{Related work}
\label{section: related work}

As mentioned in the introduction, the problem of minimizing or reducing a given finite fuzzy automaton (\FfA) while preserving its accepted fuzzy language has been extensively studied. Most existing work focuses either on minimizing deterministic \FfAs \cite{LI20071423,Halamish2015,YANG202172,DEMENDIVILGRAU2024109108,ShamsizadehZG24} or on reducing nondeterministic \FfAs (without necessarily achieving minimality) \cite{BASAK2002223,PEEVA20084152,CSIP.10,SCI.14,SCB.18}, with an exception of \cite{LITFS2015}, which addresses the minimization of nondeterministic \FfAs over the G\"odel structure. It is worth noting that minimizing a nondeterministic finite (crisp) automaton is PSPACE-complete~\cite{Jiang1993}; consequently, the minimization of nondeterministic \FfAs is PSPACE-hard.
Several works \cite{Belohlvek2009OnAM,YangL.19,EPTCS386.6,StanimirovicMC.22,FSS-D-25-00029} investigate the problem of constructing a minimal or reduced \FfA that approximates a given one. While the works \cite{BASAK2002223,LI20071423,Belohlvek2009OnAM,Halamish2015,YangL.19,YANG202172,StanimirovicMC.22,ShamsizadehZG24} focus exclusively on \FfAs over the G\"odel structure or Heyting algebras, the studies in \cite{PEEVA20084152,CSIP.10,SCI.14,SCB.18,EPTCS386.6,DEMENDIVILGRAU2024109108,FSS-D-25-00029} consider \FfAs defined over arbitrary residuated lattices.
Quotient constructions, which rely on merging equivalent states, are among the most commonly used techniques for minimizing or reducing \FfAs.

It is well-known that the largest auto-bisimulation can be used to minimize a crisp Kripke model. For example, in the book~\cite{BRV2001}, Blackburn et al.\ wrote ``the quotient of a [Kripke] model with respect to its largest auto-bisimulation can be regarded as a minimal representation of this model'', without specifying what the minimal model preserves. 
Minimization of crisp interpretations in description logics is a related topic. It was studied by Divroodi and Nguyen in~\cite{BSDL-INS} for various description logics that extend \ALCreg with features among inverse roles, nominals, the universal role, qualified number restrictions, and the $\Self$ concept constructor, under the criteria of preserving bisimilarity, terminological axioms, or concept assertions. 
As usual, the minimization method of~\cite{BSDL-INS} is also based on bisimulation and the quotient construction. 

Extending the study on minimization for fuzzy Kripke models and fuzzy interpretations in FDLs, a natural approach is to start with fuzzy bisimulations. This term was first coined by Fan in her work~\cite{Fan15} on fuzzy bisimulation for modal logics under the G\"odel semantics, where she provided logical characterizations of both fuzzy and crisp bisimulations for such logics. Inspired by that work, in~\cite{FSS2020} together with colleagues we provided logical characterizations of fuzzy/crisp bisimulations in FDLs under the G\"odel semantics, and in~\cite{FBSML} we provided logical characterizations of fuzzy bisimulations in modal logics over residuated lattices. 
Related work includes the work~\cite{DBLP:journals/tfs/NguyenN23} on crisp bisimulations in FDLs under a general semantics, the works~\cite{EleftheriouKN12,aml/MartiM18,fuin/Diaconescu20} on crisp bisimulations in fuzzy/many-valued modal logics, and the works~\cite{DBLP:journals/fss/00010SS25,DBLP:journals/ijar/StankovicSC25} on approximate bisimulations for fuzzy modal logics over Heyting algebras. 
It is also worth mentioning the works~\cite{fss/WuD16,DBLP:journals/fss/WuCBD18,DBLP:journals/ijar/WuCHC18,DBLP:journals/tfs/QiaoZF23,fss/QiaoZP23} on crisp/fuzzy bisimulations of fuzzy transition systems and the works~\cite{ai/FanL14,IgnjatovicCS15} on fuzzy/crisp bisimulations of weighted/fuzzy social networks.

As stated in the introduction, the study of minimization for fuzzy Kripke models and fuzzy interpretations in FDLs has received relatively little attention. Apart from the works \cite{minimization-by-fBS,DBLP:journals/fss/Nguyen24}, which were discussed in Section~\ref{sec: motivation}, we are only aware of the additional works \cite{DBLP:journals/fss/00010SS25,10.1007/978-981-99-7743-7_7} on this topic. In \cite{DBLP:journals/fss/00010SS25}, Stanković et al.\ showed that approximate bisimulations can be used to minimize fuzzy Kripke models in fuzzy modal logics over Heyting algebras. In that work, the quotient construction is called a factor fuzzy Kripke model. It preserves, up to a given approximation threshold, the set of states that are equivalent to those in the original model. We do not have access to \cite{10.1007/978-981-99-7743-7_7}; however, according to its abstract, the paper studies the property of finite paths in nondeterministic fuzzy Kripke structures (NFKSs) and the behavior of quotient NFKSs, treating NFKSs as fuzzy labeled transition systems.

In a more relaxed context, one may also consider blockmodeling of valued (social) networks~\cite{Doreian2009,DBLP:journals/socnet/Ziberna07} or fuzzy graph clustering~\cite{DONG20061760,PENG202138,10339895}. However, none of these works employ logic-based criteria for reduction.
}

\section{Discussion and conclusions}
\label{section: conc}

The motivating example demonstrates that minimizing interpretations in FDLs without the Baaz projection operator and the universal role generally yields smaller interpretations than in FDLs with these constructors. While the Baaz projection operator and the universal role increase the expressive power to distinguish between individuals, they can, in many cases, be too strong -- analogous to overfitting in machine learning. In particular, in the presence of the Baaz projection operator, the similarity between two individuals $x$ and $x'$, defined as the infimum of \mbox{$C^\mI(x) \fequiv C^\mI(x')$} over all concepts $C$ of a considered FDL in a given fuzzy interpretation~$\mI$, is always crisp (either~0 or~1)~\cite{DBLP:journals/tfs/NguyenN23}. As for the universal role, it is often undesirable when locality is important. Thus, focusing on FDLs without the Baaz projection operator and the universal role is both natural and well justified.

We have presented the {\em first} algorithm (Algorithm~\ref{alg2}) that minimizes a finite fuzzy interpretation while preserving fuzzy concept assertions in FDLs without the Baaz projection operator and the universal role, under the G\"odel semantics. The algorithm is provided in an extended form that supports approximate preservation: it minimizes a finite fuzzy interpretation $\mI$ while preserving fuzzy concept assertions up to a degree~$\gamma \in (0,1]$. 
The algorithm is efficient, with time complexity \mbox{$O((m \log l + n)\log n)$}, where $n$ is the size of the domain of~$\mI$, $m$ is the number of nonzero instances of atomic roles in~$\mI$, and $l$ is the number of distinct fuzzy values occurring in such instances, plus~2. \red{This complexity is low, as it is comparable to that of related problems involving bisimulations. (Recall that the Paige--Tarjan algorithm~\cite{PaigeT87} computes the coarsest partition of a finite crisp graph in time $O((m+n)\log n)$.) When the number of fuzzy values occurring in the input fuzzy interpretation is bounded by a constant, the complexity $O((m\log l + n)\log n)$ coincides with $O((m+n)\log n)$.} When $m \geq n$ and $l$ is comparable to~$m$, the complexity bound can be simplified to \mbox{$O(m \log^{2} n)$}. We have implemented the algorithm and made the source code publicly available~\cite{min2025-prog}.

\red{At present, a comparison of Algorithm~\ref{alg2} with other algorithms 
on ``standard benchmarks'' is not applicable. First, no such benchmarks exist. Second, Algorithm~\ref{alg2} is, to the best of our knowledge, the first and currently the only algorithm addressing this specific minimization problem, and it achieves a low complexity. Since the task is minimization, the algorithm produces an optimal reduction of a given finite fuzzy interpretation with respect to the considered minimization criterion.}

It is well known that DLs are variants of modal logics. For instance, the DL \ALCreg corresponds to Propositional Dynamic Logic (PDL). Kripke models in modal logics correspond to interpretations in DLs: states (possible worlds) and actions (accessibility relations) in a Kripke model correspond to individuals and roles in a DL interpretation, respectively. In a pointed Kripke model, the distinguished actual world corresponds to a unique named individual in a DL interpretation.
By treating pointed fuzzy Kripke models as fuzzy interpretations in FDLs, Algorithm~\ref{alg2} can be used to minimize a finite pointed fuzzy Kripke model while preserving the truth degrees of formulas of a fuzzy modal logic $\mL$ at the actual world under the G\"odel semantics, up to a degree $\gamma \in (0,1]$, where $\mL$ can range from the sublogic without disjunction and universal modalities of the basic fuzzy multimodal logic $\mathit{fK}_n$ to the fuzzy version of PDL with converse.

\red{This work focuses on minimizing interpretations only under the G\"odel semantics. First, this semantics for FDLs differs significantly from other t-norm-based semantics. It is shown in \cite[Theorem 4.1 and Remark 4.5]{FBSML} that, in the presence of the \(\to\) (implication) operator, concepts are invariant under fuzzy bisimulations with respect to the G\"odel semantics, but not with respect to the product or \L{}ukasiewicz semantics. Furthermore, the greatest fuzzy auto-bisimulation of a finite fuzzy interpretation under the G\"odel semantics can be computed in time \(O(m \log n \log l + n^2)\)~\cite{DBLP:journals/isci/Nguyen23}, whereas no polynomial-time algorithms are currently known for computing such a fuzzy relation under the product or \L{}ukasiewicz semantics \cite{DBLP:journals/tfs/NguyenMS23}. Second, minimizing interpretations under these latter semantics is of a different nature, beyond the scope of this article, and is left for future work.}


\biboptions{sort&compress}
\bibliography{BSfDL}
\bibliographystyle{plain}

\appendix

\renewcommand{\thetheorem}{A.\arabic{theorem}}
\renewcommand{\thesection}{A}
\section{Additional examples}
\label{appendix A}

In Section~\ref{section: main}, we provided examples illustrating the execution of Algorithm~\ref{alg2} in the case $\Phi = \emptyset$. In this appendix, we present additional examples to illustrate how the algorithm behaves when $\emptyset \subset \Phi \subseteq \{I, O\}$ -- that is, in cases with inverse roles and/or nominals. These examples are described in detail to support a better understanding of the algorithm's behavior, but the reader may choose to skip them if they are not needed.

\begin{example}\label{example: JHRHS 2}
Reconsider the fuzzy interpretation $\mI$ from the motivating example in Section~\ref{sec: motivation}, now for the case $\Phi = \{O\}$ (i.e., with nominals). The compact fuzzy partition corresponding to the greatest fuzzy $\Phi$-auto-bisimulation of~$\mI$ is $\bB = \{\{u\}_1$, $\{u'\}_1$, $\{\{v_1,v'_1\}_1$, $\{\{v_2,v'_2\}_1$, $\{v_3\}_1\}_{0.8}\}_{0.7}\}_0$. Denote $B_v = \{\{v_1,v'_1\}_1$, $\{\{v_2,v'_2\}_1$, $\{v_3\}_1\}_{0.8}\}_{0.7}$. Consider the execution of Algorithm~\ref{alg2} for $\mI$ using $\Phi$ and $\gamma = 1$. The resulting fuzzy interpretation, denoted as $\mI_4$ in Figure~\ref{fig: JHDKJ 2}, is illustrated there. Details are provided below.

After executing statements \ref{step: alg2 1}-\ref{step: alg2 12}, we have:
\begin{itemize}
\item $\Delta^\mIp = \{u, u'\}$, $a^\mIp = u$, $b^\mIp = u'$;
\item for $B \in \bB.\blocks()$, 
\[ 
B.\repr = 
\begin{cases}
u & \textrm{if } B = \{u\}_1 \textrm{ or } B = \bB,\\ 
u' & \textrm{if } B = \{u'\}_1,\\
\Null & \textrm{otherwise};
\end{cases}
\]
\item $Q$ contains $\tuple{u,r,v_1}$, $\tuple{u,r,v_2}$, $\tuple{u,r,v_3}$, $\tuple{u',r,v'_1}$ and $\tuple{u',r,v'_2}$;
\item $D = \{1, 0.7, 0.6, 0.5, 0.4\}$. 
\end{itemize}

The first iteration of the ``foreach'' loop at statement~\ref{step: alg2 13} is executed with $d = 1$, in which the inner ``while'' loop terminates immediately.

Consider the second iteration of the ``foreach'' loop at statement~\ref{step: alg2 13}, with $d = 0.7$:
\begin{itemize}
\item The first iteration of the inner ``while'' loop is executed with $\tuple{x,R,y} = \tuple{u,r,v_3}$, in which:
    \begin{itemize}
    \item since $\bB.\findBlock(v_3, 0.7) = B_v$ and $B_v.\repr = \Null$, $v_3$ is added to $\Delta^\mIp$ and $B_v.\repr$ is set to $v_3$;
    \item $r^\mIp(u,v_3)$ is set to $0.7$. 
    \end{itemize} 
\item The second iteration of the inner ``while'' loop is executed with $\tuple{x,R,y} = \tuple{u',r,v'_2}$. In this iteration, since $\bB.\findBlock(v'_2, 0.7) = B_v$ and $B_v.\repr = v_3$, $r^\mIp(u',v_3)$ is set to $0.7$. 
\end{itemize}

The third iteration of the ``foreach'' loop at statement~\ref{step: alg2 13} is executed with $d = 0.6$. During this iteration, the inner ``while'' loop executes only one iteration, with $\tuple{x,R,y} = \tuple{u',r,v'_1}$, in which we have $\bB.\findBlock(v'_1, 0.6) = B_v$, with $B_v.\repr = v_3$, and $\mI'$ remains unchanged.

The fourth (resp.\ fifth) iteration of the ``foreach'' loop at statement~\ref{step: alg2 13} is executed with $d = 0.5$ (resp.\ $d = 0.4$). During this iteration, the inner ``while'' loop executes only one iteration, with $\tuple{x,R,y}$ equal to $\tuple{u,r,v_1}$ (resp. $\tuple{u,r,v_2}$), in which $\mIp$ remains unchanged. The reasons are similar to those given in the above paragraph.

No more iterations of the ``foreach'' loop at statement~\ref{step: alg2 13} are executed. The algorithm terminates and returns $\mIp$ with $|\Delta^\mIp| = 3$.
\myend
\end{example}

\begin{figure}[t]
\begin{center}
\begin{tabular}{|c|c|}
\hline
\begin{tikzpicture}
\node (po) {$\mI_4:\qquad\quad$};
\node (u) [node distance=0.7cm, below of=po] {$u: a$};
\node (up) [node distance=2.4cm, right of=u] {$u':b\quad$};
\node (ub) [node distance=2.0cm, below of=u] {};
\node (v3) [node distance=1.2cm, right of=ub] {$v_3:A_{\,0.9}$};
\draw[->] (u) to node [left]{\footnotesize{0.7}} (v3);
\draw[->] (up) to node [right]{\footnotesize{0.7}} (v3);
\end{tikzpicture}
&
\begin{tikzpicture}
\node (po) {$\mI_5:\qquad\quad$};
\node (u) [node distance=0.7cm, below of=po] {$u: a$};
\node (up) [node distance=2.4cm, right of=u] {$u':b$};
\node (v3) [node distance=2.0cm, below of=u] {$v_3:A_{\,0.9}$};
\node (vp2) [node distance=2.0cm, below of=up] {$v'_2:A_{\,0.8}$};
\draw[->] (u) to node [left]{\footnotesize{0.7}} (v3);
\draw[->] (up) to node [right]{\footnotesize{0.7}} (vp2);
\end{tikzpicture}
\\
\hline
\end{tabular}
\caption{Fuzzy interpretations mentioned in Examples~\ref{example: JHRHS 2} and~\ref{example: JHRHS 3}.\label{fig: JHDKJ 2}}
\end{center}
\end{figure}

\begin{example}\label{example: JHRHS 3}
Let us continue Examples~\ref{example: JHRHS} and~\ref{example: JHRHS 2}. We now reconsider the fuzzy interpretation $\mI$ from the motivating example in Section~\ref{sec: motivation}, this time for the case $\Phi = \{I\}$ (i.e., with inverse roles). The compact fuzzy partition corresponding to the greatest fuzzy $\Phi$-auto-bisimulation of~$\mI$ is $\bB = \{\{\{u\}_1, \{u'\}_1\}_{0.4}$, $\{\{\{v_1\}_1, \{v_3\}_1\}_{0.5}$, $\{v_2\}_1$, $\{\{v'_1\}_1, \{v'_2\}_1\}_{0.6}\}_{0.4}\}_0$. Consider the execution of Algorithm~\ref{alg2} for $\mI$ using $\Phi$ and $\gamma = 1$. The resulting fuzzy interpretation, denoted as $\mI_5$ in Figure~\ref{fig: JHDKJ 2}, is illustrated there. Details are provided below.

After executing statements \ref{step: alg2 1}-\ref{step: alg2 12}, we have:
\begin{itemize}
\item $\Delta^\mIp = \{u, u'\}$, $a^\mIp = u$, $b^\mIp = u'$;
\item for $B \in \bB.\blocks()$, 
\[ 
B.\repr = 
\begin{cases}
u & \textrm{if } B = \{u\}_1 \textrm{ or } B = \{\{u\}_1, \{u'\}_1\}_{0.4} \textrm{ or } B = \bB,\\ 
u' & \textrm{if } B = \{u'\}_1,\\
\Null & \textrm{otherwise};
\end{cases}
\]
\item $Q$ contains $\tuple{u,r,v_1}$, $\tuple{u,r,v_2}$, $\tuple{u,r,v_3}$, $\tuple{u',r,v'_1}$ and $\tuple{u',r,v'_2}$;
\item $D = \{1, 0.7, 0.6, 0.5, 0.4\}$. 
\end{itemize}

The first iteration of the ``foreach'' loop at statement~\ref{step: alg2 13} is executed with $d = 1$, in which the inner ``while'' loop terminates immediately.

Consider the second iteration of the ``foreach'' loop at statement~\ref{step: alg2 13}, with $d = 0.7$:
\begin{itemize}
\item The first iteration of the inner ``while'' loop is executed with $\tuple{x,R,y} = \tuple{u,r,v_3}$, in which:
    \begin{itemize}
    \item for $B = \bB.\findBlock(v_3, 0.7) = \{v_3\}_1$, since $B.\repr = \Null$, $v_3$ is added to $\Delta^\mIp$, and $B.\repr$, $B'.\repr$ and $B''.\repr$ are set to $v_3$, where $B' = B.\parent = \{\{v_1\}_1, \{v_3\}_1\}_{0.5}$ and $B'' = B'.\parent = \{\{\{v_1\}_1, \{v_3\}_1\}_{0.5}$, $\{v_2\}_1$, $\{\{v'_1\}_1, \{v'_2\}_1\}_{0.6}\}_{0.4}$; 
    \item the tuple $\tuple{v_3, \cnv{r}, u}$ is inserted into $Q$ and $r^\mIp(u,v_3)$ is set to $0.7$. 
    \end{itemize} 
\item The second iteration of the inner ``while'' loop is executed with $\tuple{x,R,y} = \tuple{u',r,v'_2}$, in which:
    \begin{itemize}
    \item for $B = \bB.\findBlock(v'_2, 0.7) = \{v'_2\}_1$, since $B.\repr = \Null$, $v'_2$ is added to $\Delta^\mIp$, and $B.\repr$ and $B'.\repr$ are set to $v'_2$, where $B' = B.\parent = \{\{v'_1\}_1, \{v'_2\}_1\}_{0.6}$;
    \item the tuple $\tuple{v'_2, \cnv{r}, u'}$ is inserted into $Q$ and $r^\mIp(u',v'_2)$ is set to $0.7$. 
    \end{itemize} 
\item The third iteration of the inner ``while'' loop is executed with $\tuple{x,R,y} = \tuple{v_3,\cnv{r},u}$, in which we have $\bB.\findBlock(u, 0.7) = \{u\}_1$ and $\mI'$ remains unchanged.
\item The fourth iteration of the inner ``while'' loop is executed with $\tuple{x,R,y} = \tuple{v'_2,\cnv{r},u'}$, in which we have $\bB.\findBlock(u', 0.7) = \{u'\}_1$ and $\mI'$ remains unchanged.
\end{itemize}

The third iteration of the ``foreach'' loop at statement~\ref{step: alg2 13} is executed with $d = 0.6$. During this iteration, the inner ``while'' loop executes only one iteration, with $\tuple{x,R,y} = \tuple{u',r,v'_1}$, in which, for $B = \bB.\findBlock(v'_1, 0.6) = \{\{v'_1\}_1, \{v'_2\}_1\}_{0.6}$, we have $B.\repr = v'_2$ and $\mI'$ remains unchanged.

The fourth (resp.\ fifth) iteration of the ``foreach'' loop at statement~\ref{step: alg2 13} is executed with $d = 0.5$ (resp.\ $d = 0.4$). During this iteration, the inner ``while'' loop executes only one iteration, with $\tuple{x,R,y}$ equal to $\tuple{u,r,v_1}$ (resp. $\tuple{u,r,v_2}$), in which $\mIp$ remains unchanged. The reasons are similar to those given in the above paragraph.

No more iterations of the ``foreach'' loop at statement~\ref{step: alg2 13} are executed. 
The algorithm terminates and returns $\mIp$ with $|\Delta^\mIp| = 4$. 
\myend
\end{example}

\begin{example}\label{example: JHRHS 4}
Executing Algorithm~\ref{alg2} for the fuzzy interpretation $\mI$ from the motivating example in Section~\ref{sec: motivation} using $\Phi = \{I,O\}$ (i.e., with inverse roles and nominals) and $\gamma = 1$ results in the same fuzzy interpretation $\mI_5$ discussed in Example~\ref{example: JHRHS 3} and illustrated in Figure~\ref{fig: JHDKJ 2}. The steps are very similar to the ones discussed in Example~\ref{example: JHRHS 3}, with the following differences: 
\begin{itemize}
\item The compact fuzzy partition corresponding to the greatest fuzzy $\Phi$-auto-bisimulation of~$\mI$ is $\bB = \{\{u\}_1$, $\{u'\}_1$, $\{\{\{v_1\}_1, \{v_3\}_1\}_{0.5}$, $\{v_2\}_1\}_{0.4}$, $\{\{v'_1\}_1, \{v'_2\}_1\}_{0.6}\}_0$. 
\item After executing statements \ref{step: alg2 1}-\ref{step: alg2 12}, we have $B.\repr = u$ only if $B = \{u\}_1$ or $B = \bB$.
\item After the second iteration of the ``foreach'' loop at statement~\ref{step: alg2 13} (with $d = 0.7$), we have $B.\repr = v_3$ for each $B$ among $\{v_3\}_1$, $\{\{v_1\}_1, \{v_3\}_1\}_{0.5}$, and $\{\{\{v_1\}_1, \{v_3\}_1\}_{0.5}, \{v_2\}_1\}_{0.4}$.
\myend
\end{itemize}
\end{example}


\begin{example}\label{example: HGRJS 2}
Reconsider the fuzzy interpretation $\mI$ specified in Example~\ref{example: HGRJS} and illustrated in Figure~\ref{fig: HGRJS}.
Consider the execution of Algorithm~\ref{alg2} for $\mI$ using $\Phi = \{O\}$ (i.e., with nominals) and $\gamma = 1$. The resulting fuzzy interpretation, denoted as $\mI_2$ in Figure~\ref{fig: HGRJS 2}, is illustrated there. Details are provided below.

The compact fuzzy partition $\bB$ corresponding to the greatest fuzzy $\Phi$-auto-bisimulation of~$\mI$ is 
\[ \bB = 
\{
  \{u_1\}_1, \{u_2\}_1, \{v_3\}_1, \{w_1\}_1, \{w_2\}_1, 
  \{\{v_1\}_1, \{v_2\}_1\}_{0.5}
\}_0.
\]
After executing statements \ref{step: alg2 1}-\ref{step: alg2 12}, we have:
\begin{itemize}
\item $\Delta^\mIp = \{u_1, u_2\}$, $a^\mIp = u_1$, $b^\mIp = u_2$;
\item for $B \in \bB.\blocks()$, 
\[ B.\repr = 
\begin{cases}
u_1 & \textrm{if } B = \{u_1\}_1 \textrm{ or } B = \bB,\\
u_2 & \textrm{if } B = \{u_2\}_1,\\
\Null & \textrm{otherwise;}
\end{cases}
\]
\item $Q$ contains $\tuple{u_1,r,v_1}$, $\tuple{u_1,r,v_2}$ and $\tuple{u_2,r,v_3}$;
\item $D = \{1, 0.9, 0.8, 0.7, 0.4, 0.3, 0.2\}$. 
\end{itemize}

In the first four iterations of the ``foreach'' loop at statement~\ref{step: alg2 13}, which are executed with $d \in \{1, 0.9, 0.8, 0.7\}$, the inner ``while'' loop terminates immediately.

Consider the fifth iteration of the ``foreach'' loop at statement~\ref{step: alg2 13}, with $d = 0.4$:
\begin{itemize}
\item The first iteration of the inner ``while'' loop is executed with $\tuple{x,R,y} = \tuple{u_1, r, v_2}$, in which:
    \begin{itemize}
    \item for $B = \bB.\findBlock(v_2, 0.4) = \{\{v_1\}_1, \{v_2\}_1\}_{0.5}$, since $B.\repr = \Null$, $v_2$ is added to $\Delta^\mIp$ and $B.\repr$ is set to $v_2$; 
    \item the triples $\tuple{v_2, r, v_1}$ and $\tuple{v_2, r, w_1}$ are inserted into $Q$ and $r^\mIp(u_1,v_2)$ is set to~$0.4$. 
    \end{itemize} 
\item The second iteration of the inner ``while'' loop is executed with $\tuple{x,R,y} = \tuple{v_2, r, v_1}$. In this iteration, for $B = \bB.\findBlock(v_1, 0.4) = \{\{v_1\}_1, \{v_2\}_1\}_{0.5}$, we have $B.\repr = v_2$ and $r^\mIp(v_2,v_2)$ is set to~$0.4$.
\item The third iteration of the inner ``while'' loop is executed with $\tuple{x,R,y} = \tuple{v_2, r, w_1}$, in which:
    \begin{itemize}
    \item for $B = \bB.\findBlock(w_1, 0.4) = \{w_1\}_1$, since $B.\repr = \Null$, $w_1$ is added to $\Delta^\mIp$ and $B.\repr$ is set to $w_1$;
    \item the triple $\tuple{w_1,s,u_2}$ is inserted into $Q$ and $r^\mIp(v_2,w_1)$ is set to $0.4$.
    \end{itemize}
\item The fourth iteration of the inner ``while'' loop is executed with $\tuple{x,R,y} = \tuple{u_2, r, v_3}$, in which:
    \begin{itemize}
    \item for $B = \bB.\findBlock(v_3, 0.4) = \{v_3\}_1$, since $B.\repr = \Null$, $v_3$ is added to $\Delta^\mIp$ and $B.\repr$ is set to $v_3$; 
    \item the triples $\tuple{v_3, r, w_2}$ and $\tuple{v_3, r, v_3}$ are inserted into $Q$ and $r^\mIp(u_2,v_3)$ is set to~$0.4$. 
    \end{itemize} 
\item The fifth iteration of the inner ``while'' loop is executed with $\tuple{x,R,y} = \tuple{v_3, r, w_2}$, in which:
    \begin{itemize}
    \item for $B = \bB.\findBlock(w_2, 0.4) = \{w_2\}_1$, since $B.\repr = \Null$, $w_2$ is added to $\Delta^\mIp$ and $B.\repr$ is set to $w_2$; 
    \item the triple $\tuple{w_2, s, u_1}$ is inserted into $Q$ and $r^\mIp(v_3,w_2)$ is set to~$0.4$. 
    \end{itemize} 
\item The sixth iteration of the inner ``while'' loop is executed with $\tuple{x,R,y} = \tuple{v_3, r, v_3}$. In this iteration, for $B = \bB.\findBlock(v_3, 0.4) = \{v_3\}_1$, we have $B.\repr = v_3$ and $r^\mIp(v_3,v_3)$ is set to~$0.4$.
\end{itemize}

Consider the sixth iteration of the ``foreach'' loop at statement~\ref{step: alg2 13}, with $d = 0.3$. During this iteration, the inner ``while'' loop executes only one iteration, with $\tuple{x,R,y} = \tuple{u_1,r,v_1}$, in which, for $B = \bB.\findBlock(v_1, 0.3) = \{\{v_1\}_1, \{v_2\}_1\}_{0.5}$, we have $B.\repr = v_2$ and $\mIp$ remains unchanged. 

Consider the seventh iteration of the ``foreach'' loop at statement~\ref{step: alg2 13}, with $d = 0.2$: 
\begin{itemize}
\item The first iteration of the inner ``while'' loop is executed with $\tuple{x,R,y} = \tuple{w_1,s,u_2}$, in which, for $B = \bB.\findBlock(u_2, 0.2) = \{u_2\}_1$, we have $B.\repr = u_2$, and $s^\mIp(w_1,u_2)$ is set to $0.2$.
\item The second iteration of the inner ``while'' loop is executed with $\tuple{x,R,y} = \tuple{w_2,s,u_1}$, in which, for $B = \bB.\findBlock(u_1, 0.2) = \{u_1\}_1$, we have $B.\repr = u_1$, and $s^\mIp(w_2,u_1)$ is set to~$0.2$.
\end{itemize}

After the seventh iteration, the ``foreach'' loop at statement~\ref{step: alg2 13} terminates and the algorithm returns $\mIp$ with $|\Delta^\mIp| = 6$.
\myend
\end{example}    

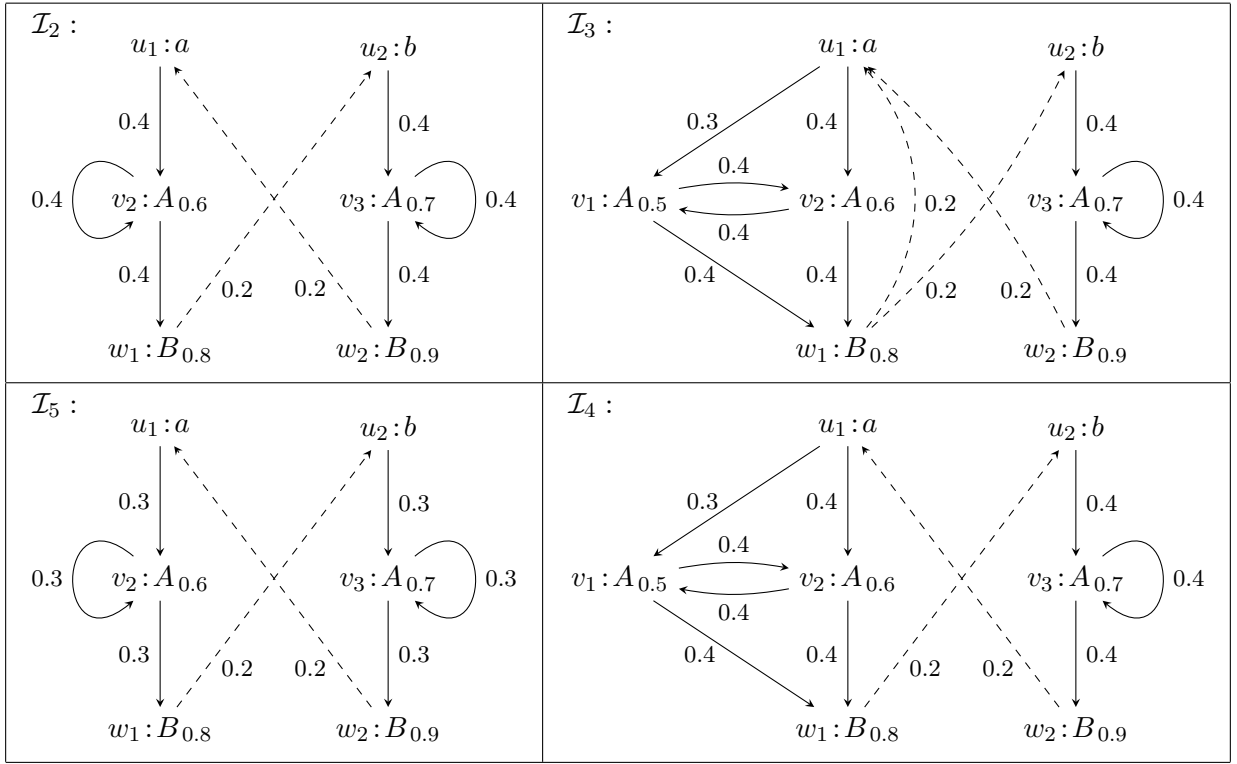
\begin{figure}[t]
\begin{center}
\begin{tabular}{|c|c|}
\hline
\begin{tikzpicture}[->,>=stealth,auto,black]
	\node (I) {$\mI_2:\qquad$};
	\node (bI) [node distance=0.3cm, below of=I] {};
	\node (u1) [node distance=1.0cm, right of=bI] {$u_1\!:\!a$};
	\node (u2) [node distance=3.0cm, right of=u1] {$u_2\!:\!b$};
	\node (v2) [node distance=2.0cm, below of=u1] {$v_2\!:\!A_{\,0.6}$};
	\node (v3) [node distance=2.0cm, below of=u2] {$v_3\!:\!A_{\,0.7}$};
	\node (w1) [node distance=2.0cm, below of=v2] {$w_1\!:\!B_{\,0.8}$};
	\node (w2) [node distance=2.0cm, below of=v3] {$w_2\!:\!B_{\,0.9}$};
	\draw (u1) to node [left]{\footnotesize{0.4}} (v2);	
	\draw (u2) to node [right]{\footnotesize{0.4}} (v3);	
	\draw (v2) to node [left]{\footnotesize{0.4}} (w1);	
	\draw (v3) to node [right]{\footnotesize{0.4}} (w2);
    \draw (v2) edge[out=140,in=-140,looseness=6] node[left]{\footnotesize{0.4}} (v2);
    \draw (v3) edge[out=40,in=-40,looseness=6] node[right]{\footnotesize{0.4}} (v3);
	\draw[dashed] (w1) to node[right,pos=0.15,xshift=2pt]{\footnotesize{0.2}} (u2);	
	\draw[dashed] (w2) to node[left,pos=0.15,xshift=-2pt]{\footnotesize{0.2}} (u1);
\end{tikzpicture}
&
\begin{tikzpicture}[->,>=stealth,auto,black]
	\node (I) {$\mI_3:\qquad$};
	\node (bI) [node distance=0.3cm, below of=I] {};
	\node (u1) [node distance=3.0cm, right of=bI] {$u_1\!:\!a$};
	\node (u2) [node distance=3.0cm, right of=u1] {$u_2\!:\!b$};
	\node (v1) [node distance=2.0cm, below of=bI] {$v_1\!:\!A_{\,0.5}$};
	\node (v2) [node distance=2.0cm, below of=u1] {$v_2\!:\!A_{\,0.6}$};
	\node (v3) [node distance=2.0cm, below of=u2] {$v_3\!:\!A_{\,0.7}$};
	\node (w1) [node distance=2.0cm, below of=v2] {$w_1\!:\!B_{\,0.8}$};
	\node (w2) [node distance=2.0cm, below of=v3] {$w_2\!:\!B_{\,0.9}$};
	\draw (u1) to node [left,xshift=-3pt]{\footnotesize{0.3}} (v1);	
	\draw (u1) to node [left]{\footnotesize{0.4}} (v2);	
	\draw (u2) to node [right]{\footnotesize{0.4}} (v3);	
	\draw (v1) to node [left,xshift=-3pt]{\footnotesize{0.4}} (w1);	
	\draw (v2) to node [left]{\footnotesize{0.4}} (w1);	
	\draw (v3) to node [right]{\footnotesize{0.4}} (w2);
    \draw (v1) edge[bend left=10] node[above]{\footnotesize{0.4}} (v2);
    \draw (v2) edge[bend left=10] node[below]{\footnotesize{0.4}} (v1);
    \draw (v3) edge[out=40,in=-40,looseness=6] node[right]{\footnotesize{0.4}} (v3);
    \draw[dashed] (w1) edge[bend right=40] node[right,pos=0.475]{\footnotesize{0.2}} (u1);
	\draw[dashed] (w1) edge[bend right=10] node[right,pos=0.15,xshift=2pt]{\footnotesize{0.2}} (u2);	
	\draw[dashed] (w2) edge[bend right=10]  node[left,pos=0.115,xshift=-1pt]{\footnotesize{0.2}} (u1);
\end{tikzpicture}
\\
\hline
\begin{tikzpicture}[->,>=stealth,auto,black]
	\node (I) {$\mI_5:\qquad$};
	\node (bI) [node distance=0.3cm, below of=I] {};
	\node (u1) [node distance=1.0cm, right of=bI] {$u_1\!:\!a$};
	\node (u2) [node distance=3.0cm, right of=u1] {$u_2\!:\!b$};
	\node (v2) [node distance=2.0cm, below of=u1] {$v_2\!:\!A_{\,0.6}$};
	\node (v3) [node distance=2.0cm, below of=u2] {$v_3\!:\!A_{\,0.7}$};
	\node (w1) [node distance=2.0cm, below of=v2] {$w_1\!:\!B_{\,0.8}$};
	\node (w2) [node distance=2.0cm, below of=v3] {$w_2\!:\!B_{\,0.9}$};
	\draw (u1) to node [left]{\footnotesize{0.3}} (v2);	
	\draw (u2) to node [right]{\footnotesize{0.3}} (v3);	
	\draw (v2) to node [left]{\footnotesize{0.3}} (w1);	
	\draw (v3) to node [right]{\footnotesize{0.3}} (w2);
    \draw (v2) edge[out=140,in=-140,looseness=6] node[left]{\footnotesize{0.3}} (v2);
    \draw (v3) edge[out=40,in=-40,looseness=6] node[right]{\footnotesize{0.3}} (v3);
	\draw[dashed] (w1) to node[right,pos=0.15,xshift=2pt]{\footnotesize{0.2}} (u2);	
	\draw[dashed] (w2) to node[left,pos=0.15,xshift=-2pt]{\footnotesize{0.2}} (u1);
\end{tikzpicture}
&
\begin{tikzpicture}[->,>=stealth,auto,black]
	\node (I) {$\mI_4:\qquad$};
	\node (bI) [node distance=0.3cm, below of=I] {};
	\node (u1) [node distance=3.0cm, right of=bI] {$u_1\!:\!a$};
	\node (u2) [node distance=3.0cm, right of=u1] {$u_2\!:\!b$};
	\node (v1) [node distance=2.0cm, below of=bI] {$v_1\!:\!A_{\,0.5}$};
	\node (v2) [node distance=2.0cm, below of=u1] {$v_2\!:\!A_{\,0.6}$};
	\node (v3) [node distance=2.0cm, below of=u2] {$v_3\!:\!A_{\,0.7}$};
	\node (w1) [node distance=2.0cm, below of=v2] {$w_1\!:\!B_{\,0.8}$};
	\node (w2) [node distance=2.0cm, below of=v3] {$w_2\!:\!B_{\,0.9}$};
	\draw (u1) to node [left,xshift=-3pt]{\footnotesize{0.3}} (v1);	
	\draw (u1) to node [left]{\footnotesize{0.4}} (v2);	
	\draw (u2) to node [right]{\footnotesize{0.4}} (v3);	
	\draw (v1) to node [left,xshift=-3pt]{\footnotesize{0.4}} (w1);	
	\draw (v2) to node [left]{\footnotesize{0.4}} (w1);	
	\draw (v3) to node [right]{\footnotesize{0.4}} (w2);
    \draw (v1) edge[bend left=10] node[above]{\footnotesize{0.4}} (v2);
    \draw (v2) edge[bend left=10] node[below]{\footnotesize{0.4}} (v1);
    \draw (v3) edge[out=40,in=-40,looseness=6] node[right]{\footnotesize{0.4}} (v3);
	\draw[dashed] (w1) to node[right,pos=0.15,xshift=2pt]{\footnotesize{0.2}} (u2);	
	\draw[dashed] (w2) to node[left,pos=0.15,xshift=-2pt]{\footnotesize{0.2}} (u1);
\end{tikzpicture}
\\
\hline
\end{tabular}
\caption{Fuzzy interpretation mentioned in Examples~\ref{example: HGRJS 2}--\ref{example: HGRJS 4}.\label{fig: HGRJS 2}}
\end{center}
\end{figure}

\begin{example}\label{example: HGRJS 3}
Let us continue Examples~\ref{example: HGRJS} and~\ref{example: HGRJS 2}. We now reconsider the fuzzy interpretation $\mI$ specified in Example~\ref{example: HGRJS} and illustrated in Figure~\ref{fig: HGRJS}, this time for the case $\Phi = \{I\}$ (i.e., with inverse roles). Executing Algorithm~\ref{alg2} for $\mI$ using $\Phi$ and $\gamma = 1$ results in the fuzzy interpretation denoted as $\mI_3$ and illustrated in Figure~\ref{fig: HGRJS 2}. Details are provided below.

The compact fuzzy partition $\bB$ corresponding to the greatest fuzzy $\Phi$-auto-bisimulation of~$\mI$ is 
\[ \bB = 
\{
  \{\{u_1\}_1, \{u_2\}_1\}_{0.3}, 
  \{\{v_1\}_1, \{v_2\}_1, \{v_3\}_1\}_{0.3}, 
  \{\{w_1\}_1, \{w_2\}_1\}_{0.3}
\}_0.
\]
Denote $B_u = \{\{u_1\}_1, \{u_2\}_1\}_{0.3}$, $B_v = \{\{v_1\}_1, \{v_2\}_1, \{v_3\}_1\}_{0.3}$, and $B_w = \{\{w_1\}_1, \{w_2\}_1\}_{0.3}$. Thus, $\bB = \{B_u, B_v, B_w\}_0$. 
After executing statements \ref{step: alg2 1}-\ref{step: alg2 12}, we have:
\begin{itemize}
\item $\Delta^\mIp = \{u_1, u_2\}$, $a^\mIp = u_1$, $b^\mIp = u_2$;
\item for $B \in \bB.\blocks()$, 
\[ B.\repr = 
\begin{cases}
u_1 & \textrm{if } B = \{u_1\}_1 \textrm{ or } B = B_u \textrm{ or } B = \bB,\\
u_2 & \textrm{if } B = \{u_2\}_1,\\
\Null & \textrm{otherwise;}
\end{cases}
\]
\item $Q$ contains $\tuple{u_1,r,v_1}$, $\tuple{u_1,r,v_2}$, $\tuple{u_1,\cnv{s},w_2}$, $\tuple{u_2,r,v_3}$ and $\tuple{u_2,\cnv{s},w_1}$;
\item $D = \{1, 0.9, 0.8, 0.7, 0.4, 0.3, 0.2\}$. 
\end{itemize}

In the first four iterations of the ``foreach'' loop at statement~\ref{step: alg2 13}, which are executed with $d \in \{1, 0.9, 0.8, 0.7\}$, the inner ``while'' loop terminates immediately.

Consider the fifth iteration of the ``foreach'' loop at statement~\ref{step: alg2 13}, with $d = 0.4$:
\begin{itemize}
\item The first iteration of the inner ``while'' loop is executed with $\tuple{x,R,y} = \tuple{u_1, r, v_2}$, in which:
    \begin{itemize}
    \item for $B = \bB.\findBlock(v_2, 0.4) = \{v_2\}_1$, since $B.\repr = \Null$, $v_2$ is added to $\Delta^\mIp$, and $B.\repr$ and $B_v.\repr$ are set to $v_2$; 
    \item the triples $\tuple{v_2, r, w_1}$, $\tuple{v_2, r, v_1}$, $\tuple{v_2, \cnv{r}, v_1}$ and $\tuple{v_2, \cnv{r}, u_1}$ are inserted into $Q$ and $r^\mIp(u_1,v_2)$ is set to~$0.4$. 
    \end{itemize} 

\item The second iteration of the inner ``while'' loop is executed with $\tuple{x,R,y} = \tuple{v_2, \cnv{r}, v_1}$, in which:
    \begin{itemize}
    \item for $B = \bB.\findBlock(v_1, 0.4) = \{v_1\}_1$, since $B.\repr = \Null$, $v_1$ is added to $\Delta^\mIp$ and $B.\repr$ is set to $v_1$; 
    \item the triples $\tuple{v_1, r, w_1}$, $\tuple{v_1, r, v_2}$, $\tuple{v_1, \cnv{r}, v_2}$ and $\tuple{v_1, \cnv{r}, u_1}$ are inserted into $Q$ and $r^\mIp(v_1,v_2)$ is set to~$0.4$. 
    \end{itemize} 

\item The third iteration of the inner ``while'' loop is executed with $\tuple{x,R,y} = \tuple{v_1, r, v_2}$. In this iteration, for $B = \bB.\findBlock(v_2, 0.4) = \{v_2\}_1$, we have $B.\repr = v_2$ and $\mIp$ remains unchanged.

\item The fourth iteration of the inner ``while'' loop is executed with $\tuple{x,R,y} = \tuple{v_1, r, w_1}$, in which:
    \begin{itemize}
    \item for $B = \bB.\findBlock(w_1, 0.4) = \{w_1\}_1$, since $B.\repr = \Null$, $w_1$ is added to $\Delta^\mIp$, and $B.\repr$ and $B_w.\repr$ are set to $w_1$; 
    \item the triples $\tuple{w_1, s, u_2}$, $\tuple{w_1, \cnv{r}, v_1}$ and $\tuple{w_1, \cnv{r}, v_2}$ are inserted into $Q$ and $r^\mIp(v_1,w_1)$ is set to~$0.4$. 
    \end{itemize} 

\item The fifth iteration of the inner ``while'' loop is executed with $\tuple{x,R,y} = \tuple{w_1, \cnv{r}, v_1}$. In this iteration, for $B = \bB.\findBlock(v_1, 0.4) = \{v_1\}_1$, we have $B.\repr = v_1$ and $\mIp$ remains unchanged.

\item The sixth iteration of the inner ``while'' loop is executed with $\tuple{x,R,y} = \tuple{v_2, r, w_1}$. In this iteration, for $B = \bB.\findBlock(w_1, 0.4) = \{w_1\}_1$, we have $B.\repr = w_1$ and $r^\mIp(v_2,w_1)$ is set to~$0.4$.

\item The seventh iteration of the inner ``while'' loop is executed with $\tuple{x,R,y} = \tuple{v_2, r, v_1}$, in which $r^\mIp(v_2,v_1)$ is set to~$0.4$.

\item The eighth and ninth iterations of the inner ``while'' loop are executed with $\tuple{x,R,y}$ equal to $\tuple{v_1, \cnv{r}, v_2}$ and $\tuple{w_1, \cnv{r}, v_2}$, respectively, and do not make any changes.

\item The 10\textsuperscript{th} iteration of the inner ``while'' loop is executed with $\tuple{x,R,y} = \tuple{u_2, r, v_3}$, in which:
    \begin{itemize}
    \item for $B = \bB.\findBlock(v_3, 0.4) = \{v_3\}_1$, since $B.\repr = \Null$, $v_3$ is added to $\Delta^\mIp$ and $B.\repr$ is set to $v_3$; 
    \item the triples $\tuple{v_3, r, w_2}$, $\tuple{v_3, r, v_3}$, $\tuple{v_3, \cnv{r}, v_3}$ and $\tuple{v_3, \cnv{r}, u_2}$ are inserted into $Q$ and $r^\mIp(u_2,v_3)$ is set to~$0.4$. 
    \end{itemize} 
    
\item The 11\textsuperscript{th} iteration of the inner ``while'' loop is executed with $\tuple{x,R,y} = \tuple{v_3, r, w_2}$, in which:
    \begin{itemize}
    \item for $B = \bB.\findBlock(w_2, 0.4) = \{w_2\}_1$, since $B.\repr = \Null$, $w_2$ is added to $\Delta^\mIp$ and $B.\repr$ is set to $w_2$; 
    \item the triples $\tuple{w_2, s, u_1}$ and $\tuple{w_2, \cnv{r}, v_3}$ are inserted into $Q$ and $r^\mIp(v_3,w_2)$ is set to~$0.4$. 
    \end{itemize} 
    
\item The 12\textsuperscript{th} iteration of the inner ``while'' loop is executed with $\tuple{x,R,y} = \tuple{v_3, r, v_3}$. In this iteration, for $B = \bB.\findBlock(v_3, 0.4) = \{v_3\}_1$, we have $B.\repr = v_3$ and $r^\mIp(v_3,v_3)$ is set to~$0.4$.

\item The next four iterations of the inner ``while'' loop, executed with $\tuple{x,R,y}$ equal to $\tuple{v_3, \cnv{r}, v_3}$, $\tuple{w_2, \cnv{r}, v_3}$, $\tuple{v_2, \cnv{r}, u_1}$, and $\tuple{v_3, \cnv{r}, u_2}$, respectively, do not make any changes.
\end{itemize}

Consider the sixth iteration of the ``foreach'' loop at statement~\ref{step: alg2 13}, with $d = 0.3$:
\begin{itemize}
\item The first iteration of the inner ``while'' loop is executed with $\tuple{x,R,y} = \tuple{u_1,r,v_1}$, in which we have $\bB.\findBlock(v_1, 0.3) = B_v$, with $B_v.\repr = v_2$, and $\mIp$ remains unchanged.

\item The second iteration of the inner ``while'' loop is executed with $\tuple{x,R,y} = \tuple{v_1, \cnv{r}, u_1}$, in which $r^\mIp(u_1,v_1)$ is set to~$0.3$.
\end{itemize}
 
Consider the seventh iteration of the ``foreach'' loop at statement~\ref{step: alg2 13}, with $d = 0.2$: 
\begin{itemize}
\item The first iteration of the inner ``while'' loop is executed with $\tuple{x,R,y} = \tuple{u_1,\cnv{s},w_2}$, in which we have $\bB.\findBlock(w_2, 0.2) = B_w$, $B_w.\repr = w_1$, and $s^\mIp(w_1,u_1)$ is set to $0.2$.
\item The second iteration of the inner ``while'' loop is executed with $\tuple{x,R,y} = \tuple{u_2,\cnv{s},w_1}$, in which we have $B = \bB.\findBlock(w_1, 0.2) = B_w$ and $s^\mIp(w_1,u_2)$ is set to~$0.2$.
\item The third iteration of the inner ``while'' loop, executed with $\tuple{x,R,y} = \tuple{w_1, s, u_2}$, does not make any changes.
\item The fourth iteration of the inner ``while'' loop is executed with $\tuple{x,R,y} = \tuple{w_2, s, u_1}$, in which we have $\bB.\findBlock(u_1, 0.2) = B_u$, $B_u.\repr = u_1$, and $s^\mIp(w_2,u_1)$ is set to $0.2$.
\end{itemize}

After the seventh iteration, the ``foreach'' loop at statement~\ref{step: alg2 13} terminates, and the algorithm returns $\mIp$, which has the same domain size as $\mI$. 
Note also that, compared to $\mI$, $\mIp$ contains one additional non-zero role instance: $s^\mIp(w_1, s, u_1) = 0.2$.
\myend
\end{example}    

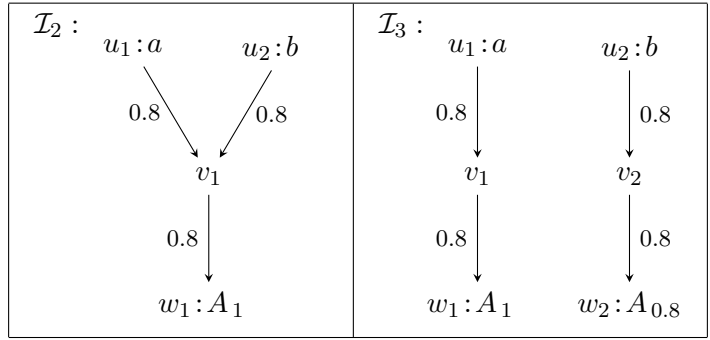
\begin{figure}[t]
\begin{center}
\begin{tabular}{|c|c|}
\hline
\begin{tikzpicture}[->,>=stealth,auto,black]
	\node (I) {$\mI_2:$};
	\node (bI) [node distance=0.3cm, below of=I] {};
	\node (u1) [node distance=1cm, right of=bI] {$u_1\!:\!a$};
	\node (u1b) [node distance=1.7cm, below of=u1] {};
	\node (v1) [node distance=1cm, right of=u1b] {$v_1$};
	\node (w1) [node distance=1.7cm, below of=v1] {$w_1\!:\!A_{\,1}\ \,$};
	\node (u2) [node distance=2cm, right of=u1] {$u_2\!:\!b\quad$};
	\draw (u1) to node [left]{\footnotesize{0.8}} (v1);	
	\draw (v1) to node [left]{\footnotesize{0.8}} (w1);	
	\draw (u2) to node [right]{\footnotesize{0.8}} (v1);	
\end{tikzpicture}
&
\begin{tikzpicture}[->,>=stealth,auto,black]
	\node (I) {$\mI_3:$};
	\node (bI) [node distance=0.3cm, below of=I] {};
	\node (u1) [node distance=1cm, right of=bI] {$u_1\!:\!a$};
	\node (v1) [node distance=1.7cm, below of=u1] {$v_1$};
	\node (w1) [node distance=1.7cm, below of=v1] {$w_1\!:\!A_{\,1}\ \,$};
	\node (u2) [node distance=2cm, right of=u1] {$u_2\!:\!b$};
	\node (v2) [node distance=2cm, right of=v1] {$v_2$};
	\node (w2) [node distance=2cm, right of=w1] {$w_2\!:\!A_{\,0.8}$};
	\draw (u1) to node [left]{\footnotesize{0.8}} (v1);	
	\draw (v1) to node [left]{\footnotesize{0.8}} (w1);	
	\draw (u2) to node [right]{\footnotesize{0.8}} (v2);	
	\draw (v2) to node [right]{\footnotesize{0.8}} (w2);	
\end{tikzpicture}
\\
\hline
\end{tabular}
\caption{Fuzzy interpretation mentioned in Example~\ref{example: YNSJA 2}.\label{fig: YNSJA 2}}
\end{center}
\end{figure}

\begin{example}\label{example: HGRJS 4}
Executing Algorithm~\ref{alg2} for the fuzzy interpretation $\mI$ specified in Example~\ref{example: HGRJS} and illustrated in Figure~\ref{fig: HGRJS}, using $\Phi = \{I,O\}$ (i.e., with inverse roles and nominals) and $\gamma = 1$ (resp.~$\gamma = 0.3$), results in the fuzzy interpretation denoted as $\mI_4$ (resp.~$\mI_5$) and illustrated in Figure~\ref{fig: HGRJS 2}. In both cases, the compact fuzzy partition corresponding to the greatest fuzzy $\Phi$-auto-bisimulation of~$\mI$ is 
\(
\{
  \{u_1\}_1, \{u_2\}_1, \{v_3\}_1, \{w_1\}_1, \{w_2\}_1, 
  \{\{v_1\}_1, \{v_2\}_1\}_{0.3}
\}_0.
\)
\myend
\end{example}

\begin{example}\label{example: YNSJA 2}
Consider the executions of Algorithm~\ref{alg2} for the fuzzy interpretation $\mI$ specified in Example~\ref{example: YNSJA} and illustrated in Figure~\ref{fig: YNSJA}, using $\gamma = 0.8$ and each of the values of $\Phi$ considered below.

In the case $\Phi = \{I\}$ (i.e., with inverse roles), the resulting fuzzy interpretation $\mIp$ is the same as in the case $\Phi = \emptyset$, and is illustrated in Figure~\ref{fig: YNSJA}. The compact fuzzy partition corresponding to the greatest fuzzy $\Phi$-auto-bisimulation of~$\mI$ is also the same as in the case $\Phi = \emptyset$ (see Example~\ref{example: YNSJA}).

In the case $\Phi = \{O\}$ (i.e., with nominals), the resulting fuzzy interpretation is denoted as $\mI_2$ and illustrated in Figure~\ref{fig: YNSJA 2}. The compact fuzzy partition corresponding to the greatest fuzzy $\Phi$-auto-bisimulation of~$\mI$ is
\(
\{
  \{u_1\}_1, \{u_2\}_1, 
  \{\{v_1\}_1, \{v_2\}_1\}_{0.8},
  \{\{w_1\}_1, \{w_2\}_1\}_{0.8}
\}_0.
\)

In the case $\Phi = \{I, O\}$ (i.e., with inverse roles and nominals), the resulting fuzzy interpretation is denoted as $\mI_3$ and is also illustrated in Figure~\ref{fig: YNSJA 2}. The compact fuzzy partition corresponding to the greatest fuzzy $\Phi$-auto-bisimulation of~$\mI$ is
\(
\{
  \{u_1\}_1, \{u_2\}_1, 
  \{v_1\}_1, \{v_2\}_1,
  \{w_1\}_1, \{w_2\}_1
\}_0.
\)
\myend
\end{example}

It is worth noting that the richer the considered FDL is, the weaker the domain reduction becomes; conversely, the lower the threshold $\gamma$ is, the stronger the domain reduction becomes.

\begin{example}
Let $\RN$, $\CN$, $\IN$ and $\mI$ be as in Example~\ref{example: JHKSM}. 
Applying Algorithm~\ref{alg2} to $\mI$ using $\Phi = \emptyset$ and $\gamma = 1$ results in a fuzzy interpretation $\mIp$, which can be obtained from $\mI$ by deleting the individuals $\simulation$ and $\minimization$, and setting $\isRelatedTo^\mIp(\bisimulation,\bisimulation) = 0.9$, $\isRelatedTo^\mIp(\fAutomata,\fAutomata) = 0.9$, and $\isRelatedTo^\mIp(\fDescLogic,\fDescLogic) = 0.8$. Further details can be found using our implementation~\cite{min2025-prog} together with the input file {\em in4.txt}.
\myend
\end{example}


\end{document}